\documentclass[pdflatex,sn-mathphys-num]{sn-jnl}

\usepackage{graphicx}
\usepackage{multirow}
\usepackage{amsmath,amssymb,amsfonts}
\usepackage{amsthm}
\usepackage{mathrsfs}
\usepackage[title]{appendix}
\usepackage{xcolor}
\usepackage{textcomp}
\usepackage{manyfoot}
\usepackage{booktabs}
\usepackage{algorithm}
\usepackage{algorithmicx}
\usepackage{algpseudocode}
\usepackage{listings}
\usepackage{cleveref}

\usepackage{textcomp}
\usepackage{mathtools}

\usepackage{caption}
\usepackage{subcaption}

\newcommand{\Z}{\mathbb{Z}}
\newcommand{\N}{\mathbb{N}}
\newcommand{\R}{\mathbb{R}}

\newcommand{\1}{\mathbb{1}}

\DeclareMathOperator{\diag}{diag}

\theoremstyle{thmstyleone}
\newtheorem{theorem}{Theorem}
\newtheorem{corollary}{Corollary}
\newtheorem{lemma}{Lemma}

\newtheorem{proposition}[theorem]{Proposition}

\theoremstyle{thmstyletwo}
\newtheorem{example}{Example}

\theoremstyle{thmstylethree}

\begin{document}

\title[Quantum Location]{Quantum Optimization in Loc(Q)ation Science: QUBO Formulations, Benchmark Problems, and a Computational Study}

\author*[1]{\fnm{Felix P.} \sur{Broesamle}}\email{felix.broesamle@kit.edu}

\author*[1]{\fnm{Stefan} \sur{Nickel}}\email{stefan.nickel@kit.edu}

\affil[1]{\orgdiv{Institute for Operations Research}, \orgname{Karlsruhe Institute of Technology}, 
\orgaddress{
\city{Karlsruhe}, \postcode{76131},
\country{Germany}}}

\abstract{Recent advances in quantum computing and the increasing availability of quantum hardware have substantially enhanced the practical relevance of quantum approaches to discrete optimization. Among these, the Quadratic Unconstrained Binary Optimization (QUBO) formulation provides a unifying modeling framework for a broad class of $\mathbf{NP}$-hard problems and is naturally suited to quantum computing and quantum-inspired algorithms. Location science, network design, and logistics represent core application domains of discrete optimization, combining high practical impact with substantial computational challenges. In this work, we develop QUBO formulations for several fundamental problems in these domains, including a nonlinear integer formulation of the Discrete Ordered Median Problem (DOMP). Beyond their modeling relevance, these QUBO formulations serve as representative benchmark problems for assessing quantum algorithms and quantum hardware. We further derive a tight bound for the penalty parameter ensuring equivalence between the QUBO formulation and its underlying integer program. Finally, we conduct a comprehensive computational study using QAOA, WS-QAOA, and classical heuristics for QUBO instances of the $p$-Median Problem and the Fixed-Charge Facility Location Problem (FCFLP), and introduce two effective warm-start strategies for WS-QAOA based on its linear programming relaxation.}

\keywords{Location Science, Quantum Computing, Quantum Algorithms, Quantum Optimization, Benchmark Problems, Discrete Optimization, Combinatorial Optimization, Quadratic Unconstrained Binary Optimization (QUBO)}

\maketitle

\section{Introduction}
\label{sec:Introduction}

Solving discrete optimization problems is a cornerstone of decision-making across numerous domains, particularly in location science, network design, and logistics. The optimization problems arising in these areas are often $\mathbf{NP}$-hard, making them computationally difficult to solve using classical exact methods; indeed the development of integer programming theory and methods has been foundational in this context, see e.g., \cite{NemhauserWolsey:IntegerCombinatorialOptimization}, \cite{Wolsey:IntegerProgramming}, \cite{Conforti2014}, \cite{DreznerHamacher:FacilityLocation}, \cite{NickelPuerto:LocationTheory} and \cite{LaporteNickelLocationScience}.
As a result, there is growing interest in exploring approaches that may offer advantages over classical deterministic approaches, e.g., heuristics and metaheuristics. One such promising paradigm is quantum computing, which leverages the principles of quantum mechanics to potentially accelerate the solution of combinatorial optimization problems. 
\newline
Many quantum optimization methods are based on the Quantum Adiabatic Theorem, see \cite{Born1928:QuantumAdiabaticTheorem} and \cite{Kato:QuantumAdiabaticTheorem}. This theorem asserts that a quantum system initialized in the ground state of a simple Hamiltonian will remain in the instantaneous ground state as the Hamiltonian is slowly evolved into a more complex one, provided the evolution is sufficiently slow and certain spectral-gap conditions hold. By encoding an optimization problem into the final Hamiltonian, the system  ideally transitions into the optimal solution state at the end of the adiabatic evolution. This principle forms the foundation of Adiabatic Quantum Computing (AQC), see e.g., \cite{farhi:AdiabaticQuantumOptimization}, and underlies the operational mechanism of contemporary quantum annealers. It also motivates the design of the Quantum Approximate Optimization Algorithm (QAOA) \cite{farhi2014QAOA}, which can be interpreted as a discretized analogue of the adiabatic evolution and is implementable on gate-based quantum computers.
\newline
A class of problems particularly well-suited for such quantum devices are those expressible as instances of the Quadratic Unconstrained Binary Optimization Problem (QUBO). The QUBO formulation enables a wide range of discrete optimization problems to be reformulated as energy-minimization tasks, thereby aligning naturally with the operational principles of both quantum annealers and QAOA. In \cite{LucasIsing}, QUBO formulations are presented for the $21$ $\mathbf{NP}$-complete problems from Karp's list \cite{Karp1972}.
\newline
However, despite rapid technological advances and increasing attention, a significant gap remains in the availability of well-structured, practically relevant benchmark problems for assessing the capabilities of quantum algorithms and hardware \cite{Abbas:ChallengesandOpportunities}. Benchmarks must not only reflect theoretical complexity but also capture the combinatorial structure observed in real-world applications.
\newline
In the domain of location science, network design and logistics, classical modeling and solution techniques are well established, see e.g., \cite{LaporteNickelLocationScience}; however, to the best of our knowledge, optimization problems arising in location science have not yet been investigated using their QUBO formulations and quantum optimization algorithms.
\newline
In this work, we aim to help to close this gap by providing QUBO formulations for a collection of fundamental and computationally challenging problems from location science, network design, and logistics. Specifically, we consider the $p$-Median Problem, the $p$-Center Problem, the Fixed-Charge Facility Location Problem (FCFLP), the Generalized Assignment Problem (GAP), the Discrete Ordered Median Problem (DOMP), and the Discrete Multi-Period Facility Location Problem. These problems are not only of high practical relevance but are also $\mathbf{NP}$-hard in general, making them well suited as benchmark problems for quantum optimization. 
\newline
The choice of the penalty parameter in a QUBO formulation, when constraints are incorporated into the objective function, is critical. To address this, in Section \ref{sec:QUBO_Formulations} we propose a penalty parameter that ensures equivalence between the QUBO and its underlying integer program.
\newline
Additionally, we conduct a comprehensive computational study in Section \ref{sec:Numerical_Experiments} to evaluate the performance of current quantum algorithms on the proposed instances. We consider both QAOA \cite{farhi2014QAOA} and WS-QAOA \cite{Egger2021warmstartingquantum}, where WS-QAOA is a variant of the standard QAOA incorporating warm-start initialization. For the warm-start, a common approach is to consider the continuous relaxation or the semidefinite programming (SDP) relaxation of the QUBO \cite{Egger2021warmstartingquantum}. However, since the continuous relaxation yields in general a nonconvex continuous optimization problem, computing a global minimum is $\mathbf{NP}$-hard; in practice, only a local minimum of the continuous relaxation can be obtained efficiently. Similarly, the SDP relaxation is not necessarily tight and may therefore fail to produce high-quality warm-start points in general. To overcome this limitation, we propose two linear programming (LP) based warm-start strategies for the FCFLP in Section \ref{sec:WarmStarting_L} (FCFLP) and \ref{sec:WarmStarting_C}. The numerical experiments show that these warm-start strategies yield high-quality starting points, outperforming both the continuous relaxation and the SDP-based warm-start strategies. Furthermore, the proposed warm-start approaches can also be applied to QUBOs derived from general integer programs. 
\newline
The results of the computational study provide insights into the current state and limitations of quantum optimization technology. Moreover, in the computational experiments we consider the aggregated and disaggregated formulation of the FCFLP, as we broadly ask how the polyhedral formulation of the original optimization problem impacts the solutions obtained by quantum optimization algorithms when using QUBO formulations.
\newline
\newline
Quantum optimization algorithms, such as QAOA, have attracted significant attention, but are currently limited by Noisy Intermediate-Scale Quantum (NISQ) devices \cite{PreskillNISQ, Pellow_Jarman_2024_Noise, Quek2024_Noise}, which suffer from a limited number of qubits, restricted connectivity, and significant noise, often dominating algorithmic performance. To isolate algorithmic behavior, we perform all quantum experiments using the Qiskit \texttt{AerSimulator} \cite{qiskit2024, IBMQuantum} in a noise-free configuration, which computes the exact evolution of quantum states under prescribed gate operations.
\newline
\newline
\textbf{Related works:} QUBO formulations have been widely studied as a unifying modeling framework for combinatorial optimization. A general introduction and systematic methodology for deriving QUBOs from integer programs, along with their application to various combinatorial problems, is provided in \cite{QUBO:Kochenberger_Tutorial}. An overview of industrial applications of quantum computing and quantum optimization is presented in \cite{QuantumApplications}. Several studies address financial applications, such as portfolio optimization \cite{QUBO:Egger_Finance, QUBO:Brandhofer_Portfolio, QUBO:Hodson_portfolio}. Reference \cite{QUBO:Bochkarev} reviews quantum technologies for discrete optimization and QUBO formulations for combinatorial problems. Further examples include QUBO formulations for the Job Shop Scheduling Problem \cite{QUBO:Scheduling}, the binary knapsack problem \cite{QUBO:BinaryKnapsack_Dwave}, and routing problems \cite{QUBO:Cattelan_Routing, QUBO:Routing_2, QUBO:Routing2021}. A comprehensive catalogue of QUBO and Ising formulations for classical $\mathbf{NP}$-hard problems is presented in \cite{LucasIsing}, specifically the $21$ $\mathbf{NP}$-hard problems listed by Karp \cite{Karp1972}.
\newline
\newline
\textbf{Our contribution:} In this work, we introduce QUBO formulations for classical optimization problems arising in location science, network design, and logistics, including the $p$-Median Problem, the $p$-Center Problem, the Fixed-Charge-Facility Location Problem (FCFLP), the Discrete Ordered Median Problem (DOMP), and the Discrete Multi-Period Facility Location Problem. In doing so, we extend the benchmark problems in the field of quantum optimization to include practically highly relevant problems, which are also $\mathbf{NP}$-hard and challenging to solve. Furthermore, we propose a penalty parameter selection in Theorem \ref{theorem:PenArbitrary}, ensuring that a global optimum of the QUBO formulation corresponds to a global optimum of the underlying integer optimization problem. In the computational study, we investigate the newly introduced QUBO formulations for the $p$-Median Problem and the FCFLP applying QAOA, WS-QAOA, and classical heuristics. Additionally, we introduce two promising LP-based warm-start strategies for the FCFLP for WS-QAOA, which are also applicable for QUBOs arising from general integer programs. Finally, by presenting the aggregated and disaggregated formulation of the FCFLP together with the corresponding QUBO formulations, as well as the numerical evaluation of both formulations, we aim to draw attention to the need for further research on how the polyhedral properties of integer programs influence the behavior of quantum algorithms.
\newline
\newline
The structure of this paper is as follows: in Section \ref{sec:Preliminaries}, we introduce the fundamental concepts of quantum optimization, including the quantum optimization algorithms QAOA and WS-QAOA, as well as the basic principles of location science. In Section \ref{sec:QUBO_Formulations}, we present QUBO formulations for the aforementioned optimization problems. Finally, in Section \ref{sec:Numerical_Experiments}, we provide comprehensive numerical experiments.


\section{Preliminaries}
\label{sec:Preliminaries}
In this section, we introduce the fundamental concepts of quantum adiabatic optimization, Ising Hamiltonians, and their equivalent QUBO formulations, which form the theoretical foundation for quantum optimization methods. We further present the gate-based quantum algorithms QAOA and WS-QAOA, which are employed in the numerical experiments in Section \ref{sec:Numerical_Experiments}. Finally, we provide a brief overview of location science.
\subsection{Quantum Adiabatic Optimization, Ising Hamiltonian, and QUBOs}
\label{sec:Preliminaries_IsingHamiltonian_QUBO}
Quantum Adiabatic Optimization leverages the adiabatic theorem of quantum mechanics \cite{Born1928:QuantumAdiabaticTheorem, Kato:QuantumAdiabaticTheorem} to solve combinatorial optimization problems \cite{farhi:AdiabaticQuantumOptimization}. The process initializes a quantum system in the ground state of a simple Hamiltonian $H_0$, then gradually and slowly evolves it to the problem (cost) Hamiltonian $H_P$, whose ground state encodes the optimal solution \cite{LucasIsing}. The time-dependent Hamiltonian is typically defined as 
\begin{align*}
	H(t) = \left(1 - \frac{t}{T} \right) H_0 + \frac{t}{T}H_P,\quad \quad  t \in [0,T],
\end{align*}
where $T$ denotes the total evolution time. Provided the evolution is slow enough and spectral gaps do not close, the state of the system remains in the instantaneous ground state throughout the process. Measuring the system at $t = T$ yields an optimal solution with high probability, or even with probability $1$, if $T$ is large enough.
\newline
\newline
$H_P$ is chosen as the quantum analogue of a classical \textit{Ising Hamiltonian},
\begin{align*}
	H_P(z_1,...,z_n) = - \sum_{i <j} J_{ij}z_iz_j - \sum_{i}h_iz_i,
\end{align*}
with \textit{spin variable} $z_i \in \{-1,+1\}$, and real-valued pairwise couplings $J_{ij}$ and local fields $h_i$. In the quantum setting, the spin variables $z_i$ are replaced by Pauli-$Z$ operators acting on the $i$-th qubit, respectively. The initial Hamiltonian $H_0 = -h_0\sum_i X_i$ induces a transverse field, producing a uniform superposition over all computational basis states, where $X_i$ denotes the Pauli-$X$ operator acting on the $i$-th qubit.
\newline
\newline
Minimizing the classical \textit{Ising energy} is equivalent to solving a \textit{Quadratic Unconstrained Binary Optimization (QUBO)} problem via the substitution $z_i = 2x_i - 1$ with $x_i \in \{0,1\}$. This yields
\begin{align}
	\min_{x \in \{0,1\}^n} \quad x^TQx + q^Tx + c,  \label{QUBO}
\end{align}
with suitable $Q \in \R^{n \times n}$, $q \in \R^n$, and $c \in \R$, providing a discrete optimization problem naturally compatible with quantum optimization algorithms based on the quantum adiabatic theorem.
\subsection{Quantum Approximate Optimization Algorithm (QAOA)}
\label{sec:Preliminaires_QAOA}
The \textit{Quantum Approximate Optimization Algorithm (QAOA)}, originally proposed by \cite{farhi2014QAOA}, is a \textit{quantum-classical hybrid algorithm} for solving combinatorial optimization problems. Conceptually, QAOA is inspired by quantum adiabatic optimization, as it discretizes the continuous adiabatic evolution into a sequence of alternating unitaries. In contrast to quantum annealing, however, QAOA is specifically designed for gate-based, universal quantum computers.
\newline
QAOA depends on two sets of \textit{real-valued parameters},
\begin{align*}
	\gamma = (\gamma_1,...,\gamma_p), \quad \text{ and } \quad \beta = (\beta_1,...,\beta_p),
\end{align*}
where the \textit{depth parameter (number of layers)} $p \in \N$ controls both the expressive power of the variational ansatz and the number of alternating operator layers. The algorithm starts from the uniform superposition state $\lvert + \rangle^{\otimes n}$ and alternates between applying the \textit{cost (problem) Hamiltonian} $H_C$ and a \textit{mixing Hamiltonian} $H_M$, yielding the variational quantum state
\begin{align*}
	\lvert \psi(\beta, \gamma) \rangle = e^{-i \beta_pH_M}e^{-i\gamma_p H_C}\cdots e^{-i\beta_1H_M}e^{-i\gamma_1H_C}\lvert + \rangle^{\otimes n}.
\end{align*}
The cost (problem) Hamiltonian $H_C$ encodes the objective function of the underlying QUBO problem
\begin{align*}
	H_C \lvert x \rangle = C(x) \lvert x \rangle, \quad \quad \forall x \in \{0,1\}^n,
\end{align*}
where $C(x)$ denotes the \textit{QUBO cost function}. The mixing Hamiltonian $H_M$ is defined as
\begin{align*}
	H_M = \sum_{i = 1}^n X_i, 
\end{align*}
where $X_i$ denotes the Pauli-$X$ operator acting on the $i$-th qubit. 
\newline
For a fixed parameter pair $(\beta, \gamma)$, the quantum circuit implementing the above unitary generates a distribution over computational basis states, i.e., over points in $\{0,1\}^n$. Sampling from this distribution yields candidate solutions to the QUBO. The parameters are optimized classically to minimize the expected value
\begin{align*}
	\langle \psi(\beta, \gamma) \rvert H_C \lvert \psi(\beta, \gamma) \rangle,
\end{align*}
which serves as the variational estimate of the problem's objective function.
\begin{proposition}[\cite{QUBO:Bochkarev}, Remark 1]
	\label{prop:Physical_Qubits}
	An arbitrary QUBO instance with $N$ variables that is solved with QAOA requires $N$ physical qubits. While transpilation might affect the ouput quality, it does not necessarily require additional physical qubits.
\end{proposition}
\subsection{Warm-Start Quantum Approximate Optimization Algorithm (WS-QAOA)}
\label{sec:Preliminaries_WS_QAOA}
Warm-Start QAOA (WS-QAOA), introduced by \cite{Egger2021warmstartingquantum}, is a variant of QAOA that incorporates problem-specific information into both the \textit{initial state} and the \textit{mixing Hamiltonian}. In contrast to the standard QAOA, which starts from the uniform superposition $\lvert + \rangle^{\otimes n}$, WS-QAOA initializes the quantum circuit using a \textit{warm-start point} $x^0 \in \R^n$.
\newline
The WS-QAOA circuit is initialized in the state
\begin{align*}
	\lvert \phi^0 \rangle = \bigotimes_{i = 0}^{n-1} \hat{R}_{Y}(\theta_i) \lvert 0 \rangle
\end{align*}
where $\hat{R}_Y(\theta_i)$ denotes a rotation around the $Y$-axis of qubit $i$ by the angle $\theta_i = 2 \arcsin \left( \sqrt{x_i^0}  \right)$ and $x_i^0 \in [0,1]$ is the $i$-th component of the warm-starting point $x^0$. 
The \textit{warm-start mixer Hamiltonian} is defined as 
\begin{align*}
	\hat{H}_M^{(ws)} =\sum_{i = 0}^{n-1} \hat{H}_{M,i}^{(ws)},
\end{align*} with individual components
\begin{align*}
	\hat{H}_{M, i}^{(ws)} &= \begin{pmatrix}
		2c_i^* - 1 & -2 \sqrt{c_i^*(1-c_i^*)} \\
		-2 \sqrt{c_i^* (1 - c_i^*)} & 1 - 2c_i^*
	\end{pmatrix} \\
	& = -\sin(\theta_i)\hat{X} - \cos(\theta_i) \hat{Z},
\end{align*}
for $i = 0,...,n-1$. $\hat{H}_{M, i}^{(ws)}$ can be implemented using the single-qubit rotations $\hat{R}_Y(\theta_i)\hat{R}_Z(-2\beta)\hat{R}_Y(-\theta_i)$.
\newline
While WS-QAOA can improve convergence and solution quality by biasing the search toward promising regions of the solution space, it may encounter \textit{reachability issues} when the warm-start point $x^0$ has integral components, which differ from the global optimal solution $x^*$ of \eqref{QUBO}, i.e., when $x_i^0 = 0$ and $x_i^* = 1$, or vice versa.
\newline
A practical strategy to mitigate this issue is to construct a modified warm-start point $\bar{x}^0$ by projecting $x^0$ onto the $\ell_{\infty}$-norm ball centered at $(\frac{1}{2},...,\frac{1}{2}) \in \R^n$ with radius $r = 0.5 - \varepsilon$ for $\varepsilon \in (0, 0.5]$. This projection ensures that all components of $\bar{x}^0$ satisfy $\bar{x}^0_i \in [\varepsilon, 1-\varepsilon]$ for all $i = 1,...,n$.

\subsection{Location Science}
\label{sec:Preliminaries_Location_Science}
Location science is a well-established research field concerned with determining the optimal placement of facilities to serve a given set of demand points. Since its origins in the classical works \cite{vonThunen1842} and \cite{Weber1909}, the field has evolved into a core area of operations research, combining elements of mathematics, economics, geography and computer science \cite{LaporteNickelLocationScience, DreznerHamacher:FacilityLocation, NickelPuerto:LocationTheory, Eiselt2011}. Over the past decades, location science has developed a rich theoretical foundation and a wide range of modeling and solution approaches, motivated by both methodological advances and practical applications in diverse domains \cite{LocationScience:Introduction}. 
\newline
Discrete location problems, in particular, have received extensive attention due to their practical relevance in logistics, transportation, network design and public service planning. These problems typically involve selecting a subset of candidate sites for facility placement while minimizing or maximizing a given objective function that represents costs, distances, or service quality \cite{LocationScience:Introduction}.
\newline
In the discrete setting, a finite set of candidate sites is available for facility placement, and a finite set of clients requires service from these facilities. Let $I = \{1,...,n\}$ denote the set of potential sites, which, without loss of generality, $I$ is assumed to coincide with the set of clients $J$. The objective is typically to decide which sites to open and how to assign clients to open facilities according to a specified cost or distance metric.


\section{Optimization Problems arising in Location Science as QUBOs}
\label{sec:QUBO_Formulations}
In this section, we present QUBO formulations for the $p$-Median Problem (Section \ref{sec:QUBO_p_median}), the $p$-Center Problem (Section \ref{sec:QUBO_p_center}), the Fixed-Charge Facility Location Problem (FCFLP, Section \ref{sec:QUBO_FCFLP}), the Generalized Assignment Problem (GAP, Section \ref{sec:QUBO_GAP}), the Discrete Ordered Median Problem (DOMP, Section \ref{sec:QUBO_DOMP}), and the Discrete Multi-Period Facility Location Problem (Section \ref{sec:QUBO_DMPFLP}). We consider the general setting described above in Section \ref{sec:Preliminaries_Location_Science}. Before deriving the QUBO formulations, Section \ref{sec:QUBO_Penalty} establishes a result that yields a tight penalty parameter ensuring equivalence between an integer program and its QUBO formulation.
\newline
\newline
We follow standard practice in the literature, see, e.g., \cite{LucasIsing}, and present each QUBO formulation \eqref{QUBO} in terms of the corresponding cost Hamiltonian written in binary variables. The corresponding upper triangular (or symmetric) QUBO matrix can then be obtained by computing the linear and quadratic coefficients of the Hamiltonian.
\subsection{Formulating QUBOs and Choice of Penalty Parameter}
\label{sec:QUBO_Penalty}
We consider general integer programs (IP) of the form
\begin{subequations}\label{IP}
	\begin{align}
		\min_x \quad & \tilde{c}^Tx ,\\
		\mbox{s.t.} \quad  & \tilde{A}x = \tilde{b}, \\
		&Bx \leq d, \\
		&l \leq x \leq u, \\
		& x \in \Z^{\tilde{n}},
	\end{align}
\end{subequations}
where $\tilde{c} \in \R^{\tilde{n}}, \tilde{A} \in \Z^{m_1 \times \tilde{n}}, B \in \Z^{m_2 \times \tilde{n}}, \tilde{b} \in \Z^{m_1}, d \in \Z^{m_2},$ and $l, u \in \Z^{\tilde{n}}$.
\newline
A common first step in constructing a QUBO formulation \eqref{QUBO} is to reformulate \eqref{IP} into an equivalent binary program of the form
\begin{subequations}  \label{BP}
	\begin{align}
		\min_{x} \quad & c^Tx  \\
		\mbox{s.t.} \quad & Ax = b, \\
		& g_i(x) \leq 0, & \forall i = 1,...,r, \\
		& x \in \{0,1\}^n,
	\end{align}
\end{subequations}
where $c \in \R^n, A \in \Z^{m \times n},b \in \Z^m$, and $g_i: \{0,1\}^n \rightarrow \Z$, are simple integer-valued constraint functions, for which direct penalty-terms exist, see Table \ref{tab:penalty_terms}.
\newline
Inequality constraints in \eqref{IP} are first transformed into equality constraints by introducing integer slack variables or kept as simple constraints $g_i$. Subsequently, each integer variable is represented using a binary encoding. 
\newline
Now, to obtain a QUBO formulation, the constraints of \eqref{BP} are removed and their violation in the objective function penalized. The general linear equality constraints $Ax = b$ are incorporated into the objective via row-wise quadratic penality terms, i.e,
\begin{align*}
	(b_i - A_ix)^2, \quad \quad i = 1,...,m,
\end{align*} 
where $A_i$ denotes the $i$-th row of $A$ and any violation of constraint $i$ yields at least a value of $1$. 
For each constraint $g_i(x) \leq 0$, we associate a penalty function $p_i(x)$, defined such that any violation produces an integer-valued penalty of at least one.
See Table \ref{tab:penalty_terms} for the penalty-terms.
\begin{table}[ht]
	\centering
	\caption{Table of simple constraints and corresponding penalty terms}
	\label{tab:penalty_terms}
	\begin{tabular}{@{}l r@{}}
		\toprule
		Constraint & Penalty term \\
		\midrule
		$x_i + x_j \leq 1$ & $P \cdot x_ix_j$ \\
		$x_i \leq x_j$ & $P \cdot x_i(1 - x_j) $ \\
		$x_1 + \cdots + x_n \leq 1$ & $P \cdot \sum_{i = 1}^{n-1}\sum_{j = i+1}^nx_ix_j$ \\
		\bottomrule
	\end{tabular}
\end{table}
\newline
We now present the penalized quadratic unconstrained binary optimization (QUBO) formulation of \eqref{BP}, expressed here as an alternative to the classical matrix-based representation \eqref{QUBO}:
\begin{align}
	\min_{x \in \{0,1\}^n} f(x) = c^Tx + P\cdot \sum_{i=1}^{r}p_i(x)   + P \cdot \sum_{i = 1}^m(b_i - A_ix)^2, \label{QUBO-PEN}
\end{align}
where $P > 0$ is a penalty parameter.
\newline
\newline
The choice of the penalty parameter $P$ is critical for both theoretical and practical performance, analogously to classical penalty methods in nonlinear programming, see, e.g., \cite{Bazaraa}. $P$ must be sufficiently large to ensure that any global minimizer of \eqref{QUBO-PEN} is feasible for \eqref{BP}, yet not so large as to induce numerical instability or an ill-conditioned energy landscape.
\newline
\newline
We next provide tight conditions on the penalty parameter that guarantee equivalence between the QUBO formulation \eqref{QUBO-PEN} and the binary program \eqref{BP}, and therefore the integer program \eqref{IP}.
\begin{theorem}
	\label{theorem:PenArbitrary}
	Consider the optimization problem  \eqref{IP}, its equivalent formulation \eqref{BP}, and the associated QUBO formulation \eqref{QUBO-PEN}, and assume that \eqref{IP} is feasible.
	Let $c \in \R^n$ denote the cost vector of \eqref{BP}. If the penalty parameter in \eqref{QUBO-PEN} satisfies $P > \sum_{i=1}^n\lvert c_i \rvert$, then every global minimizer of the optimization problem \eqref{QUBO-PEN} is feasible for \eqref{IP} and is also a global minimizer of \eqref{IP}.
\end{theorem}
\begin{proof}
	Let $P > \sum_{i=1}^n\lvert c_i \rvert$. We show the statement for \eqref{BP}; the result for \eqref{IP} then follows immediately from their equivalence.
	Let $\emptyset \neq X =\{x \in \{0,1\}^n: Ax = b, g_i(x)\leq 0,i=1,...,r\}$ denote the feasible set of \eqref{BP}. We set $C^- = \sum_{i=1:c_i<0}^n \lvert c_i\rvert$ and $C^+ = \sum_{i=1:c_i> 0}^n c_i$. Then, for any $x \in \{0,1\}^n$, we have $c^Tx \in \{-C^- ,...,C^+\}$.
	\newline
	Let $y^*$ be a global minimizer of \eqref{QUBO-PEN}. We show by contradiction that $y^*$ is feasible for \eqref{BP}. Therefore, suppose that $y^*$ is infeasible for \eqref{BP}. Then, it holds 
	\begin{align}
		\sum_{i=1}^r p_i(y^*) + \sum_{i = 1}^m(b_i - A_iy^*)^2 \geq 1, \label{PenActive}
	\end{align} 
	and we obtain
	\begin{align*}
		f(y^*) &= c^Ty^* + P\cdot \sum_{i=1}^r p_i(y^*) + P \cdot \sum_{i = 1}^m(b_i - A_iy^*)^2 \\ & \stackrel{\eqref{PenActive}}{\geq} c^Ty^*  + P \stackrel{c^Ty^* \in \{-C^- ,...,C^+\}}{\geq}-C^- 
		+ P \\ &> c^Tx, \quad\quad   \forall x \in X.
	\end{align*}
	This is a contradiction to the global optimality of $y^*$. Now, as $y^*$ is feasible, we conclude by the global optimality of $y^*$:
	\begin{align*}
		f(y^*) = c^Ty^* + P \cdot 0 = c^Ty^* = \min_{x \in X}c^Tx.
	\end{align*}
\end{proof}
\begin{corollary}
	\label{theorem:PenPos}
	Consider the optimization problem \eqref{IP}, its equivalent formulation \eqref{BP}, and the associated QUBO formulation \eqref{QUBO-PEN}, and assume that \eqref{IP} is feasible.
	Let $c \in \R^n$, with $c \geq 0$, denote the cost vector of \eqref{BP}. If the penalty parameter in \eqref{QUBO-PEN} satisfies $P > \sum_{i = 1}^nc_i$, then every global minimizer of the optimization problem \eqref{QUBO-PEN} is feasible for \eqref{IP} and is also a global minimizer of \eqref{IP}.
\end{corollary}
\begin{proof}
	This follows directly from Theorem \ref{theorem:PenArbitrary}.
\end{proof}
For certain structured instances of \eqref{IP} and \eqref{BP}, smaller penalty parameters may suffice, particularly if additional problem-specific information, preprocessing, or relaxation-based estimates are available. 
We note that, for several combinatorial problems, \cite{QUBO:Bochkarev} provides a problem-specific selection of a tight penalty parameter. A similar proof strategy was employed in \cite{QUBO:Routing2021} to verify that the penalty parameter from Theorem \ref{theorem:PenArbitrary} is valid for a vehicle routing problem with time windows.
\newline
Henceforth, we assume that all considered optimization problems are feasible.

\subsection{p-Median Problem}
\label{sec:QUBO_p_median}
The $p$-Median Problem is a classical optimization problem in location science, first formalized in \cite{HakimiPMedianProblem1}. The objective is to select $p$ facility locations, called \textit{medians}, from a set of candidate sites to minimize the total distance or cost between demand points and their assigned facilities, see also \cite{LocationScience:Chapter2pMedian}. Let $d_j > 0, j \in J$ denote the demand of each client, and $c_{ij} \geq 0$ the cost of serving client $j \in J$ from facility $i \in I$. From now on, we assume $I = J = \{1,...,n\}$, so that there are $n$ potential facility locations and $p < n$.
\newline
To formulate the $p$-Median Problem, we introduce the following binary variables:
\begin{itemize}
	\itemsep0em
	\item $y_i = 1$, if a facility is placed at site $i \in I$ and 0 otherwise. 
	\item $x_{ij} = 1$, if client $j  \in J$ is assigned to a facility located at $i \in I$ and $0$ otherwise.
\end{itemize}
The $p$-Median Problem is formulated as the following optimization problem:
\begin{subequations}\label{p-Median-problem}
	\begin{align}
		\displaystyle\min_{x, y} \quad & \sum_{i \in I}\sum_{j \in J} d_j c_{ij}x_{ij} \\
		\mbox{s.t.} \quad & \sum_{i \in I} x_{ij} = 1, & \forall j \in J, \\
		& x_{ij} \leq y_i,   & \forall  i \in I, j \in J, \\
		& \sum_{i \in I} y_i = p,&  \\
		& x_{ij} \in \{0,1\}, & \forall  i \in I, j \in J,\\
		& y_i \in \{0,1\}, & \forall  i \in I.
	\end{align}
\end{subequations}
Here, the first condition ensures that every client is assigned to exactly one facility. The second constraint links the allocation and facility variables, ensuring that a client can only be assigned to an open facility. The third constraint ensures the choice of exactly $p$ facilities. 
\begin{example}
	\label{exp:p_median_4_2_data}
	We consider the $p$-Median Problem with $n = 4$ potential facility locations. The input data are given by the assignment cost matrix $C \in \R^{4 \times 4}$ and the demand vector $d \in \R^4$:
	\begin{align*}
		C = \begin{pmatrix}
			2 & 11 & 13 & 6 \\
			14 & 0 & 15 & 11 \\
			5 & 14 & 1 & 6 \\
			5 & 12 & 15 & 2
		\end{pmatrix}, \quad \quad d = \begin{pmatrix}
			4 \\ 4 \\ 13 \\ 11
		\end{pmatrix}
	\end{align*}
	For $p = 2$, the optimal solution opens the facilities $2$ and $3$, i.e., $y_2 = y_3 = 1$, while $y_1 = y_4 = 0$. The corresponding client-facility assignments are $x_{22} = x_{31} = x_{33} = x_{34} = 1$, with all remaining $x_{ij} = 0$. The resulting optimal objective value is $99$.
\end{example}
\textbf{QUBO formulation}: The cost Hamiltonian for \eqref{p-Median-problem} is given by
\begin{subequations} \label{p-Median-QUBO}
	\begin{align}
		H  = & \sum_{i \in I}\sum_{j \in J} d_j c_{ij}x_{ij} \\
		& + P\sum_{j \in J}\left(1 - \sum_{i \in I}x_{ij}\right)^2 \\
		& + P \sum_{i \in I} \sum_{j \in J} x_{ij}(1 - y_i) \\
		&+ P \left(p - \sum_{i \in I}y_i \right)^2
	\end{align}
\end{subequations}
The $p$-Median Problem has the appealing property that no auxiliary variables or constraints are required to express it as a QUBO \eqref{QUBO}.
\begin{lemma}
	\label{lemma:p-median_qubits}
	The $p$-Median Problem \eqref{p-Median-problem} can be formulated as a QUBO with $n^{p\text{-Median}} = \lvert I \rvert \lvert J \rvert + \lvert I \rvert$ binary variables. Therefore, the QUBO formulation \eqref{p-Median-QUBO} can be solved on a gate-based Quantum Computer with $n^{p\text{-Median}}$ physical qubits via QAOA or WS-QAOA.
\end{lemma}
\begin{proof}
	In the QUBO formulation \eqref{p-Median-QUBO}, only the binary variables $x_{ij}$ for $i \in I, j \in J$, and $y_i$ for $i \in I$, are used, giving a total of $n^{p\text{-Median}} = \lvert I \rvert \lvert J \rvert + \lvert I \rvert$ binary variables. The second part of the statement follows directly from Proposition \ref{prop:Physical_Qubits}.
\end{proof}
\begin{corollary}
	\label{korollar:p_median_penalty}
	Let the penalty parameter in \eqref{p-Median-QUBO} be $P = \sum_{j \in J} \sum_{i \in I} d_jc_{ij} + 1$. Then, any global minimizer of \eqref{p-Median-QUBO} yields a global minimizer of \eqref{p-Median-problem}.
\end{corollary}
\begin{proof}
	The $p$-Median Problem \eqref{p-Median-problem} is already formulated in the form \eqref{BP}, with objective function $\sum_{j \in J} \sum_{i \in I} d_jc_{ij}x_{ij}$. Consequently, the statement follows from Corollary \ref{theorem:PenPos}.
\end{proof}
\begin{example}
	We consider Example \ref{exp:p_median_4_2_data}. The penalty parameter is given by $P = 1100$. The upper-triangular QUBO matrix is given as follows:
	\begin{align*}
		Q = 
		\resizebox{0.8\linewidth}{!}{$\left[\begin{array}{*{20}c} 
				8 & 0 & 0 & 0 & 2200 & 0 & 0 & 0 & 2200 & 0 & 0 & 0 & 2200 & 0 & 0 & 0 & -1100 & 0 & 0 & 0 \\
				0 & 44 & 0 & 0 & 0 & 2200 & 0 & 0 & 0 & 2200 & 0 & 0 & 0 & 2200 & 0 & 0 & -1100 & 0 & 0 & 0 \\
				0 & 0 & 169 & 0 & 0 & 0 & 2200 & 0 & 0 & 0 & 2200 & 0 & 0 & 0 & 2200 & 0 & -1100 & 0 & 0 & 0 \\
				0 & 0 & 0 & 66 & 0 & 0 & 0 & 2200 & 0 & 0 & 0 & 2200 & 0 & 0 & 0 & 2200 & -1100 & 0 & 0 & 0 \\
				0 & 0 & 0 & 0 & 56 & 0 & 0 & 0 & 2200 & 0 & 0 & 0 & 2200 & 0 & 0 & 0 & 0 & -1100 & 0 & 0 \\
				0 & 0 & 0 & 0 & 0 & 0 & 0 & 0 & 0 & 2200 & 0 & 0 & 0 & 2200 & 0 & 0 & 0 & -1100 & 0 & 0 \\
				0 & 0 & 0 & 0 & 0 & 0 & 195 & 0 & 0 & 0 & 2200 & 0 & 0 & 0 & 2200 & 0 & 0 & -1100 & 0 & 0 \\
				0 & 0 & 0 & 0 & 0 & 0 & 0 & 121 & 0 & 0 & 0 & 2200 & 0 & 0 & 0 & 2200 & 0 & -1100 & 0 & 0 \\
				0 & 0 & 0 & 0 & 0 & 0 & 0 & 0 & 20 & 0 & 0 & 0 & 2200 & 0 & 0 & 0 & 0 & 0 & -1100 & 0 \\
				0 & 0 & 0 & 0 & 0 & 0 & 0 & 0 & 0 & 56 & 0 & 0 & 0 & 2200 & 0 & 0 & 0 & 0 & -1100 & 0 \\
				0 & 0 & 0 & 0 & 0 & 0 & 0 & 0 & 0 & 0 & 13 & 0 & 0 & 0 & 2200 & 0 & 0 & 0 & -1100 & 0 \\
				0 & 0 & 0 & 0 & 0 & 0 & 0 & 0 & 0 & 0 & 0 & 66 & 0 & 0 & 0 & 2200 & 0 & 0 & -1100 & 0 \\
				0 & 0 & 0 & 0 & 0 & 0 & 0 & 0 & 0 & 0 & 0 & 0 & 20 & 0 & 0 & 0 & 0 & 0 & 0 & -1100 \\
				0 & 0 & 0 & 0 & 0 & 0 & 0 & 0 & 0 & 0 & 0 & 0 & 0 & 48 & 0 & 0 & 0 & 0 & 0 & -1100 \\
				0 & 0 & 0 & 0 & 0 & 0 & 0 & 0 & 0 & 0 & 0 & 0 & 0 & 0 & 195 & 0 & 0 & 0 & 0 & -1100 \\
				0 & 0 & 0 & 0 & 0 & 0 & 0 & 0 & 0 & 0 & 0 & 0 & 0 & 0 & 0 & 22 & 0 & 0 & 0 & -1100 \\
				0 & 0 & 0 & 0 & 0 & 0 & 0 & 0 & 0 & 0 & 0 & 0 & 0 & 0 & 0 & 0 & -3300 & 2200 & 2200 & 2200 \\
				0 & 0 & 0 & 0 & 0 & 0 & 0 & 0 & 0 & 0 & 0 & 0 & 0 & 0 & 0 & 0 & 0 & -3300 & 2200 & 2200 \\
				0 & 0 & 0 & 0 & 0 & 0 & 0 & 0 & 0 & 0 & 0 & 0 & 0 & 0 & 0 & 0 & 0 & 0 & -3300 & 2200 \\
				0 & 0 & 0 & 0 & 0 & 0 & 0 & 0 & 0 & 0 & 0 & 0 & 0 & 0 & 0 & 0 & 0 & 0 & 0 & -3300
			\end{array}\right] $}
	\end{align*}
	The offset for this QUBO problem is $c = 8800$.
\end{example}
\subsection{p-Center Problem}
\label{sec:QUBO_p_center}
Minimizing the total or average distance between clients and facilities, as in the $p$-Median Problem, may not be an appropriate criterion for all types of facility location decisions. Such measures tend to favor clients clustered in population centers, often at the expense of those who are spatially dispersed. In contexts where equitable service is desired, a criterion focusing on the worst-served clients is more appropriate, leading to the $p$-Center Problem \cite{HakimiPMedianProblem1, HakimiPMedianProblem2, LocationScience:Chapter3pCenter}.
\newline
Let $d_{ij} \geq 0, i \in I, j \in J$, denote the distance (or cost) from client $j \in J$ to facility $i \in I$. We assume, without loss of generality, that $d_{ij} \in \N_0$ for all $i \in I, j \in J$.
\newline
For the $p$-Center Problem, we introduce the binary variables:
\begin{itemize}
	\itemsep0em
	\item $y_i = 1$, if a facility is placed at site $i \in I$ and 0 otherwise.
	\item $x_{ij} = 1$, if client $j \in J$ is assigned to a facility located at $i \in I$, and $0$ otherwise.
\end{itemize}
The $p$-Center Problem can be formulated as
\begin{subequations}\label{p-Center-Problem}
	\begin{align}
		\displaystyle\min_{x,y,z} \quad & z  \\
		\mbox{s.t.} \quad & \sum_{i \in I} x_{ij} = 1, & \forall  j \in J, \\
		& \sum_{i \in I} d_{ij}x_{ij} \leq z, & \forall j \in J, \\
		& x_{ij} \leq y_i, & \forall  i \in I, j \in J, \\
		& \sum_{i \in I} y_i \leq p, \\
		& x_{ij} \in \{0,1\}, & \forall i \in I, j \in J, \\
		& y_i \in \{0,1\} , & \forall  i \in I, \\
		& z \in \Z.
	\end{align}
\end{subequations}
The objective function, together with the second constraint, ensures that the objective value $z$ is greater than or equal to the maximum distance between any clients and facilities they are assigned to. The first constraint guarantees that each client is assigned to exactly one facility, while the third constraint links the allocation and facility variables, ensuring that a client can only be assigned to an open facility. The fourth constraint limits the number of facilities that can be opened to at most $p$.
\newline
\newline
\textbf{QUBO formulation}: Let $d_{max} = \max_{i, j} d_{ij}$ and $d_{min} = \min_{i, j} d_{ij}$. To transform the inequality constraint $\sum_{i\in I} d_{ij}x_{ij} \leq z$ for each $j \in J$, into an equality constraint, we introduce an integral slack variable $s_j \in \{0,...,d_{max}\}$ for each $j \in J$, yielding $\sum_{i\in I} d_{ij}x_{ij} + s_j = z$. Next, we replace the integral variable $z$ by binary variables. We introduce binary variables $z_k \in \{0,1\}, k \in \{1,...,\lceil \log_2( d_{max} - d_{min}  + 1) \rceil \}$ and express $z$ as 
\begin{align*}
	z = d_{min} + \sum_{k = 1}^{\lceil \log_2( d_{max} -  d_{min} + 1) \rceil}2^{k-1}z_k,
\end{align*}
since $z \geq d_{min}$ for feasible points. Each $s_j$ is represented using binary variables $s_{jk} \in \{0,1\}, k \in \{1,...,\lceil \log_2(d_{max} + 1) \rceil\}$ for each $j \in J$, via $s_j = \sum_{k = 1}^{\lceil \log_2( d_{max} + 1)\rceil}2^{k-1}s_{jk}$. To express the constraint $\sum_{i \in I} y_i \leq p$ as an equality constraint, we introduce an integer slack variable $u \in \{0,...,p\}$. This slack variable is then encoded in binary form using $u_k \in \{0,1\}, k \in \{1,...,\lceil \log_2(p+1) \rceil\}$, such that $u = \sum_{k = 1}^{\lceil \log_2(p + 1) \rceil} 2^{k-1}u_k$.
\newline
\newline
We set $\ell^z = \lceil \log_2( d_{max} - d_{min}  + 1) \rceil$ and $\ell^s = \lceil \log_2( d_{max}  + 1) \rceil$.
Combining all components yields the following cost Hamiltonian for \eqref{p-Center-Problem}:
\begin{subequations} \label{p-Center-QUBO}
	\begin{align}
		H & = d_{min} +
		\sum_{k = 1}^{\ell^z}2^{k-1}z_k  \\
		&+ P \sum_{j \in J}\left(1 - \sum_{i \in I} x_{ij} \right)^2 \\
		&+  P \sum_{j \in J} \Bigg( d_{min} + \sum_{k = 1}^{\ell^z}2^{k-1}z_k - \sum_{k = 1}^{\ell^s}2^{k-1}s_{jk} - \sum_{i \in I}d_{ij}x_{ij} \Bigg)^2
		\\ &  + P \sum_{i \in I} \sum_{j \in J}x_{ij}(1 - y_{i}) \\
		& + P \left(p - \sum_{i \in I}y_i - \sum_{k = 1}^{\lceil \log_2(p + 1) \rceil} 2^{k-1}u_k\right)^2.
	\end{align}
\end{subequations}
\begin{lemma}
	The $p$-Center Problem \eqref{p-Center-Problem} can be formulated as a QUBO with $n^{p\text{-Center}} = \lvert I \rvert \lvert J \rvert + \lvert I \rvert + \lceil  \log_2( d_{max} -  d_{min} + 1) \rceil + \lvert J \rvert \lceil  \log_2( d_{max} + 1) \rceil + \lceil  \log_2(p + 1) \rceil$ binary variables. Therefore, the QUBO formulation \eqref{p-Center-QUBO} can be solved on a gate-based Quantum Computer with $n^{p\text{-Center}}$ physical qubits via QAOA or WS-QAOA.
\end{lemma}
\begin{proof}
	In the QUBO formulation \eqref{p-Center-QUBO}, the binary variables $x_{ij}$ for $i \in I, j \in J$, $y_i$ for $i \in I$, and the auxiliary binary variables $z_k$ for $k = 1,...,\lceil \log_2( d_{max} - d_{min}  + 1) \rceil$, $s_{jk}$ for $j \in J, k = 1,...,\lceil \log_2( d_{max} + 1) \rceil$, and $u_k$ for $k = 1,...,\lceil \log_2(p+1) \rceil$ are used. Consequently, the total number of binary variables is $n^{p\text{-Center}} = \lvert I \rvert \lvert J \rvert + \lvert I \rvert + \lceil  \log_2( d_{max} -  d_{min} + 1) \rceil + \lvert J \rvert \lceil  \log_2( d_{max} + 1) \rceil + \lceil  \log_2(p + 1) \rceil$, and the second part of the statement follows from Proposition \ref{prop:Physical_Qubits}.
\end{proof}
\begin{corollary}
	\label{korollar:p_center_penalty}
	Let the penalty parameter in \eqref{p-Center-QUBO} be $P = \sum_{i = 1}^{\lceil \log_2(d_{max} - d_{min} + 1)\rceil}2^{k-1} + 1$. Then, any global minimizer of \eqref{p-Center-QUBO} yields a global minimizer of \eqref{p-Center-Problem}.
\end{corollary}
\begin{proof}
	For the $p$-Center Problem, the objective function of its equivalent formulation of the form \eqref{BP}, used to construct the QUBO \eqref{p-Center-QUBO}, is $\sum_{k = 1}^{\lceil \log_2( d_{max} - d_{min}  + 1) \rceil}2^{k-1}z_k$. Consequently, the statement follows from Corollary \ref{theorem:PenPos}.
\end{proof}


\subsection{Fixed-Charge Facility Location Problem}
\label{sec:QUBO_FCFLP}
The Fixed-Charge Facility Location Problem (FCFLP) is a classical problem in location science and operations research, in which two decisions must be made:
\begin{enumerate}
	\item \textit{location decisions}, determinig which facilities $i \in I$ to open, and
	\item \textit{allocation decisions}, specifying how client demands $j \in J$ are served from the opened facilities. 
\end{enumerate}
We assume $I = J = \{1,...,n\}$. Each facility opening incurs a fixed (setup) cost, while serving client demands generates variable assignment costs. The objective is to determine a cost-minimizing selection of facilities to open and an assignment of clients to those facilities that satisfies all demand and capacity constraints \cite{Conforti2014, LocationScience:ChapterFCFLP}. The model parameters are as follows: $q_i > 0$ denotes the capacity of facility $i \in I$, $d_j > 0$ is the demand of client $j \in J$, $f_i$ is the fixed cost of opening facility $i$, and $c_{ij} \geq 0$ is the cost of serving all demand of client $j$ from facility $i$. We assume $q_i, d_j \in \N$ for all $i \in I, j \in J$.
\newline
For the formulation of the Fixed-Charge Facility Location Problem we introduce the binary decision variables:
\begin{itemize}
	\item $y_i = 1$, if a facility is open at location $i \in I$, and $0$ otherwise.
	\item $x_{ij} = 1$, if client $j \in J$ is served by facility $i \in I$, and $0$ otherwise.
\end{itemize}
The Fixed-Charge Facility Location Problem is then formulated as the following optimization problem:
\begin{subequations}\label{FCFLP-1}
	\begin{align}
		\displaystyle\min_{x,y} \quad & \sum_{i \in I}f_iy_i  + \sum_{i \in I} \sum_{j \in J}c_{ij}x_{ij}   \\
		\mbox{s.t.} \quad & \sum_{i \in I} x_{ij} = 1, & \forall j \in J, \\
		& \sum_{j \in J} d_jx_{ij} \leq q_iy_i  & \forall  i \in I, \label{FCFLP-1-A}\\
		& x_{ij} \in \{0,1\}, & \forall i \in I, j \in J, \\
		& y_i \in \{0,1\}, & \forall i \in I.
	\end{align}
\end{subequations}
The objective function consists of two components: the first term represents the total fixed cost of opening facilities, and the second term accounts for the total assignment (service) cost of allocating clients to open facilites. The constraints \eqref{FCFLP-1-A} play a double role: they ensure that no facility exceeds its capacity and simultaneously couples the location and allocation decisions. The assignment constraints guarantee that each client is served by exactly one open facility.
\newline
An alternative but equivalent formulation of \eqref{FCFLP-1} is given by
\begin{subequations}\label{FCFLP-2} 
	\begin{align}
		\displaystyle\min_{x,y} \quad & \sum_{i \in I}f_iy_i + \sum_{i \in I} \sum_{j \in J}c_{ij}x_{ij} \\
		\mbox{s.t.} \quad & \sum_{i \in I} x_{ij} = 1, & \forall j \in J, \\
		& \sum_{j \in J} d_jx_{ij} \leq q_i & \forall  i \in I, \label{FCFLP-2-A} \\
		& x_{ij} \leq y_i, & \forall i \in I, j \in J,  \label{FCFLP-2-B}\\
		& x_{ij} \in \{0,1\}, & \forall i \in I, j \in J, \\
		& y_j \in \{0,1\}, & \forall i \in I.
	\end{align}
\end{subequations}
In this formulation, the constraint \eqref{FCFLP-1-A} is replaced by the two constraints \eqref{FCFLP-2-A} and \eqref{FCFLP-2-B}. Constraint \eqref{FCFLP-2-A} ensures that the maximum capacity of a facility is not exceeded, while constraint \eqref{FCFLP-2-B} couples the variables $x$ and $y$, thereby ensuring that clients can only be assigned to open facilities.
\newline
We refer to \eqref{FCFLP-1} as the \textit{aggregated} formulation and to \eqref{FCFLP-2} as the \textit{disaggregated} formulation.
\begin{example}
	\label{exp:FCFLP_3}
	We consider the FCFLP with $\lvert I \rvert = \lvert J \rvert = n = 3$. The input data are given by the assignment cost matrix $C \in \R^{3 \times 3}$, the demand vector $d \in \R^3$, the fixed cost vector $f \in \R^3$, and the maximum capacity vector $q \in \R^3$:
	\begin{align*}
		C = \begin{pmatrix}
			0 & 4 & 4 \\
			9 & 0 & 2 \\
			5 & 5 & 0 
		\end{pmatrix},
		d = \begin{pmatrix}
			3 \\ 8 \\10
		\end{pmatrix}
		, 
		f = \begin{pmatrix}
			25 \\ 9 \\ 17
		\end{pmatrix}, 
		q = \begin{pmatrix}
			12 \\ 10 \\ 10
		\end{pmatrix}
	\end{align*}
	The optimal solution opens the facilities $1$ and $2$, i.e., $y_1 = y_2 = 1$, while $y_3 = 0$. The corresponding client-facility assignments are $x_{11} = x_{12} = x_{23} = 1$, with all remaining $x_{ij} = 0$. The resulting optimal objective value is $40$.
\end{example}
The disaggregated formulation \eqref{FCFLP-2} is theoretically stronger than the aggregated one \eqref{FCFLP-1}, as the constraint set \eqref{FCFLP-2-A} $+$ \eqref{FCFLP-2-B} dominates \eqref{FCFLP-1-A} in the LP-relaxation. 
\newline
Consequently, as discussed by \cite{Conforti2014}, the disaggregated formulation typically yields tighter bounds in a branch-and-cut algorithm. However, modern state-of-the-art solvers can automatically detect and generate disaggregated constraints on the fly whenever they are violated by the current feasible solution. Therefore, in practice, the aggregated formulation is often preferable, as it is more compact and leads to significantly smaller and faster to solve linear relaxations \cite{Conforti2014}.
\newline
This is an interesting observation, and it raises the question of whether a similar behavior can be observed for QUBO formulations, namely, whether one formulation performs bettern than the other. Therefore, in the numerical experiments in Section \ref{sec:Numerical_Experiments}, we consider and compare both the aggregated and the disaggregated formulations.
\newline
\newline
\textbf{QUBO formulation}: To transform the inequality constraints \eqref{FCFLP-1-A} and \eqref{FCFLP-2-A} into equality constraints, we introduce integer slack variables $u_i \in \{0,...,q_i\}$ for each $i \in I$. This yields the reformulated constraints $\sum_{j \in J} d_jx_{ij} + u_i = q_iy_i$ and $\sum_{j \in J} d_jx_{ij} + u_i =  q_i$ for each $i \in I$. For the binary encoding of each slack variable $u_i$, we define $k_i = \lceil \log_2(q_i + 1) \rceil$ and introduce binary variables $z_{ik} \in \{0,1\}$ for $k \in \{1,...,k_i\}$. The integer variable $u_i$ is then represented as $u_i = \sum_{k = 1}^{k_i}2^{k-1}z_{ik}$.
\newline
\newline
This yields the following cost Hamiltonians: for \eqref{FCFLP-1}, it is given by
\begin{subequations}\label{FCFLP1-QUBO}
	\begin{align}
		H =& \sum_{i \in I} f_iy_i + \sum_{i \in I} \sum_{j \in J}c_{ij}x_{ij}    \\
		&+ P \sum_{j \in J}\left( 1 - \sum_{i \in I}x_{ij} \right)^2\\
		& + P \sum_{i \in I} \left(q_iy_i - \sum_{j \in J}d_jx_{ij} - \sum_{k = 1}^{k_i} 2^{k-1}z_{ik} \right)^2,
	\end{align}
\end{subequations}
and for \eqref{FCFLP-2}, it is given by
\begin{subequations}\label{FCFLP2-QUBO}
	\begin{align}
		H =& \sum_{i \in I} f_iy_i + \sum_{i \in I} \sum_{j \in J}c_{ij}x_{ij}  \\
		&+ P \sum_{j \in J}\left( 1 - \sum_{i \in I}x_{ij} \right)^2 \\
		& + P \sum_{i \in I} \left(q_i - \sum_{j \in J}d_jx_{ij} - \sum_{k = 1}^{k_i} 2^{k-1}z_{ik} \right)^2\\
		& + P \sum_{i \in I} \sum_{j \in J}x_{ij}(1 - y_{i}).
	\end{align}
\end{subequations}
\begin{lemma}
	\label{lemma:fcflp_qubits}
	The FCFLP can be formulated as a QUBO with $n^{FCFLP} = \lvert I \rvert \lvert J \rvert + \lvert I \rvert + \sum_{i \in I}\lceil \log_2(q_i + 1) \rceil$ binary variables. Therefore, the QUBO formulations \eqref{FCFLP1-QUBO} and \eqref{FCFLP2-QUBO} can be solved on a gate-based Quantum Computer with $n^{FCFLP}$ physical qubits via QAOA or WS-QAOA.
\end{lemma}
\begin{proof}
	Both QUBO formulations, \eqref{FCFLP1-QUBO} and \eqref{FCFLP2-QUBO}, use the binary variables $x_{ij}$ for $i \in I, j \in J$, $y_i$ for $i \in I$, and the auxiliary binary variables $z_{ik} \in \{0,1\}$ for $i \in I$, $k = 1,...,\lceil \log_2(q_i + 1) \rceil$. Consequently, the total number of binary variables is  $n^{FCFLP} = \lvert I \rvert \lvert J \rvert + \lvert I \rvert + \sum_{i \in I}\lceil \log_2(q_i + 1) \rceil$, and the second part of the statement follows from Proposition \ref{prop:Physical_Qubits}.
\end{proof}
\begin{corollary}
	\label{korollar:fcflp}
	Let the penalty parameter in \eqref{FCFLP1-QUBO}, \eqref{FCFLP2-QUBO} be $P = \sum_{i \in I}f_i + \sum_{i \in I}\sum_{j \in J}cc_{ij} + 1$. Then, any global minimizer of \eqref{FCFLP1-QUBO}, \eqref{FCFLP2-QUBO} yields a global minimizer of \eqref{FCFLP-1}, \eqref{FCFLP-2}.
\end{corollary}
\begin{proof}
	For the FCFLP, the objective function of the equivalent formulations of the form \eqref{BP}, used to construct both QUBOs \eqref{FCFLP1-QUBO} and \eqref{FCFLP2-QUBO}, is $\sum_{i \in I}f_iy_i + \sum_{i \in I}\sum_{j \in J}c_{ij}x_{ij}$.
	Thus, the statement is a consequence of Corollary \ref{theorem:PenPos}.
\end{proof}
\begin{example}
	We consider Example \ref{exp:FCFLP_3}. The penalty parameter is given by $P = 81$. First, we consider the QUBO formulation for \eqref{FCFLP-1}. We first introduce the slack variables. It holds $\lceil \log_2(12+1) \rceil = \lceil \log_2(10+1) \rceil = \lceil \log_2(10+1) \rceil = 4$. Therefore, we introduce the variables $s_{11},...,s_{14},s_{21},...,s_{24},s_{31},...,s_{34}$. We then consider the flattened vector $(x,y,s)$ for the QUBO formulation. 
	The QUBO matrix for \eqref{FCFLP-1} is given by:
	\begin{align*}
		Q_1 =
		\resizebox{0.87\linewidth}{!}{$
			\left[
			\begin{array}{*{24}c}
				648 & 3888 & 4860 & 162 & 0 & 0 & 162 & 0 & 0 & -5832 & 0 & 0 & 486 & 972 & 1944 & 3888 & 0 & 0 & 0 & 0 & 0 & 0 & 0 & 0 \\
				0 & 5107 & 12960 & 0 & 162 & 0 & 0 & 162 & 0 & -15552 & 0 & 0 & 1296 & 2592 & 5184 & 10368 & 0 & 0 & 0 & 0 & 0 & 0 & 0 & 0 \\
				0 & 0 & 8023 & 0 & 0 & 162 & 0 & 0 & 162 & -19440 & 0 & 0 & 1620 & 3240 & 6480 & 12960 & 0 & 0 & 0 & 0 & 0 & 0 & 0 & 0 \\
				0 & 0 & 0 & 657 & 3888 & 4860 & 162 & 0 & 0 & 0 & -4860 & 0 & 0 & 0 & 0 & 0 & 486 & 972 & 1944 & 3888 & 0 & 0 & 0 & 0 \\
				0 & 0 & 0 & 0 & 5103 & 12960 & 0 & 162 & 0 & 0 & -12960 & 0 & 0 & 0 & 0 & 0 & 1296 & 2592 & 5184 & 10368 & 0 & 0 & 0 & 0 \\
				0 & 0 & 0 & 0 & 0 & 8021 & 0 & 0 & 162 & 0 & -16200 & 0 & 0 & 0 & 0 & 0 & 1620 & 3240 & 6480 & 12960 & 0 & 0 & 0 & 0 \\
				0 & 0 & 0 & 0 & 0 & 0 & 653 & 3888 & 4860 & 0 & 0 & -4860 & 0 & 0 & 0 & 0 & 0 & 0 & 0 & 0 & 486 & 972 & 1944 & 3888 \\
				0 & 0 & 0 & 0 & 0 & 0 & 0 & 5108 & 12960 & 0 & 0 & -12960 & 0 & 0 & 0 & 0 & 0 & 0 & 0 & 0 & 1296 & 2592 & 5184 & 10368 \\
				0 & 0 & 0 & 0 & 0 & 0 & 0 & 0 & 8019 & 0 & 0 & -16200 & 0 & 0 & 0 & 0 & 0 & 0 & 0 & 0 & 1620 & 3240 & 6480 & 12960 \\
				0 & 0 & 0 & 0 & 0 & 0 & 0 & 0 & 0 & 11689 & 0 & 0 & -1944 & -3888 & -7776 & -15552 & 0 & 0 & 0 & 0 & 0 & 0 & 0 & 0 \\
				0 & 0 & 0 & 0 & 0 & 0 & 0 & 0 & 0 & 0 & 8109 & 0 & 0 & 0 & 0 & 0 & -1620 & -3240 & -6480 & -12960 & 0 & 0 & 0 & 0 \\
				0 & 0 & 0 & 0 & 0 & 0 & 0 & 0 & 0 & 0 & 0 & 8117 & 0 & 0 & 0 & 0 & 0 & 0 & 0 & 0 & -1620 & -3240 & -6480 & -12960 \\
				0 & 0 & 0 & 0 & 0 & 0 & 0 & 0 & 0 & 0 & 0 & 0 & 81 & 324 & 648 & 1296 & 0 & 0 & 0 & 0 & 0 & 0 & 0 & 0 \\
				0 & 0 & 0 & 0 & 0 & 0 & 0 & 0 & 0 & 0 & 0 & 0 & 0 & 324 & 1296 & 2592 & 0 & 0 & 0 & 0 & 0 & 0 & 0 & 0 \\
				0 & 0 & 0 & 0 & 0 & 0 & 0 & 0 & 0 & 0 & 0 & 0 & 0 & 0 & 1296 & 5184 & 0 & 0 & 0 & 0 & 0 & 0 & 0 & 0 \\
				0 & 0 & 0 & 0 & 0 & 0 & 0 & 0 & 0 & 0 & 0 & 0 & 0 & 0 & 0 & 5184 & 0 & 0 & 0 & 0 & 0 & 0 & 0 & 0 \\
				0 & 0 & 0 & 0 & 0 & 0 & 0 & 0 & 0 & 0 & 0 & 0 & 0 & 0 & 0 & 0 & 81 & 324 & 648 & 1296 & 0 & 0 & 0 & 0 \\
				0 & 0 & 0 & 0 & 0 & 0 & 0 & 0 & 0 & 0 & 0 & 0 & 0 & 0 & 0 & 0 & 0 & 324 & 1296 & 2592 & 0 & 0 & 0 & 0 \\
				0 & 0 & 0 & 0 & 0 & 0 & 0 & 0 & 0 & 0 & 0 & 0 & 0 & 0 & 0 & 0 & 0 & 0 & 1296 & 5184 & 0 & 0 & 0 & 0 \\
				0 & 0 & 0 & 0 & 0 & 0 & 0 & 0 & 0 & 0 & 0 & 0 & 0 & 0 & 0 & 0 & 0 & 0 & 0 & 5184 & 0 & 0 & 0 & 0 \\
				0 & 0 & 0 & 0 & 0 & 0 & 0 & 0 & 0 & 0 & 0 & 0 & 0 & 0 & 0 & 0 & 0 & 0 & 0 & 0 & 81 & 324 & 648 & 1296 \\
				0 & 0 & 0 & 0 & 0 & 0 & 0 & 0 & 0 & 0 & 0 & 0 & 0 & 0 & 0 & 0 & 0 & 0 & 0 & 0 & 0 & 324 & 1296 & 2592 \\
				0 & 0 & 0 & 0 & 0 & 0 & 0 & 0 & 0 & 0 & 0 & 0 & 0 & 0 & 0 & 0 & 0 & 0 & 0 & 0 & 0 & 0 & 1296 & 5184 \\
				0 & 0 & 0 & 0 & 0 & 0 & 0 & 0 & 0 & 0 & 0 & 0 & 0 & 0 & 0 & 0 & 0 & 0 & 0 & 0 & 0 & 0 & 0 & 5184 \\
			\end{array}
			\right]
			$}
	\end{align*}
	The offset for the QUBO formulation \eqref{FCFLP1-QUBO} is $c = 243$.
	The QUBO matrix for \eqref{FCFLP-2} is given by:
	\begin{align*}
		Q_2 =
		\resizebox{0.87\linewidth}{!}{$
			\left[
			\begin{array}{*{24}c}
				-5103 & 3888 & 4860 & 162 & 0 & 0 & 162 & 0 & 0 & -81 & 0 & 0 & 486 & 972 & 1944 & 3888 & 0 & 0 & 0 & 0 & 0 & 0 & 0 & 0 \\
				0 & -10364 & 12960 & 0 & 162 & 0 & 0 & 162 & 0 & -81 & 0 & 0 & 1296 & 2592 & 5184 & 10368 & 0 & 0 & 0 & 0 & 0 & 0 & 0 & 0 \\
				0 & 0 & -11336 & 0 & 0 & 162 & 0 & 0 & 162 & -81 & 0 & 0 & 1620 & 3240 & 6480 & 12960 & 0 & 0 & 0 & 0 & 0 & 0 & 0 & 0 \\
				0 & 0 & 0 & -4122 & 3888 & 4860 & 162 & 0 & 0 & 0 & -81 & 0 & 0 & 0 & 0 & 0 & 486 & 972 & 1944 & 3888 & 0 & 0 & 0 & 0 \\
				0 & 0 & 0 & 0 & -7776 & 12960 & 0 & 162 & 0 & 0 & -81 & 0 & 0 & 0 & 0 & 0 & 1296 & 2592 & 5184 & 10368 & 0 & 0 & 0 & 0 \\
				0 & 0 & 0 & 0 & 0 & -8098 & 0 & 0 & 162 & 0 & -81 & 0 & 0 & 0 & 0 & 0 & 1620 & 3240 & 6480 & 12960 & 0 & 0 & 0 & 0 \\
				0 & 0 & 0 & 0 & 0 & 0 & -4126 & 3888 & 4860 & 0 & 0 & -81 & 0 & 0 & 0 & 0 & 0 & 0 & 0 & 0 & 486 & 972 & 1944 & 3888 \\
				0 & 0 & 0 & 0 & 0 & 0 & 0 & -7771 & 12960 & 0 & 0 & -81 & 0 & 0 & 0 & 0 & 0 & 0 & 0 & 0 & 1296 & 2592 & 5184 & 10368 \\
				0 & 0 & 0 & 0 & 0 & 0 & 0 & 0 & -8100 & 0 & 0 & -81 & 0 & 0 & 0 & 0 & 0 & 0 & 0 & 0 & 1620 & 3240 & 6480 & 12960 \\
				0 & 0 & 0 & 0 & 0 & 0 & 0 & 0 & 0 & 25 & 0 & 0 & 0 & 0 & 0 & 0 & 0 & 0 & 0 & 0 & 0 & 0 & 0 & 0 \\
				0 & 0 & 0 & 0 & 0 & 0 & 0 & 0 & 0 & 0 & 9 & 0 & 0 & 0 & 0 & 0 & 0 & 0 & 0 & 0 & 0 & 0 & 0 & 0 \\
				0 & 0 & 0 & 0 & 0 & 0 & 0 & 0 & 0 & 0 & 0 & 17 & 0 & 0 & 0 & 0 & 0 & 0 & 0 & 0 & 0 & 0 & 0 & 0 \\
				0 & 0 & 0 & 0 & 0 & 0 & 0 & 0 & 0 & 0 & 0 & 0 & -1863 & 324 & 648 & 1296 & 0 & 0 & 0 & 0 & 0 & 0 & 0 & 0 \\
				0 & 0 & 0 & 0 & 0 & 0 & 0 & 0 & 0 & 0 & 0 & 0 & 0 & -3564 & 1296 & 2592 & 0 & 0 & 0 & 0 & 0 & 0 & 0 & 0 \\
				0 & 0 & 0 & 0 & 0 & 0 & 0 & 0 & 0 & 0 & 0 & 0 & 0 & 0 & -6480 & 5184 & 0 & 0 & 0 & 0 & 0 & 0 & 0 & 0 \\
				0 & 0 & 0 & 0 & 0 & 0 & 0 & 0 & 0 & 0 & 0 & 0 & 0 & 0 & 0 & -10368 & 0 & 0 & 0 & 0 & 0 & 0 & 0 & 0 \\
				0 & 0 & 0 & 0 & 0 & 0 & 0 & 0 & 0 & 0 & 0 & 0 & 0 & 0 & 0 & 0 & -1539 & 324 & 648 & 1296 & 0 & 0 & 0 & 0 \\
				0 & 0 & 0 & 0 & 0 & 0 & 0 & 0 & 0 & 0 & 0 & 0 & 0 & 0 & 0 & 0 & 0 & -2916 & 1296 & 2592 & 0 & 0 & 0 & 0 \\
				0 & 0 & 0 & 0 & 0 & 0 & 0 & 0 & 0 & 0 & 0 & 0 & 0 & 0 & 0 & 0 & 0 & 0 & -5184 & 5184 & 0 & 0 & 0 & 0 \\
				0 & 0 & 0 & 0 & 0 & 0 & 0 & 0 & 0 & 0 & 0 & 0 & 0 & 0 & 0 & 0 & 0 & 0 & 0 & -7776 & 0 & 0 & 0 & 0 \\
				0 & 0 & 0 & 0 & 0 & 0 & 0 & 0 & 0 & 0 & 0 & 0 & 0 & 0 & 0 & 0 & 0 & 0 & 0 & 0 & -1539 & 324 & 648 & 1296 \\
				0 & 0 & 0 & 0 & 0 & 0 & 0 & 0 & 0 & 0 & 0 & 0 & 0 & 0 & 0 & 0 & 0 & 0 & 0 & 0 & 0 & -2916 & 1296 & 2592 \\
				0 & 0 & 0 & 0 & 0 & 0 & 0 & 0 & 0 & 0 & 0 & 0 & 0 & 0 & 0 & 0 & 0 & 0 & 0 & 0 & 0 & 0 & -5184 & 5184 \\
				0 & 0 & 0 & 0 & 0 & 0 & 0 & 0 & 0 & 0 & 0 & 0 & 0 & 0 & 0 & 0 & 0 & 0 & 0 & 0 & 0 & 0 & 0 & -7776
			\end{array}
			\right]
			$}
	\end{align*}
	The offset for the QUBO formulation \eqref{FCFLP2-QUBO} is $c = 28107$.
	We observe that the QUBO matrix for the aggregated formulation \eqref{FCFLP-1} $Q_1$ and the disaggregated formulation \eqref{FCFLP-2} $Q_2$ look quite different. Notably, the matrix $Q_1$ contains fewer nonzero entries than $Q_2$, since in the QUBO formulation \eqref{FCFLP1-QUBO}, the variables $y$ and $z$ appear together in a single penalty term, whereas in the QUBO formulation \eqref{FCFLP2-QUBO} they do not.
\end{example}


\subsection{Generalized Assignment Problem}
\label{sec:QUBO_GAP}
We consider the same general setting as in the previous section. However, in this case, the \textit{location decisions} are assumed to be fixed, that is, the set of open facilities $S \subseteq I$ is given in advance. The remaining task is to determine the optimal allocation of clients to these pre-established facilities, subject to capacity constraints. This problem is known as the Generalized Assignment Problem (GAP), a well-known combinatorial optimization problem that remains $\mathbf{NP}$-hard \cite{Fisher:GAP, LocationScience:ChapterFCFLP}.
\newline
The GAP can be formulated as the following optimization problem:
\begin{subequations}\label{GAP}
	\begin{align}
		\displaystyle\min_{x} \quad & \sum_{i \in S} \sum_{j \in J}c_{ij}x_{ij}   \\
		\mbox{s.t.} \quad & \sum_{i \in S}x_{ij} = 1, & \forall j \in J, \\
		& \sum_{j \in J} d_jx_{ij} \leq q_i, & \forall i \in S, \\
		& x_{ij} \in \{0,1\}, & \forall  i \in S, j \in J.
	\end{align}
\end{subequations}
The objective function minimizes the total assignment cost of clients to facilities, where $c_{ij}$ denotes the cost of serving client $j \in J$ from facility $i \in S$. The first constraint ensures that each client is assigned to exactly one facility, while the second enforces the capacity limits $q_i$ of the facilities.
\newline
\newline
\textbf{QUBO formulation}: To transform \eqref{GAP} into a QUBO, we follow the same procedure as for the FCFLP. Specifically, we introduce integer slack variables $u_i \in \{0,...,q_i\}$ for each $i \in S$ to convert the capacity constraints into equality constraints of the form: $\sum_{j \in J}d_jx_{ij} + u_i = q_i$, for all $i \in S$. Each integer slack variable $u_i$ is represented in binary form as $u_i = \sum_{k = 1}^{\ell_i}2^{k-1}z_{ik}$, where $\ell_i = \lceil \log_2(q_i + 1) \rceil $ and $z_{ik} \in \{0,1\}$.
\newline
\newline
This yields the following cost Hamiltonian for \eqref{GAP}:
\begin{subequations}\label{GAP-QUBO}
	\begin{align}
		H =& \sum_{i \in S} \sum_{j \in J} c_{ij} x_{ij}  \\
		& + P \sum_{j \in J}\left(1 - \sum_{i \in S}x_{ij}\right)^2 \\
		& + P \sum_{i \in S} \left(q_i - \sum_{j \in J}d_jx_{ij} - \sum_{k = 1}^{\ell_i} 2^{k-1}z_{ik} \right)^2.
	\end{align}
\end{subequations}
\begin{lemma}
	The GAP can be formulated as a QUBO with $n^{GAP} = \lvert S \rvert \lvert J \rvert + \sum_{i \in S} \lceil \log_2(q_i + 1) \rceil$ binary variables. Therefore, the QUBO formulation \eqref{GAP-QUBO} can be solved on a gate-based Quantum Computer with $n^{GAP}$ physical qubits via QAOA or WS-QAOA.
\end{lemma}
\begin{proof}
	The QUBO formulation \eqref{GAP-QUBO} uses the binary variables $x_{ij}$ for $i \in S, j \in J$, and the auxiliary binary variables $z_{ik}$ for $i \in S, k = 1,...,\lceil \log_2(q_i + 1) \rceil$. This gives a total of $n^{GAP} = \lvert S \rvert \lvert J \rvert + \sum_{i \in S} \lceil \log_2(q_i + 1) \rceil$ binary variables, and the second part of the statement follows from Proposition \ref{prop:Physical_Qubits}.
\end{proof}
\begin{corollary}
	\label{korollar:gap}
	Let the penalty parameter in \eqref{GAP-QUBO} be $P = \sum_{i \in S}\sum_{j \in J}c_{ij} + 1$. Then, any global minimizer of \eqref{GAP-QUBO} yields a global minimizer of \eqref{GAP}.
\end{corollary}
\begin{proof}
	For the GAP, the objective function of its equivalent formulation of the form \eqref{BP}, used to construct the QUBO \eqref{GAP-QUBO}, is $\sum_{i \in S}\sum_{j \in J}c_{ij}x_{ij}$.
	Hence, the statement is a consequence of Corollary \ref{theorem:PenPos}.
\end{proof}

\subsection{Discrete Ordered Median Problem}
\label{sec:QUBO_DOMP}
The Discrete Ordered Median Problem (DOMP) constitutes a powerful and unifying framework in discrete location analysis, encompassing a wide range of classical models such as the $p$-median, $p$-center, and uncapacitated facility location probelms. Originally introduced in \cite{Nickel:DOMP}, the DOMP extends these formulations by incorporating rank-dependent weights into the objective function, thereby enabling an easy transition between different location criteria. In contrast to traditional formulations that minimize either the total or the maximum service cost, the DOMP minimizes an ordered weighted average of allocation costs, where each cost is weighted according to its rank in the ordered list of client-facility distances \cite{LocationScience:ChapterOMP}.
\newline
This rank-based structure captures decision-maker's preferences between efficiency and equity, enabling a controlled trade-off between minimizing average service costs and reducing disparities in service quality among clients. Such modeling flexibility makes the DOMP particularly well-suited for strategic planning contexts, such as logistics network design, public facility location, and emergency service systems, where both performance and fairness considerations are essential \cite{Nickel:DOMPExact, Nickel:DOMPHeuristic}.
\newline
From a mathematical perspective, the DOMP introduces a complex combinatorial structure arising from the embedded sorting operation in the objective function. This interdependence between location, allocation and ranking renders the problem signficantly more difficult to solve in practice than classical discrete location models. Early formulations combined integer location and assignment constraints with additional permutation variables to enforce the cost ordering, resulting in integer nonlinear programs. In the following, we consider one of the foundational formulations from this early line of research from \cite{Nickel:DOMPExact}.
\newline
\newline
Formally, let $x_j \in \{0,1\}$ indicate whether a facility is located at $j \in J$, and let $y_{ij} \in \{0,1\}$ indicate whether the demand of client $i \in I$ is served by a facility at site $j$. Again, we assume that $I$ and $J$ coincide, and we assume $I = J = \{1,...,M\}$. To capture the rank-dependent cost structure of the DOMP, we also introduce the binary variable $s_{ij} \in \{0,1\}$, which equals $1$ if the cost of supplying client $j$ is the $i$-th smallest among all clients, and $0$ otherwise. These variables encode the ordering of clients by increasing service cost.
\newline
The classical $N$-Median Problem is modeled using the standard constraints: each facility is either open or closed, exactly $N$ facilities are selected, each client is assigned to exactly one facility, and assignment is only possible to open facilities:
\begin{align*}
	\sum_{j = 1}^M x_j &= N, \\ \sum_{j = 1}^My_{ij} &= 1, &\forall i \in I, \\
	y_{ij} &\leq x_j,& \forall i \in I, j \in J.
\end{align*}
The ranking variables $s_{ij}$ are constrained to define a valid permutation of the clients:
\begin{align*}
	\sum_{i = 1}^M s_{ij} = 1, \forall j \in J,\quad \quad \sum_{j = 1}^M s_{ij} = 1, \forall i \in I.
\end{align*}
Moreover, to ensure the clients are ordered by increasing allocation cost, the following monotonicity constraints are imposed for $i = 1,...,M-1$: 
\begin{align*}
	\sum_{k = 1}^Ms_{ik} \left( \sum_{j = 1}^M y_{kj}c_{kj}\right) 
	\leq \sum_{k = 1}^M s_{i+1,k} \left( \sum_{j = 1}^M y_{kj}c_{kj} \right).
\end{align*}
The DOMP objective weighs the clients' service costs $c \geq 0$ according to the chosen rank-dependent coefficients $\lambda_i$:
\begin{align*}
	\sum_{i = 1}^M \sum_{k = 1}^M s_{ik} \left( \sum_{j = 1}^M y_{kj}c_{kj}\right) \lambda_i.
\end{align*}
We note that choosing the parameter $\lambda = (1,1,...,1)$ makes the DOMP equivalent to the $N$-Median Problem, whereas $\lambda = (0,...,0,1)$ makes the DOMP equivalent to the $N$-Center Problem.
\newline
\newline
Together, these constraints and the objective define the DOMP as a nonlinear integer program with a nontrivial combinatorial structure arising from the embedded sorting of service costs. The DOMP is now given by:
\begin{subequations} \label{DOMP}
	\begin{align}
		\min_{x,y,s} \quad & \sum_{i = 1}^M \sum_{k = 1}^M s_{ik} \left( \sum_{j = 1}^M y_{kj}c_{kj}\right) \lambda_i  \\
		\mbox{s.t.} \quad & \sum_{j = 1}^Mx_j = N, \\
		& \sum_{j = 1}^M y_{ij} = 1,\quad \quad  \forall i = 1,...,M, \\
		& y_{ij} \leq x_j, \quad \quad \forall i,j = 1,...,M, \\
		& \sum_{i = 1}^M s_{ij} = 1,\quad \quad \forall j = 1,...,M, \\
		& \sum_{j = 1}^M s_{ij} = 1, \quad \quad\forall i = 1,...,M, \\
		& \sum_{k = 1}s_{ik} \left( \sum_{j = 1}^M y_{kj}c_{kj} \right)\\
		&  \leq \sum_{k = 1}^M s_{i+1, k} \left( \sum_{j = 1}^M y_{kj}c_{kj} \right), \quad \forall i = 1,....,M-1,\\
		& x \in \{0,1\}^M, y \in \{0,1\}^{M^2}, s \in \{0,1\}^{M^2}.
	\end{align}
\end{subequations}
\textbf{QUBO formulation}: Let $SC$ denote the sum of all entries of the cost matrix $c$, i.e., $SC = \sum_{i = 1}^M \sum_{j = 1}^M c_{ij}$. To transform the monotonicity constraints into equality constraints, we introduce slack variables $v_i \in \{0,...,SC\}$ for each $i = 1,...,M-1$.
To obtain a binary representation, each slack variable $v_i$ is encoded using $\lceil \log_2(SC + 1) \rceil$ binary variables $v_{ik} \in \{0,1\}$, with $v_i = \sum_{k=1}^{\lceil \log_2(SC + 1) \rceil}2^{k-1}v_{ik}$.
\newline
In a general \textit{PUBO (polynomial unconstrained binary optimization) formulation}, one could directly impose the penalty term for the monotonicity constraints:
	\begin{align*}
		P \sum_{i = 1}^{M-1} \Bigg(
		\sum_{k = 1}^{M}s_{i+1, k}\left(  \sum_{j = 1}^M y_{kj}c_{kj} \right)
		- \sum_{k = 1}^{M}s_{ik}\left(  \sum_{j = 1}^M y_{kj}c_{kj} \right)- \sum_{k = 1}^{\lceil \log_2(SC + 1) \rceil}2^{k-1}v_{ik}
		\Bigg)^2.
	\end{align*}
However, this is a polynomial of degree four and therefore not directly suitable for a QUBO formulation.
\newline
To reduce the degree to quadratic, we introduce auxiliary variables $\{0,1\} \ni u_{ikj} = s_{ik}y_{kj}$ for all $i,k,j = 1,...,M$.
The coupling $u_{ikj} = s_{ik}y_{kj}$ could be enforced via the penalty term: 
\begin{align*}
	P \sum_{i = 1}^M \sum_{k = 1}^M \sum_{j = 1}^M (u_{ikj} - s_{ik}y_{kj})^2,
\end{align*}
which is again quartic in its current form. Exploiting that all variables are binary variables, we replace the constraint $u_{ikj} = s_{ik}y_{kj}$ with the McCormick inequalities \cite{McCormick1976}, which are tight in this case and given by:
\begin{align*}
	u_{ikj} \leq s_{ik}, \quad u_{ikj} \leq y_{kj}, \quad u_{ikj} \geq s_{ik} + y_{kj} - 1.
\end{align*} Finally, the inequality $u_{ikj} \geq s_{ik} + y_{kj} - 1$ is transformed into an equality using a discrete slack variable $w_{ikj} \in \{0,1,2\}$. Binary encoding is achieved via $w_{ikj,1}, w_{ikj,2} \in \{0,1\}$ with $w_{ikj} = w_{ikj,1} + 2 w_{ikj,2}$. The factor $2$ reduces symmetry, which may improve the performance of quantum optimization algorithms.
\newline
\newline
Altogether, this yields the cost Hamiltonian:
\begin{subequations} \label{DOMP-QUBO} 
	\begin{align}
		H &= \sum_{i = 1}^M \sum_{k = 1}^M\sum_{j = 1}^M\lambda_i c_{kj} u_{ikj}  \\
		&+ P \left(\sum_{j =1}^Mx_j - N \right)^2+ P \sum_{i = 1}^M \left(\sum_{j = 1}^M y_{ij} -1\right)^2 \\
		& + P \sum_{i = 1}^M\sum_{j = 1}^M y_{ij}(1 - x_j) + P  \sum_{j = 1}^M\left( \sum_{i = 1}^M s_{ij} - 1\right)^2 \\
		& + P  \sum_{i = 1}^M \left(\sum_{j = 1}^Ms_{ij} - 1\right)^2 \\
		& + P \sum_{i = 1}^M\sum_{k = 1}^M \sum_{j = 1}^Mu_{ikj}(1 - s_{ik}) \\
		& + P  \sum_{i = 1}^M\sum_{k = 1}^M \sum_{j = 1}^Mu_{ikj}(1 - y_{kj})\\
		& + P \sum_{i = 1}^M\sum_{k=1}^M\sum_{j = 1}^M(-u_{ikj}+s_{ik}+y_{kj}-1+w_{ikj,1}+2w_{ikj,2}  )^2 \\
		& + P \sum_{i = 1}^{M-1} \left(\sum_{k = 1}^M \sum_{j = 1}^M c_{kj}(u_{ikj} - u_{i+1,kj}) + \sum_{k = 1}^{\ell_{SC}}2^{k-1}v_{ik} \right)^2,
	\end{align}
\end{subequations}
where $\ell_{SC} = \lceil \log_2(SC+1)\rceil$. 
\begin{lemma}
	The DOMP can be formulated as a QUBO with $n^{DOMP} = M + 2M^2 + 3M^3 + (M - 1) \lceil\log_2(SC + 1) \rceil$ binary variables, where $SC = \sum_{i = 1}^M \sum_{j = 1}^M c_{ij}$. Therefore, the QUBO formulation \eqref{DOMP-QUBO} can be solved on a gate-based Quantum Computer with $n^{DOMP}$ physical qubits via QAOA or WS-QAOA.
\end{lemma}
\begin{proof}
	For the QUBO formulation \eqref{DOMP-QUBO}, in addition to the binary variables $x_j$ for $j = 1,...,M$, $y_{ij}$ for $i,j = 1,...,M$, and $s_{ij}$ for $i,j = 1,...,M$, the auxiliary binary variables $v_{ik}$ for $i = 1,...,M-1, k = 1,..., \lceil \log_2(SC + 1) \rceil$, with $SC = \sum_{i = 1}^M \sum_{j = 1}^M c_{ij}$, $u_{ikj}$ for $i,k,j = 1,...,M$, and $w_{ikjt}$ for $i,k,j = 1,...,M, t \in \{1,2\}$ are used. This results in a total of $n^{DOMP} = M + 2M^2 + 3M^3 + (M - 1) \lceil\log_2(SC + 1) \rceil$ binary variables, and the second part of the statement is a consequence of Proposition \ref{prop:Physical_Qubits}.
\end{proof}
\begin{corollary}
	\label{korollar:domp}
	Let the penalty parameter in \eqref{DOMP-QUBO} be $P = \sum_{i = 1}^M\sum_{k = 1}^M\sum_{j = 1}^M\lambda_ic_{kj} + 1$. Then, any global minimizer of \eqref{DOMP-QUBO} yields a global minimizer of \eqref{DOMP}.
\end{corollary}
\begin{proof}
	For the DOMP, the objective function of its equivalent formulation of the form \eqref{BP}, used to construct the QUBO \eqref{DOMP-QUBO}, is $\sum_{i = 1}^M \sum_{k = 1}^M\sum_{j = 1}^M\lambda_i c_{kj} u_{ikj}$. Hence, the statement follows from Proposition \ref{theorem:PenPos}.
\end{proof}
\subsection{Discrete Multi-Period Facility Location Problem}
\label{sec:QUBO_DMPFLP}
We consider the $p$-Median Problem, and extend it to a multi-period setting \cite{LocationScience:ChapterMultiPeriod}. Consider a set of demand nodes $J$ that must be served over a finite planning horizon $T$. Facilities can be located at a subset of these nodes, denoted by $I \subseteq J$. In each period, exactly $p$ facilities are located at a subset of these nodes, denoted by $I \subseteq J$. The objective is to determine in every period where to locate the facilities and how to assign demand nodes to them so that the total cost of operation, opening, and closing is minimized.
\newline
The model parameters are defined as follows. $c_{ijt}$ denotes the cost of allocating demand node $j \in J$ to facility $i \in I$ in period $t \in T$. $g_{it}$ and $h_{it}$ represent, respectively, the opening and closing cost of facility $i \in I$ in period $t \in T$. $m_t$ specifies the maximum number of facilities that can be opened in each period $t \in T$. When $I = J$, the model reduces to the classical multi-period $p$-Median Problem. We introduce the following decision variables:
\begin{itemize}
	\item $x_{ijt} = 1$, if demand node $j \in J$ is served by facility $i \in I$ in period $t \in T$, and $0$ otherwise. In particular, $x_{iit} = 1$ indicates that a facility operates at site $i$ in period $t$.
	\item $z'_{it} \in \{0,1\}$: $1$, if a facility is opened at location $i \in I$ in period $t \in T$, and $0$ otherwise.
	\item $z''_{it} \in \{0,1\}$: $1$, if a facility is closed at location $i \in I$ in period $t \in T$, and $0$ otherwise.
\end{itemize}
The Discrete Multi-Period Facility Location ($p$-Median) Problem (DMPFLP), introduced by \cite{WesolowskyTruscottMultiPeriod}, is given by the following optimization problem:
\begin{subequations}\label{DMPFLP} 
	\begin{align}
		\displaystyle\min_{x, z', z''} \quad & \sum_{t \in T} \sum_{i \in I} \sum_{j \in J} c_{ijt}x_{ijt} + \sum_{t \in T} \sum_{i \in I} g_{it}z_{it}' + \sum_{t \in T} \sum_{i \in I} h_{it} z_{it}'' \\
		\mbox{s.t.} \quad & \sum_{i \in I} x_{ijt} = 1, \quad \quad \forall t \in T, j \in J, \\
		& \sum_{j \in J} x_{ijt} \leq \lvert J \rvert x_{iit},  \quad \quad  t \in T, i \in I, \\
		& \sum_{i \in I} x_{iit} = p, \quad \quad  \forall t \in T, \\
		& \sum_{i \in I}z_{it}' \leq m_t, \quad \quad  \forall t \in T, \\
		& x_{iit} - x_{ii,t-1} + z_{i, t-1}'' - z_{it}' = 0, \,  \forall t \in T \backslash \{1\}, i \in I, \\
		& x_{ijt} \in \{0,1\},  \quad \quad  \forall t \in T, i \in I, j \in J, \\
		& z_{it}', z_{it}'' \in \{0,1\}, \quad \quad  \forall t \in T, i \in I.
	\end{align}
\end{subequations}
In this formulation, facilities can be opened or closed at the beginning and end of each period, respectively. The parameter $m_t$ limits the number of facilities that can be opened during period $t$, while the binary variables $z'_{it}$ and $z''_{it}$ indicate opening and closing decisions.
\newline
This formulation allows facilities to be opened and closed multiple times across the planning horizon, a feature that may be unrealistic in many applications. We refer to \cite{LocationScience:ChapterMultiPeriod} for further modeling extensions.
\newline
\newline
\textbf{QUBO-formulation}:
To reformulate the problem as a QUBO, we first transform the second constraint into equality constraints. We introduce integer slack variables $y_{it} \in \{0, ..., \lvert J \rvert \}$ for all $t \in T, i \in I$. Each slack variable $y_{it}$ is expressed through a logarithmic binary encoding, $y_{it} = \sum_{k = 1}^{\lceil \log_2(\lvert J \rvert + 1) \rceil}2^{k-1}y_{itk}$, and $y_{itk} \in \{0,1\}$ for $k = 1,..., \lceil \log_2(\lvert J \rvert + 1) \rceil$. Similarly, we introduce an integer slack variable $u_t \in \{0,...,m_t\}$ for each $t \in T$ to transform the fourth constraint into an equality constraint. Each slack varaible $u_t$ is again represented via logarithmic encoding, $u_t = \sum_{k = 1}^{\lceil \log_2(m_t + 1) \rceil} 2^{k-1}u_{tk},$ with $u_{tk} \in \{0,1\}$, for $k = 1,...,\lceil \log_2(m_t + 1) \rceil$.
\newline
\newline
Altogether, these transformations yield the following cost Hamiltonian, which represents the QUBO formulation of \eqref{DMPFLP}:
\begin{subequations}\label{DMPFLP-QUBO}
	\begin{align}
		H = & \sum_{t \in T} \sum_{i \in I} \sum_{j \in J} c_{ijt}x_{ijt} + \sum_{t \in T} \sum_{i \in I} g_{it}z_{it}' + \sum_{t \in T} \sum_{i \in I} h_{it} z_{it}'' \\
		& + P \sum_{t \in T} \sum_{j \in J} \left(1 - \sum_{i \in I}x_{ijt} \right)^2 \\
		& + P  \sum_{t \in T} \sum_{i \in I} \left(\lvert J \rvert x_{iit} - \sum_{j \in J} x_{ijt} - \sum_{k = 1}^{\lceil \log_2(\lvert J \rvert + 1) \rceil}2^{k-1}y_{itk} \right)^2 \\
		& + P  \sum_{t \in T} \left(p - \sum_{i \in I}x_{iit} \right)^2 \\
		& + P  \sum_{t \in T} \left(m_t - \sum_{k = 1}^{\lceil \log_2(m_t + 1) \rceil} 2^{k-1}u_{tk} - \sum_{i \in I}z_{it}' \right)^2 \\
		& + P \sum_{t \in T\backslash \{1\}}\sum_{i \in I} (x_{iit} - x_{ii,t-1} + z_{i, t-1}'' - z_{it}')^2.
	\end{align}
\end{subequations}
\begin{lemma}
	The DMPFLP can be formulated as a QUBO with $n^{DMPFLP} = \lvert I \rvert \lvert J \rvert \lvert T \rvert + 2\lvert I \rvert \lvert T \rvert + \lvert I \rvert \lvert T \rvert \lceil \log_2(\lvert J \rvert + 1) \rceil + \sum_{t \in T}\lceil \log_2(m_t + 1) \rceil $ binary variables. Therefore, the QUBO formulation \eqref{DMPFLP-QUBO} can be solved on a gate-based Quantum Computer with $n^{DMPFLP}$ physical qubits via QAOA or WS-QAOA.
\end{lemma}
\begin{proof}
	For the QUBO formulation \eqref{DMPFLP-QUBO}, the binary variables $x_{ijt}$ for $t \in T, i \in I, j \in J$, $z_{it}'$ for $t \in T, i \in I$, and $z_{it}'' \in \{0,1\}$ for $t \in T, i \in $, together with the auxiliary binary variables $y_{itk}$ for $t \in T, i \in I, k = 1,...,\lceil \log_2(\lvert J \rvert + 1) \rceil$ and $u_{tk}$ for $t \in T, k = 1,..., \lceil \log_2(m_t + 1) \rceil$ are used. The second part of the statement follows from Proposition \ref{prop:Physical_Qubits}.
\end{proof}
\begin{corollary}
	\label{korollar:dmpflp}
	Let the penalty parameter in \eqref{DMPFLP-QUBO} be $P = \sum_{t \in T} \sum_{i \in I}\sum_{j \in J}c_{ijt} + \sum_{t \in T} \sum_{i \in I}g_{it} + \sum_{t \in T} \sum_{i \in I}h_{it} + 1$. Then, any global minimizer of \eqref{DMPFLP-QUBO} yields a global minimizer of \eqref{DMPFLP}.
\end{corollary}
\begin{proof}
	For the DMPFLP, the objective function of the equivalent formulation of the form  \eqref{BP}, used to construct the QUBO \eqref{DMPFLP-QUBO} is $\sum_{t \in T} \sum_{i \in I} \sum_{j \in J} c_{ijt}x_{ijt} + \sum_{t \in T} \sum_{i \in I} g_{it}z_{it}' + \sum_{t \in T} \sum_{i \in I} h_{it} z_{it}''$. Hence, the statement follows from Corollary \ref{theorem:PenPos}.
\end{proof}


\section{Numerical Experiments}
\label{sec:Numerical_Experiments}
This section presents a comprehensive computational study of the proposed QUBO formulations for the $p$-Median Problem, \eqref{p-Median-QUBO}, and the Fixed-Charge Facility Location Problem, \eqref{FCFLP1-QUBO} and \eqref{FCFLP2-QUBO}, evaluating the performance of both quantum and classical solution approaches. All computations were conducted on a machine equipped with an Intel Core Ultra 7 225U processor (2.00 GHz) and 32GB of RAM. The implementations were carried out in Python 3.11.6.
\newline
\newline
In addition to the quantum algorithms QAOA (Section \ref{sec:Preliminaires_QAOA}) and WS-QAOA (Section \ref{sec:Preliminaries_WS_QAOA}), we applied classical heuristics to the corresponding QUBO instances to provide performance benchmarks. Specifically, we considered the following classical heuristics:
\begin{enumerate}
	\item \textbf{Simulated Annealing}: Simulated Annealing is a stochastic metaheuristic inspired by the physical annealing process in metallurgy,
	introduced by \cite{KirkppatrickSimulatedAnnealing}. Starting from a random initial solution, the algorithm iteratively explores neighboring solutions. Moves that improve the objective are always accepted, whereas worsening moves are accepted with a probability that decreases over time, enabling the algorithm to escape local minima. We employed the D-Wave implementation \cite{DWaveSampler} for QUBOs.
	\item \textbf{Tabu Search}: Tabu Search is a deterministic local-search heuristic, introduced in \cite{Glover:TabuSearch}. It explores the solution space while maintaining a dynamic tabu list that prevents cycling by prohibiting recently visited solutions. Starting from an initial solution, at each iteration, the algorithm selects the best neighbor not prohibited by the tabu list, thereby enabling exploration beyond local optimality. We used the implementation provided via D-Wave \cite{DWaveSampler} for QUBOs.
\end{enumerate}
We use the following notation to refer to the individual methods and parameter settings:
\begin{itemize}
	\item \texttt{Gurobi}: Exact solution of the original mixed-integer (nonlinear) programming formulation using \texttt{Gurobi 12.0.3} \cite{Gurobi} via \texttt{gurobipy 12.0.3}.
	\item \texttt{SA-20/100/500}: Simulated Annealing with parameter \texttt{num\_reads} $= 20, 100,$ and $500$, respectively.
	\item \texttt{Tabu-0/50/250}: Tabu Search with parameter \texttt{num\_restarts} $= 0, 50$, and $250$, respectively.
	\item \texttt{QAOA-1/2/3}: Qiskit QAOA with one, two, or three layers using Qiskit's \texttt{AerSimulator}.
	\item \texttt{WS-QAOA-1/2/3}: Warm-Start QAOA with one, two, or three layers using Qiskit's \texttt{AerSimulator}.
	\item \texttt{Qrisp-1/2/3}: QAOA with one, two, or three layers, implemented using \texttt{Qrisp}'s internal \texttt{solve\_QUBO} method and simulated using Qiskit's \texttt{AerSimulator}.
\end{itemize}
All quantum experiments were implemented using the IBM software package Qiskit \cite{IBMQuantum, qiskit2024}, specifically \texttt{qiskit 2.2.3} \cite{qiskit2024}, together with \texttt{qiskit-optimization 0.7.0}, \texttt{qiskit-algorithms 0.4.0}, and \texttt{qiskit-aer 0.17.2} as well as \texttt{Qrisp} \cite{qrisp} via \texttt{qrisp 0.7.13}. For the methods \texttt{QAOA-1/2/3} and \texttt{WS-QAOA-1/2/3} we used the classical optimizer \texttt{COBYLA} \cite{Powell:COBYLA, SciPy, Zhang_2023} with \texttt{maxiter} $=50$ and \texttt{tol} $= 0.0001$. The \texttt{Qiskit PassManager} was configured with \texttt{optimization\_level} = 1. 
\newline
\newline
Because the objective of this computational study is to evaluate the algorithmic properties of the proposed QUBO formulations and the quantum optimization algorithms, we use the Qiskit \texttt{AerSimulator} in an ideal, noise-free configuration. This allows us to isolate the intrinsic performance of the quantum algorithms without confounding hardware noise, which currently limits the size and reliability of experiments on real quantum devices \cite{PreskillNISQ, Pellow_Jarman_2024_Noise, Quek2024_Noise}. Consequently, the results should be interpreted as proof-of-concept demonstrations rather than hardware-dependent performance benchmarks.
We focus primarily on the quality of the sampled solutions, specifically, their feasibility, their ratio of the obtained objective value to the global optimum, and the probability of sampling high-quality or optimal solutions.
\newline
\newline
All quantum algorithms were executed with \texttt{shots} $= 8000$, meaning that each quantum circuit was run and measured 8000 times. For \texttt{QAOA-1/2/3}, we used the \texttt{QAOA} method provided by \texttt{qiskit-optimization}. For \texttt{WS-QAOA-1/2/3}, we also used the \texttt{QAOA} method provided by \texttt{qiskit-optimization}, but manually constructed the initial state and mixer quantum circuit following \cite{Egger2021warmstartingquantum} and the procedure in Section \ref{sec:Preliminaries_WS_QAOA}. The parameter $\varepsilon$, controlling the radius of the $\ell_{\infty}$-norm ball around $(\frac{1}{2},...,\frac{1}{2})$, was set to $\varepsilon = 0.1$, corresponding to a radius of $0.4$. 
\newline
For \texttt{Qrisp-1/2/3}, we used the internal \texttt{solve\_QUBO} method from \texttt{Qrisp} using \texttt{qrisp 0.7.13} \cite{qrisp}. As backend, we selected \texttt{QiskitBackend} with the Qiskit \texttt{AerSimulator}, and we set \texttt{max\_iter} $= 50$.
\newline
\newline
The \texttt{AerSimulator} backends for \texttt{QAOA-1/2/3} and \texttt{WS-QAOA-1/2/3} were instantiated using the \texttt{SamplerV2} from \texttt{qiskit-aer} interface, with \texttt{seed} $= 123$, ensuring reproducibility and consistent sampling across experiments.
For both methods \texttt{QAOA-1/2/3} and \texttt{WS-QAOA-1/2/3}, the initial point for the variational parameters is $\beta_i = \frac{\pi}{4}$ and $\gamma_i = \frac{\pi}{8}$ for $i = 1,...,p$, where $p$ denotes the number of layers.
\newline
\newline
Both classical heuristics, Simulated Annealing and Tabu Search, were applied directly to the QUBO formulations using \texttt{dimod 0.12.21}, \texttt{dwave-system 1.34.0}, and the \texttt{dwave-ocean-sdk 9.1.0} \cite{DWaveSampler}. 
\newline
\newline
In the course of the numerical investigation of the quantum algorithm WS-QAOA, we introduce several warm-start strategies in Section \ref{sec:Exp_WarmStarting}. There, we consider optimization problems of the following types. First, we address a nonlinear optimization problem with simple box constraints; that is, we minimize a function that is at least twice continuously differentiable over a set defined solely by lower and upper bounds on the decision variables. Second, we solve semidefinite programs (SDPs). For the aforementioned type of nonlinear optimization problem, we use the \texttt{L-BFGS-B} algorithm \cite{LBFGSB1, LBFGSB2} from \texttt{SciPy} \cite{SciPy} via \texttt{scipy 1.15.3}. The SDPs are modeled using \texttt{Picos 2.6.2} \cite{PICOS}, and we use the commercial solver \texttt{MOSEK 11.0.30} \cite{mosek} to compute their solutions. For both solvers, we use the respective default settings.
\newline
\newline
Each QUBO instance was provided in the canonical form $x^TQx + c$, where $Q$ denotes the upper-triangular coefficient matrix and $c$ the scalar offset. Linear coefficients $q^Tx$ in \eqref{QUBO} were incorporated in the diagonal entries of $Q$, using $q_ix_i = Q_{ii}x_i^2 = Q_{ii}x_i$, since $x_i^2 = x_i$ for binary $x_i\in \{0,1\}$. The penalty parameters were selected as described in the Corollaries \ref{korollar:p_median_penalty}, and \ref{korollar:fcflp} in Section \ref{sec:QUBO_Formulations}.

\subsection{Test instances}
\label{sec:Exp_TestInstances}
Since we consider only very small test instances with $I = \{1,...,n\},$ $n = 3$ and $n = 4$ for the $p$-Median Problem, and $n = 3$ for the FCFLP, we directly report the complete instance data. All instances are constructed according to the prinicple that assignments costs $c_{ij} \geq 0$ for $i \neq j$ may take values in a large range, whereas self-service costs $c_{ii} \geq 0$ are substantially smaller or even zero, reflecting the reduced cost when a facility serves a client at the same location.
\newline 
The full numerical data for all test instances are provided in the Tables \ref{tab:data_p_median_3}, \ref{tab:data_p_median_4}, and \ref{tab:data_fcflp_3}. In these tables, vectors and matrices are represented using bracket notation. A vector $v \in \R^n$ is written as $[v_1,...,v_n]$. A matrix $A \in \R^{m \times n}$ is represented as a nested list of its rows, $[[a_{11},...,a_{1n}],[a_{21},...,a_{2n}],...,[a_{m1},...,a_{mn}]]$, where the $i$-th inner list corresponds to the $i$-th row of the matrix.
\begin{table}[ht]
	\centering
	\caption{$p$-Median Problem: data for $n = 3$}
	\label{tab:data_p_median_3}
	\begin{tabular}{@{}lll@{}}
		\toprule
		Instance & Demand $d\quad$  & Assignment cost $c$ \\
		\midrule
		1 & $[9,8,4]$ &  $[[1,9,1],[8,0,2],[8,7,0]]$ \\
		2 & $[9,7,4]$ & $[[1,5,3],[2,0,6],[9,4,0]]$\\
		3 & $[3,5,3]$ & $[[1,2,4],[6,0,3],[4,3,0]]$ \\
		4 & $[7,7,8]$ & $[[1,6,3],[3,1,9],[5,2,0]]$  \\
		5 & $[7,3,5]$ & $[[0,5,9],[4,1,9],[7,3,1]]$ \\
		6 & $[4,4,9]$ & $[[1,4,2],[3,1,6],[6,1,1]]$ \\
		7 & $[8,6,6]$ & $[[1,5,4],[5,1,6],[8,5,0]]$  \\
		8 & $[9,3,5]$ & $[[1,7,5],[4,1,7],[9,4,1]]$  \\
		9 & $[6,5,7]$ & $[[1,6,5],[3,1,2],[7,2,1]]$  \\
		10 & $[8,5,7]$ & $[[0,1,9],[2,0,1],[5,5,0]]$ \\
		\bottomrule
	\end{tabular}
\end{table}
\begin{table}[ht]
	\centering
	\caption{$p$-Median Problem: data for $n = 4$}
	\label{tab:data_p_median_4}
	\begin{tabular}{@{}lll@{}}
		\toprule
		Instance & Demand $d$  & Assignment cost $c$ \\
		\midrule
		1 & $[4,4,13,11]$ &  $[[2,11,13,6], [14,0,15,11], [5,14,1,6], [5,12,15,2]]$ \\
		2 & $[6,15,8,12]$ &  $[[0,15,3,2], [2,1,5,7], [3,13,1,6], [5,3,15,2]]$\\
		3 & $[7,11,12,6]$ &  $[[1,9,9,6], [3,1,12,12], [10,11,1,15], [3,3,6,2]]$\\
		4 & $[7,9,4,4]$ &  $[[0,5,3,11], [12,1,14,11], [16,5,1,3], [13,4,12,1]]$  \\
		5 & $[9,11,15,8]$ &  $[[1,11,6,11], [11,1,15,10], [10,2,2,10], [11,12,9,2]]$ \\
		6 & $[10,12,13,16]$ &  $[[0,9,12,16], [16,0,4,2], [6,12,0,10], [10,2,5,0]]$ \\
		7 & $[9,7,9,15]$ &  $[[2,8,7,4], [10,2,8,12], [9,15,2,16], [9,13,16,0]]$  \\
		8 & $[16,10,16,12]$ &  $[[2,3,8,16], [10,1,14,3], [2,12,2,9], [4,10,12,2]]$  \\
		9 & $[8,10,7,15]$ &  $[[1,2,15,15], [11,1,7,11], [10,2,0,13], [16,7,13,2]]$  \\
		10 & $[6,8,11,5]$ &  $[[2,7,13,2], [15,0,3,8], [6,12,1,2], [2,7,13,2]]$ \\
		\bottomrule
	\end{tabular}
\end{table}
\begin{table}[ht]
	\centering
	\caption{FCFLP: data for $n = 3$}
	\label{tab:data_fcflp_3}
	\begin{tabular}{@{}lllll@{}}
		\toprule
		Instance & Demand $d$  & Assignment cost $c$ & Fixed cost $f$ & Capacity $q$ \\
		\midrule
		1 & $[3,8,10]$ &  $[[0,4,4],[9,0,2],[5,5,0]]$ & $[25,9,17]$ & $[12,10,10]$\\
		2 & $[7,10,6]$ &  $[[0,2,7],[1,0,6],[3,7,0]]$ & $[17,17,1]$ & $[10,13,12]$  \\
		3 & $[3,8,9]$ &  $[[1,3,5],[7,1,4],[3,8,1]]$ & $[16,9,2]$ & $[12,13,11]$ \\
		4 & $[8,10,3]$ &  $[[0,1,9],[9,0,5],[2,4,0]]$ & $[9,2,3]$ & $[15,10,16]$ \\
		5 & $[6,5,8]$ &  $[[1,4,4],[8,0,6],[8,8,1]]$ & $[10,17,26]$ & $[8,13,9]$ \\
		6 & $[8,3,3]$ &  $[[0,2,4],[2,0,8],[7,8,1]]$ & $[8,23,11]$ & $[13,9,10]$ \\
		7 & $[6,10,7]$ &  $[[1,5,4],[3,1,1],[8,7,1]]$ & $[16,2,23]$ & $[10,14,12]$ \\
		8 & $[4,8,10]$ &  $[[0,4,5],[2,1,2],[2,9,1]]$ & $[26,13,11]$ & $[13,14,10]$\\
		9 & $[4,10,3]$ &  $[[1,6,4],[8,1,4],[9,6,1]]$ & $[6,19,1]$ & $[14,10,10]$ \\
		10 &$[7,5,9]$ &  $[[1,5,1],[7,1,8],[8,8,0]]$ & $[2,6,15]$ & $[13,10,9]$\\
		\bottomrule
	\end{tabular}
\end{table}
\newline
According to Lemma \ref{lemma:p-median_qubits} on the number of binary variables in the QUBO formulation \eqref{p-Median-QUBO} of the $p$-Median Problem, 12 qubits are required for $n = 3$ and 20 qubits for $n = 4$. Similarly, by Lemma \ref{lemma:fcflp_qubits} concerning the number of binary variables in the QUBO formulations \eqref{FCFLP1-QUBO} and \eqref{FCFLP2-QUBO} of the FCFLP, 24 qubits are needed for all test instances listed in Table \ref{tab:data_fcflp_3}, except for test instance 4, which requires 25 qubits.
\subsection{Warm-Starting}
\label{sec:Exp_WarmStarting}
In this section, we describe the procedure for generating warm-starting points for WS-QAOA (\texttt{WS-QAOA-1/2/3}). The general idea follows the same approach outlined in Section \ref{sec:Preliminaries_WS_QAOA} and \cite{Egger2021warmstartingquantum}: we first compute a \textit{good} point, and then project this point onto the $\ell_{\infty}$-norm ball centered at $(\frac{1}{2},...,\frac{1}{2})$ with an appropriate radius.
\newline
In the following, we present four strategies for warm-starting the WS-QAOA algorithm. The first two constitute classical warm-start strategies introduced by \cite{Egger2021warmstartingquantum}. The remaining two strategies, based on the LP relaxation of the FCFLP, are novel to the best of our knowledge. Each of these strategies allows us to compute a \textit{good} point, which is then projected onto the $\ell_{\infty}$-norm ball, as specified in the numerical experiment setting above.
\subsubsection{R: Continuous relaxation of QUBO}
\label{sec:WarmStarting_R}
The first warm-start strategy consists of considering the QUBO $\min_{x \in \{0,1\}^n} x^TQx + c$ and its continuous relaxation $\min_{x \in [0,1]^n} x^TQx + c$ \cite{Egger2021warmstartingquantum}. In general, this relaxation yields a nonconvex continuous optimization problem, implying that only a local minimum can be computed efficiently. To obtain a local minimum, we apply the \texttt{L-BFGS-B} method as described above. Here, we use as starting point for the \texttt{L-BFGS-B} algorithm the point $(\frac{1}{2},...,\frac{1}{2}) \in \R^n$. In the following, we append the suffix \texttt{-R} to \texttt{WS-QAOA-1/2/3} to denote this warm-start variant.
\subsubsection{S: Semidefinite programming relaxation of QUBO as quadratically constrained quadratic program}
\label{sec:WarmStarting_S}
For the second warm-start strategy, we first reformulate the QUBO $\min_{x \in \{0,1\}^n} x^TQx + c$ , where $Q$ is a symmetric real-valued matrix, as a \textit{quadratically constrained quadratic program (QCQP)}. The QUBO is equivalent to the following QCQP:
\begin{subequations} \label{QCQP}
	\begin{align}
		\min_{x \in \R^n} \quad & x^TQx + c \\
		\mbox{s.t.} \quad & x_i^2 - x_i = 0, &\forall i = 1,...,n.
	\end{align}
\end{subequations}
The corresponding standard semidefinite programming (SDP) relaxation is given by
\begin{subequations}\label{SDP-relax}
	\begin{align}
		\inf_{Y \in S^{n+1}} \quad & \big\langle \begin{pmatrix}
			c & 0 \\
			0 & Q
		\end{pmatrix}, Y \big\rangle_F \\
		\mbox{s.t.} \quad & Y_{11} = 1,  \\
		& \big\langle \begin{pmatrix}
			0 & -\frac{1}{2}e_i^T \\
			-\frac{1}{2}e_i & \diag(e_i)
		\end{pmatrix}, Y \big\rangle_F = 0, \quad \forall i = 1,...,n, \\
		& Y \in \mathcal{S}^{n+1}.
	\end{align}
\end{subequations}
Here, $S^{n+1}$ denotes the set of real-valued symmetric matrices in $\R^{(n+1) \times (n+1)}$, $\mathcal{S}^{n+1} \subset S^{n+1}$ denotes the \textit{postive semidefinite cone (PSD cone)}, $e_i \in \R^n$ is the $i$-th canonical basis vector, $\diag(v) \in \R^{n \times n}$ is the diagonal matrix with entries given by the vector $v \in \R^n$, and $\langle \cdot, \cdot \rangle_F$ denotes the Frobenius inner product. We refer to \cite{QCQP1} and \cite{QCQPShor} for further details on the SDP relaxation of a QCQP.
\newline
We solve the semidefinite program by using \texttt{Picos} and \texttt{MOSEK} as described above. From the optimal solution $X^* \in S^{n+1}$, we obtain the warm-start point by extracting the vector $x^* = (X_{1,2}^*,...,X_{1,n+1}^*)$. In the following, we append the suffix \texttt{-S} to \texttt{WS-QAOA-1/2/3} to denote this warm-start variant.
\subsubsection{L: Linear programming relaxation of integer program}
\label{sec:WarmStarting_L}
In this strategy, we compute a warm-start for the QUBO based on the integer program from which the QUBO is derived. 
\newline
\newline
\textbf{$p$-Median Problem: }For the $p$-Median Problem, we solve the linear programming relaxation via \texttt{Gurobi}.
\newline
\newline
\textbf{Fixed-Charge Facility Location Problem: }For the FCFLP, we first solve the LP relaxation of the aggregated or disaggregated formulation via \texttt{Gurobi}, yielding the solution vectors $\bar{x}$ and $\bar{y}$ with components in $[0,1]$. To incorporate account for the slack variables in the warm-start, we compute for each $i \in I$ the logarithmic binary encoding of the residual capacity, given by
\begin{align*}
	\sum_{k = 1}^{\lceil \log_2(q_i + 1) \rceil}2^{k-1}\bar{z}_{ik} = \lfloor q_i - \sum_{j \in J}d_j\bar{x}_{ij}\rfloor.
\end{align*}
The warm-start point is then formed by concatenating $(\bar{x}, \bar{y}, \bar{z})$. 
\newline
Naturally, this warm-start strategy is applicable only when the QUBO arises from an underlying integer program; however, this is precisely the setting in which most current research is conducted. Therefore, this LP-based warm-start strategy is transferable to general integer programs. We solve the LP relaxation using \texttt{Gurobi} as described above. In the following, we append the suffix \texttt{-L} to \texttt{WS-QAOA-1/2/3} to denote this warm-start variant.
\subsubsection{C: Combination of continuous relaxation and LP relaxation}
\label{sec:WarmStarting_C}
This strategy is a combined approach. We first compute a warm-start point using the LP-based strategy. That is, we solve the LP relaxation, then, for the FCFLP, compute the logarithmic binary representation of the residual capacity to construct a reasonable warm-start point including the slack variables, and use this point as the initial point for \texttt{L-BFGS-B}. The point at which \texttt{L-BFGS-B} terminates is then used as warm-start point. In the following, we append the suffix \texttt{-C} to \texttt{WS-QAOA-1/2/3} to denote this warm-start variant.

\subsection{Computational results}
\label{sec:Exp_Results}
To evaluate both QUBO formulation quality and algorithmic performance across classical heuristics and quantum approaches, we measure solution quality via feasibility, optimal solution probability, and approximation ratio.
\newline
\newline
Both quantum algorithms QAOA and WS-QAOA return a finite set of sampled solutions. In the spirit of a proof-of-concept computational study, we select the sample with the lowest objective value as the representative solution of the method, directly reflecting the intrinsic sampling behavior of the quantum algorithm with respect to the QUBO energy landscape. As our focus is on evaluating the algorithmic performance of the quantum algorithms and the underlying QUBO formulations, we do not apply feasibility-repair procedures or post-processing heuristics. Feasibility is therefore evaluated explicitly, and infeasible outcomes provide meaningful information about the sampling characteristics.
\newline
\newline
For each QUBO formulation and test instance set, we report one table and two types of plots that summarize the results across the ten test instances:
\begin{enumerate}
	\item Feasibility table: showing, for each method, the number of instances (out of ten) for which the obtained solution is feasible.
	\item Objective-value ratio plot: a line plot illustrating the relative objective quality across all test instances, restricted to feasible points. For each \texttt{METHOD} and instance $i = 1,...,10$, we compute
	\begin{align*}
		r_{\texttt{METHOD}}^i \coloneqq \dfrac{z_{\texttt{METHOD}}^i}{z_{\texttt{Gurobi}}^i},
	\end{align*}
	where $z_{\texttt{METHOD}}^i$ denotes the value obtained by evaluating the solution returned by \texttt{METHOD} in the original objective function and $z_{\texttt{Gurobi}}^i$ the globally optimal objective value computed via \texttt{Gurobi}. A value of $r_{\texttt{METHOD}}^i = 1$ indicates that the global optimum was obtained by the respective method on instance $i = 1,...,10$. Infeasible points are omitted and appear as white gaps in the corresponding lines.
	\item Relative frequency plot: a grouped bar plot showing, for each instance and quantum method, the relative frequency of sampling the solution that yields the lowest objective value of the QUBO. In contrast to the previous plot, infeasible samples are also included here to illustrate the underlying sampling behavior of the algorithms.
\end{enumerate}
Together, these plots provide a comprehensive qualitative assessment of each method's robustness, solution quality, and sampling characteristics across problem instances and QUBO formulations.
\subsubsection{$p$-Median Problem}
\label{sec:Exp_Results_p_median}
We consider the $p$-Median Problem \eqref{p-Median-problem} with $\lvert I \rvert = \lvert J \rvert = \{1,...,n\}$, where $n = 3, p = 1, 2$, and $n = 4, p = 2$, using the test instance data in Tables \ref{tab:data_p_median_3} and \ref{tab:data_p_median_4}.
\newline
A complete summary of the results is provided in Tables \ref{tab:p_median_3_1}, \ref{tab:p_median_3_2}, and \ref{tab:p_median_4_2} (feasibility tables); Figures \ref{fig:p_median_3_1}, \ref{fig:p_median_3_2}, and \ref{fig:p_median_4_2} (objective value ratio plots and relative frequency plots).
\begin{table}[h!]
	\centering
	\caption{Frequencies of feasible solutions for $n = 3$ and $p = 1$}
	\label{tab:p_median_3_1}
	\begin{tabular}{@{}l r l r l r l r@{}}
		\toprule
		\texttt{METHOD} & Count & \texttt{METHOD} & Count  & \texttt{METHOD} & Count &\texttt{METHOD} & Count   \\
		\midrule
		\texttt{SA-20} & 10  & \texttt{Tabu-0} & 10  &  \texttt{QAOA-1} & 10  &  \texttt{Qrisp-1} & 10  \\
		\texttt{SA-100} & 10 & 	\texttt{Tabu-50} & 10 & \texttt{QAOA-2} & 10 &   \texttt{Qrisp-2} & 10 \\
		\texttt{SA-500} & 10 &  \texttt{Tabu-250} & 10 &  \texttt{QAOA-3} & 10 & 	 \texttt{Qrisp-3} & 10   \\
		\texttt{WS-QAOA-1-R} & 9 & \texttt{WS-QAOA-1-S} & 10 & \texttt{WS-QAOA-1-L} & 10 & \texttt{WS-QAOA-1-C} & 10 
		\\
		\texttt{WS-QAOA-2-R} & 10 & \texttt{WS-QAOA-2-S} & 10 & \texttt{WS-QAOA-2-L} & 10 & \texttt{WS-QAOA-2-C} & 10 
		\\
		\texttt{WS-QAOA-3-R} & 10 & \texttt{WS-QAOA-3-S} & 10 & \texttt{WS-QAOA-3-L} & 10 & \texttt{WS-QAOA-3-C} & 10 
		\\
		\bottomrule
	\end{tabular}
\end{table}
\begin{figure}[h!]
	\centering
	\begin{subfigure}{0.48\linewidth}
		\centering
		\resizebox{\linewidth}{!}{\includegraphics{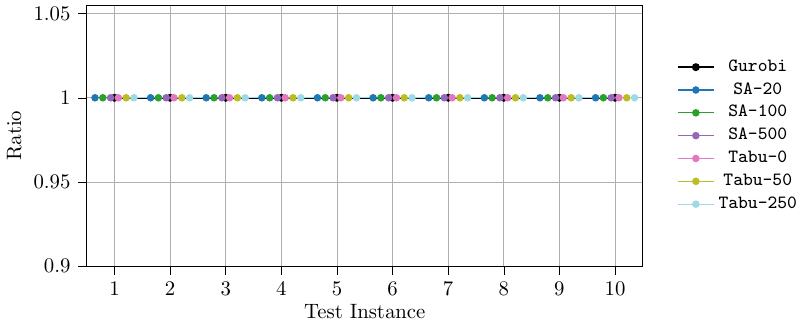}}
		\caption{Objective value ratios for classical heuristics}
		\label{fig:p_median_3_1_obj_val_h}
	\end{subfigure}
	\hfill
	\begin{subfigure}{0.48\linewidth}
		\centering
		\resizebox{\linewidth}{!}{\includegraphics{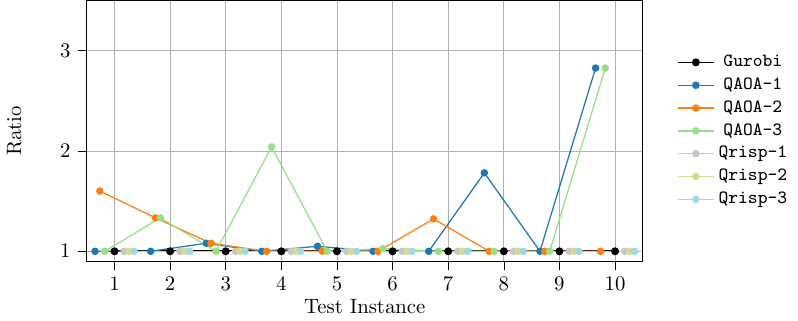}}
		\caption{Objective value ratios for quantum algorithms (QAOA)}
		\label{fig:p_median_3_1_obj_val_qaoa}
	\end{subfigure}
	
	\vspace{0.25cm}
	
	\begin{subfigure}{0.48\linewidth}
		\centering
		\resizebox{\linewidth}{!}{\includegraphics{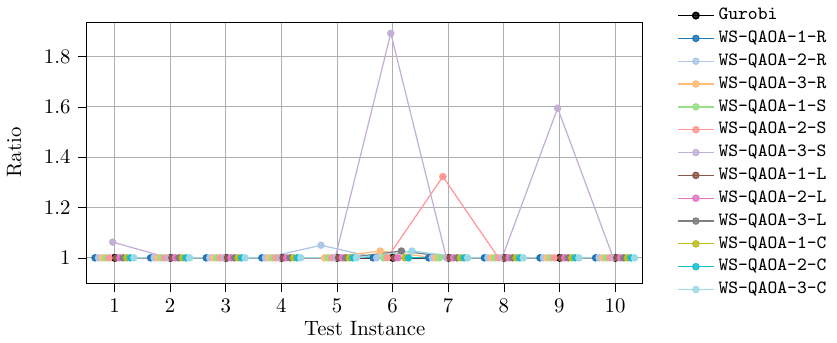}}
		\caption{Objective value ratios for quantum algorithms (WS-QAOA)}
		\label{fig:p_median_3_1_obj_val_ws_qaoa}
	\end{subfigure}
	\hfill
	\begin{subfigure}{0.48\linewidth}
		\centering
		\resizebox{\linewidth}{!}{\includegraphics{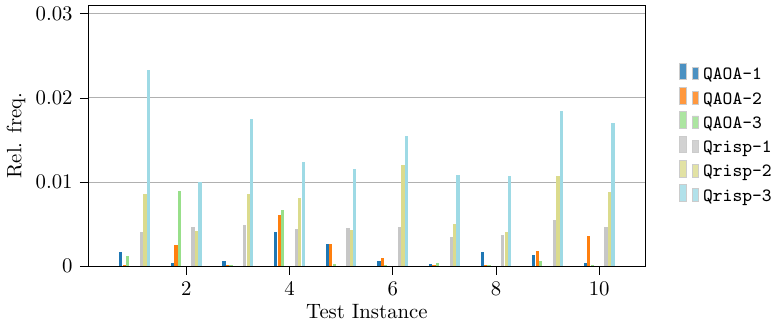}}
		\caption{Relative frequency plot (QAOA)}
		\label{fig:p_median_3_1_prob_qaoa}
	\end{subfigure}
	\vspace{0.25cm}
	\begin{subfigure}{0.48\linewidth}
		\centering
		\resizebox{\linewidth}{!}{\includegraphics{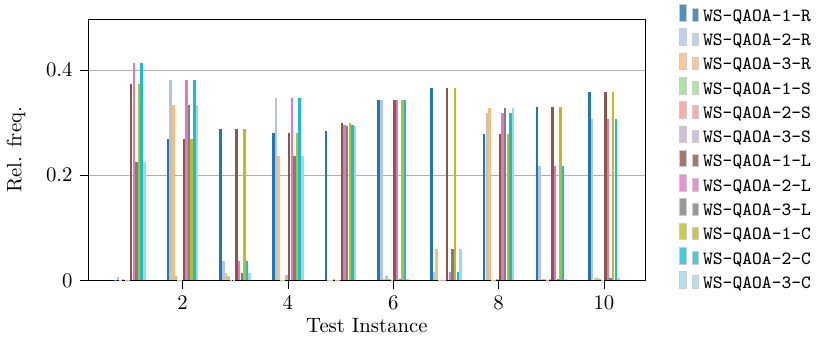}}
		\caption{Relative frequency plot (WS-QAOA)}
		\label{fig:p_median_3_1_prob_ws_qaoa}
	\end{subfigure}
	\caption{$p$-Median Problem ($n = 3$ and $p = 1$)}
	\label{fig:p_median_3_1}
\end{figure}
\begin{table}[h!]
	\centering
	\caption{Frequencies of feasible solutions for $n = 3$ and $p = 2$}
	\label{tab:p_median_3_2}
	\begin{tabular}{@{}l r l r l r l r@{}}
		\toprule
		\texttt{METHOD} & Count & \texttt{METHOD} & Count & \texttt{METHOD} & Count & \texttt{METHOD} & Count \\
		\midrule
		\texttt{SA-20} & 10 &  \texttt{Tabu-0} & 10 &    \texttt{QAOA-1} & 10 &  \texttt{Qrisp-1} & 10 \\
		\texttt{SA-100} & 10 &  \texttt{Tabu-50} & 10  & \texttt{QAOA-2} & 10  & \texttt{Qrisp-2} & 10  \\
		\texttt{SA-500} & 10 & 	\texttt{Tabu-250} & 10 &  \texttt{QAOA-3} & 10 &  \texttt{Qrisp-3} & 10  \\
		\texttt{WS-QAOA-1-R} & 10 & \texttt{WS-QAOA-1-S} & 10 & \texttt{WS-QAOA-1-L} & 10 & \texttt{WS-QAOA-1-C} & 10 
		\\
		\texttt{WS-QAOA-2-R} & 10 & \texttt{WS-QAOA-2-S} & 10 & \texttt{WS-QAOA-2-L} & 10 & \texttt{WS-QAOA-2-C} & 10 
		\\
		\texttt{WS-QAOA-3-R} & 10 & \texttt{WS-QAOA-3-S} & 10 & \texttt{WS-QAOA-3-L} & 10 & \texttt{WS-QAOA-3-C} & 10 
		\\
		\bottomrule
	\end{tabular}
\end{table}
\begin{figure}[h!]
	\centering
	
	\begin{subfigure}{0.48\linewidth}
		\centering
		\resizebox{\linewidth}{!}{\includegraphics{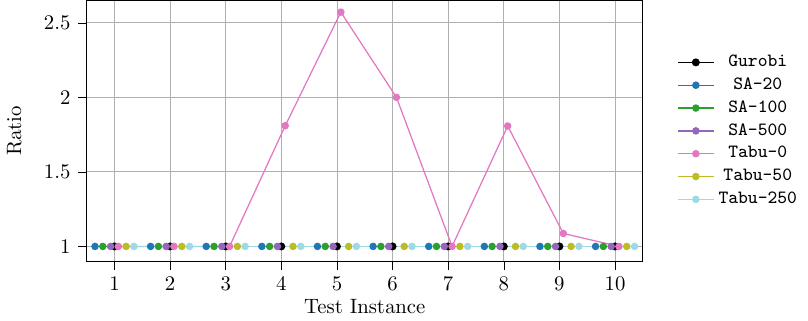}}
		\caption{Objective value ratios for classical heuristics}
		\label{fig:p_median_3_2_obj_val_h}
	\end{subfigure}
	\hfill
	\begin{subfigure}{0.48\linewidth}
		\centering
		\resizebox{\linewidth}{!}{\includegraphics{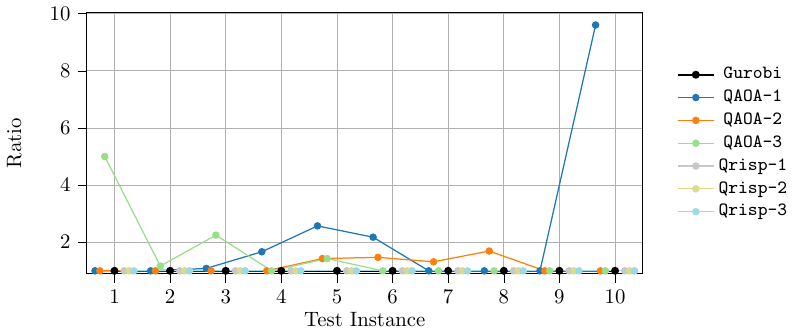}}
		\caption{Objective value ratios for quantum algorithms (QAOA)}
		\label{fig:p_median_3_2_obj_val_qaoa}
	\end{subfigure}
	
	\vspace{0.25cm}
	\begin{subfigure}{0.48\linewidth}
		\centering
		\resizebox{\linewidth}{!}{\includegraphics{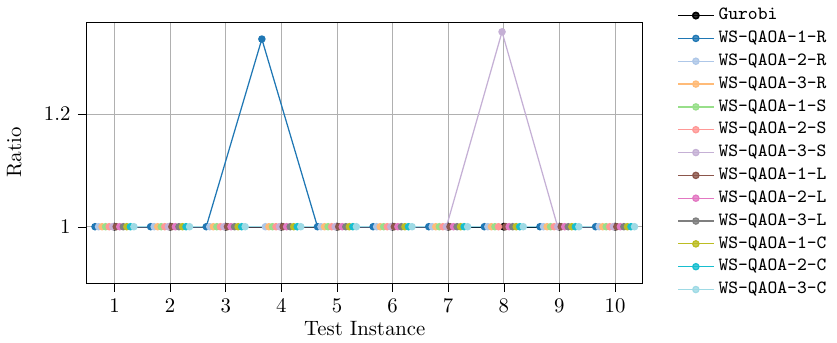}}
		\caption{Objective value ratios for quantum algorithms (WS-QAOA)}
		\label{fig:p_median_3_2_obj_val_ws_qaoa}
	\end{subfigure}
	\hfill	
	\begin{subfigure}{0.48\linewidth}
		\centering
		\resizebox{\linewidth}{!}{\includegraphics{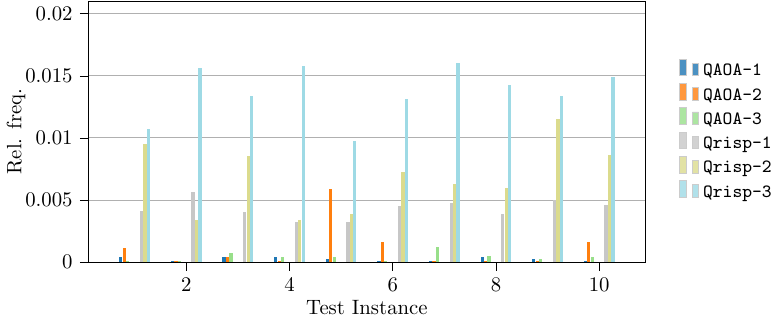}}
		\caption{Relative frequency plot (QAOA)}
		\label{fig:p_median_3_2_prob_ws_qaoa}
	\end{subfigure}
	\vspace{0.25cm}
	
	\begin{subfigure}{0.48\linewidth}
		\centering
		\resizebox{\linewidth}{!}{\includegraphics{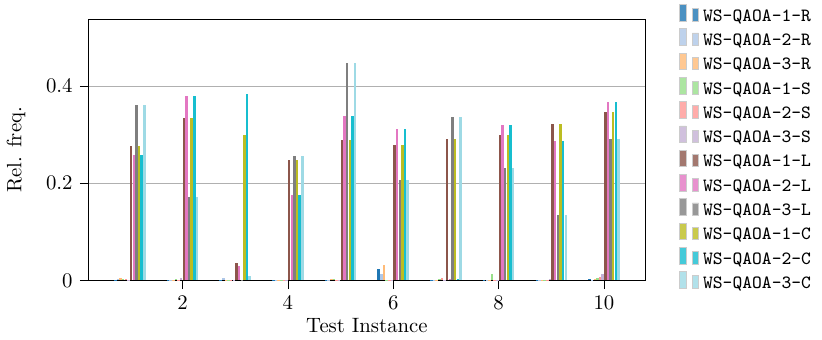}}
		\caption{Relative frequency plot (WS-QAOA)}
		\label{fig:p_median_3_2_prob_qaoa}
	\end{subfigure}
	\caption{$p$-Median Problem ($n = 3$ and $p = 2$)}
	\label{fig:p_median_3_2}
\end{figure}

\begin{table}[h!]
	\centering
	\caption{Frequencies of feasible solutions for $n = 4$ and $p = 2$}
	\label{tab:p_median_4_2}
	\begin{tabular}{@{}l r l r l r l r@{}}
		\toprule
		\texttt{METHOD} & Count &  \texttt{METHOD} & Count & \texttt{METHOD} & Count &\texttt{METHOD} & Count   \\
		\midrule
		\texttt{SA-20} & 10 &  \texttt{Tabu-0} & 10 &  \texttt{QAOA-1} & 6 &   \texttt{Qrisp-1} & 10 \\
		\texttt{SA-100} & 10 &  \texttt{Tabu-50} & 10 & \texttt{QAOA-2} & 4 &  \texttt{Qrisp-2} & 10  \\
		\texttt{SA-500} & 10 &  \texttt{Tabu-250} & 10 &  \texttt{QAOA-3} & 5 &  \texttt{Qrisp-3} & 10  \\
		\texttt{WS-QAOA-1-R} & 2 & \texttt{WS-QAOA-1-S} & 9 & \texttt{WS-QAOA-1-L} & 10 & \texttt{WS-QAOA-1-C} & 10 
		\\
		\texttt{WS-QAOA-2-R} & 3 & \texttt{WS-QAOA-2-S} & 10 & \texttt{WS-QAOA-2-L} & 10 & \texttt{WS-QAOA-2-C} & 10 
		\\
		\texttt{WS-QAOA-3-R} & 7 & \texttt{WS-QAOA-3-S} & 10 & \texttt{WS-QAOA-3-L} & 10 & \texttt{WS-QAOA-3-C} & 10 
		\\
		\bottomrule
	\end{tabular}
\end{table}

\begin{figure}[h!]
	\centering
	
	\begin{subfigure}{0.48\linewidth}
		\centering
		\resizebox{\linewidth}{!}{\includegraphics{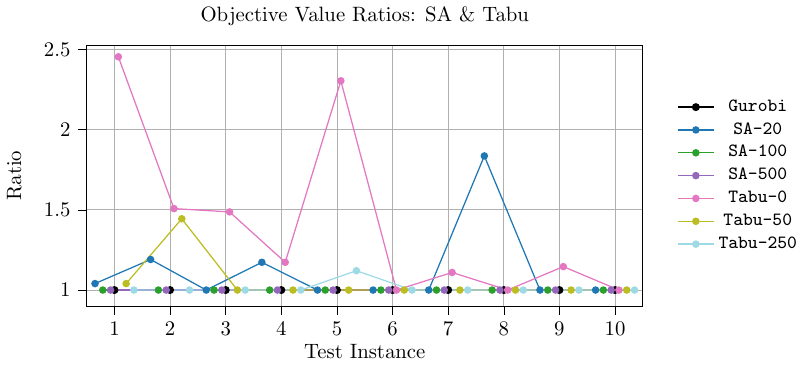}}
		\caption{Objective value ratios for classical heuristics}
		\label{fig:p_median_4_2_obj_val_h}
	\end{subfigure}
	\hfill
	\begin{subfigure}{0.48\linewidth}
		\centering
		\resizebox{\linewidth}{!}{\includegraphics{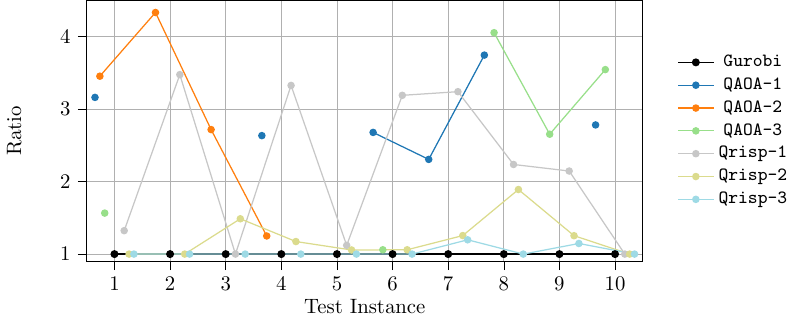}}
		\caption{Objective value ratios for quantum algorithms (QAOA)}
		\label{fig:p_median_4_2_obj_val_qaoa}
	\end{subfigure}
	
	\vspace{0.25cm}
	
	\begin{subfigure}{0.48\linewidth}
		\centering
		\resizebox{\linewidth}{!}{\includegraphics{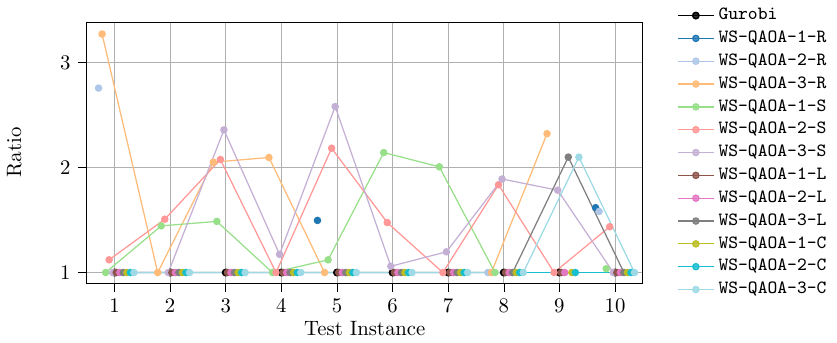}}
		\caption{Objective value ratios for quantum algorithms (WS-QAOA)}
		\label{fig:p_median_4_2_obj_val_ws_qaoa}
	\end{subfigure}
	\hfill

	\begin{subfigure}{0.48\linewidth}
		\centering
		\resizebox{\linewidth}{!}{\includegraphics{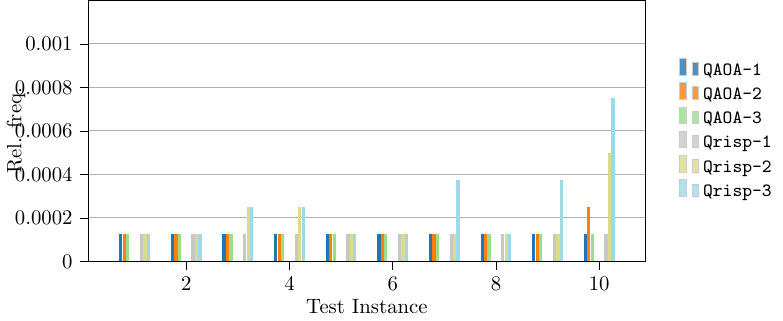}}
		\caption{Relative frequency plot (QAOA)}
		\label{fig:p_median_4_2_prob_qaoa}
	\end{subfigure}
	
	\vspace{0.25cm}

	\begin{subfigure}{0.48\linewidth}
		\centering
		\resizebox{\linewidth}{!}{\includegraphics{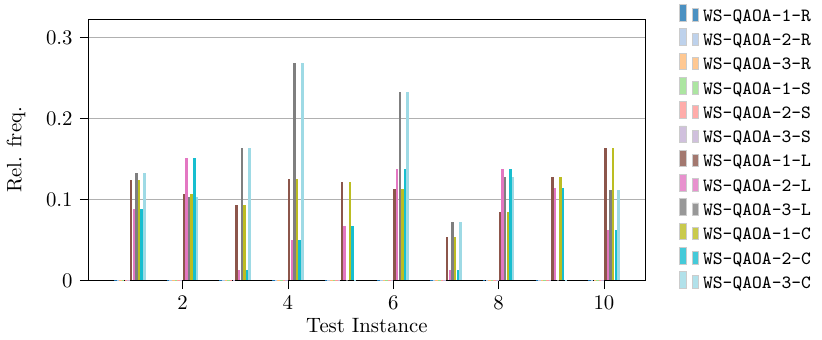}}
		\caption{Relative frequency plot (WS-QAOA)}
		\label{fig:p_median_4_2_prob_ws_qaoa}
	\end{subfigure}
	\caption{$p$-Median Problem ($n = 4$ and $p = 2$)}
	\label{fig:p_median_4_2}
\end{figure}
For $n = 3$, only \texttt{WS-QAOA-1-R} fails once to return a feasible point. Overall, the obtained solutions are of very high quality, whereas \texttt{QAOA-1/2/3} produce the worst solutions. We further observe a clear increase in the relative frequency for \texttt{Qrisp-1/2/3} with an increasing number of layers, exceeding the relative frequencies of \texttt{QAOA-1/2/3}, despite both approaches being based on standard QAOA. This difference likely arises from several implementation details, such as the initialization of the variational parameters and the stochastic behavior of the quantum simulator.
\newline
For WS-QAOA, the variants using continuous or SDP relaxations to compute warm-start points tend to produce slightly inferior solutions. In particular, the relative frequencies for \texttt{WS-QAOA-1/2/3-R} and \texttt{WS-QAOA-1/2/3-S} are noticeably lower than those of the LP-based variants.
\newline
For $n = 4$, the situation changes significantly. Classical heuristics begin to struggle in consistently identifying high-quality solutions. \texttt{SA-20/100/500} benefits from increasing the \texttt{num\_reads} parameter. Tabu Search shows a similar effect. The solutions obtained from WS-QAOA tend to yield better objective values than those obtained with standard QAOA methods. Notably, \texttt{Qrisp-3} again yields very high-quality solutions. Both \texttt{QAOA-1/2/3} and \texttt{WS-QAOA-1/2/3-R} now occasionally return infeasible points. For \texttt{WS-QAOA-1/2/3-R}, this behavior can be attributed directly to the poor quality of the computed warm-start point.
\newline
Once again, the LP-based WS-QAOA variants perform very well, returning the global optimum almost always.
\newline
We note that for all test instances of the $p$-Median Problem, the LP relaxation already produced integral solutions, solving the corresponding $p$-Median Problem to optimality. Consequently, the strong performance of WS-QAOA is not surprising, especially given that the parameter $\varepsilon = 0.1$, which specifies the $\ell_{\infty}$-norm ball used to construct the warm-start point, was chosen to be relatively small.
\newline
Nevertheless, although the underlying integer problems are simple, as the LP relaxation is integral, this property does not carry over to the QUBO formulation. Consequently, computing a warm-start point based on the integer programming formulation, when the QUBO is derived from it, is highly advantageous. We additionally note that, in this case, solving the QUBO is of little practical interest, as the problem has already been solved by the LP relaxation. These observations therefore primarily concern the algorithmic behavior of WS-QAOA and the associated warm-start strategies.
\newline
Subsequent experiments on the FCFLP illustrate the broader applicability of LP-based warm-start strategies proposed in Section \ref{sec:WarmStarting_L} (FCFLP) and \ref{sec:WarmStarting_C}. In particular, they deliver superior performance when the LP relaxations are not integral and even auxiliary variables are required to construct the associated QUBO formulation.


\subsubsection{Fixed-Charge Facility Location Problem}
\label{sec:Exp_Results_fcflp}
For the FCFLP, we consider both the aggregated formulation \eqref{FCFLP-1} and the disaggregated formulation \eqref{FCFLP-2}, with $\lvert I \rvert = \lvert J \rvert = \{1,...,n\}$ and $n = 3$, using the test instance data provided in Table \ref{tab:data_fcflp_3}.
\newline
A complete summary of the results is provided in Tables \ref{tab:fcflp1_3} and \ref{tab:fcflp2_3} (feasibility tables), and Figures \ref{fig:fcflp1_3} and \ref{fig:fcflp2_3} (objective value ratio plots and relative frequency plots).
\begin{table}[h!]
	\centering
	\caption{Aggregated formulation: frequencies of feasible solutions for $n =  3$}
	\label{tab:fcflp1_3}
	\begin{tabular}{@{}l r l r l r l r@{}}
		\toprule
		\texttt{METHOD} & Count &   \texttt{METHOD} & Count &  \texttt{METHOD} & Count &  \texttt{METHOD} & Count  \\
		\midrule
		\texttt{SA-20} & 9 & \texttt{Tabu-0} & 0 &  \texttt{QAOA-1} & 5 &  \texttt{Qrisp-1} & 2 \\
		\texttt{SA-100} & 10 &  \texttt{Tabu-50} & 4 & \texttt{QAOA-2} & 3  &  \texttt{Qrisp-2} & 3 \\
		\texttt{SA-500} & 10 &   \texttt{Tabu-250} & 4 & \texttt{QAOA-3} & 1 &   \texttt{Qrisp-3} & 4 \\
		\texttt{WS-QAOA-1-R} & 0 & \texttt{WS-QAOA-1-S} & 4 & \texttt{WS-QAOA-1-L} & 5 & \texttt{WS-QAOA-1-C} & 6 
		\\
		\texttt{WS-QAOA-2-R} & 0 & \texttt{WS-QAOA-2-S} & 4 & \texttt{WS-QAOA-2-L} & 5 & \texttt{WS-QAOA-2-C} & 5 
		\\
		\texttt{WS-QAOA-3-R} & 1 & \texttt{WS-QAOA-3-S} & 1 & \texttt{WS-QAOA-3-L} & 3 & \texttt{WS-QAOA-3-C} & 5 
		\\
		\bottomrule
	\end{tabular}
\end{table}
\begin{figure}[h!]
	\centering

	\begin{subfigure}{0.48\linewidth}
		\centering
		\resizebox{\linewidth}{!}{\includegraphics{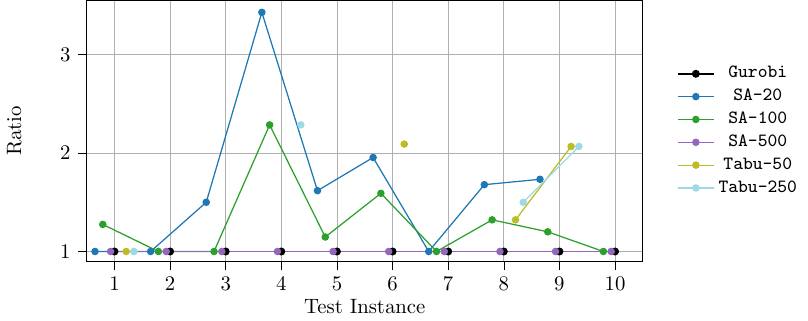}}
		\caption{Objective value ratios for classical heuristics}
		\label{fig:fcflp1_obj_val_h}
	\end{subfigure}
	\hfill
	\begin{subfigure}{0.48\linewidth}
		\centering
		\resizebox{\linewidth}{!}{\includegraphics{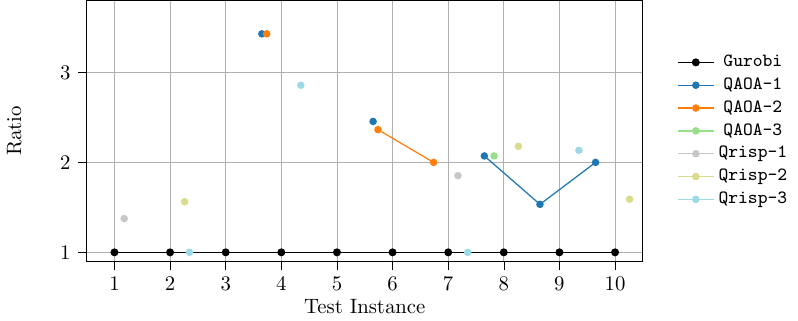}}
		\caption{Objective value ratios for quantum algorithms (QAOA)}
		\label{fig:fcflp1_3_obj_val_qaoa}
	\end{subfigure}
	
	\vspace{0.5cm}
	
	\begin{subfigure}{0.48\linewidth}
		\centering
		\resizebox{\linewidth}{!}{\includegraphics{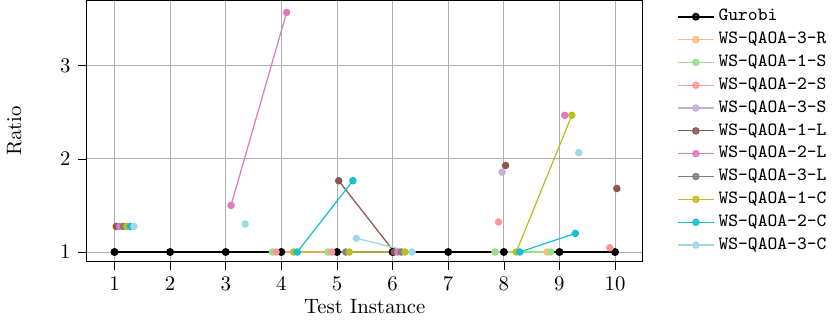}}
		\caption{Objective value ratios for classical heuristics (WS-QAOA)}
		\label{fig:fcflp1_obj_val_ws_qaoa}
	\end{subfigure}
	\hfill
	\begin{subfigure}{0.48\linewidth}
		\centering
		\resizebox{\linewidth}{!}{\includegraphics{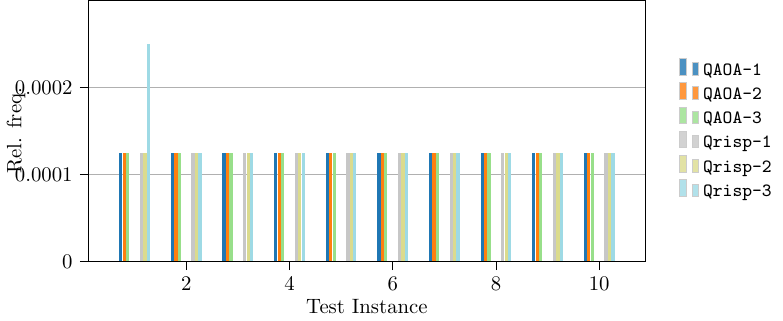}}
		\caption{Relative frequency plot (QAOA)}
		\label{fig:fcflp1_3_prob_qaoa}
	\end{subfigure}

	\begin{subfigure}{0.48\linewidth}
		\centering
		\resizebox{\linewidth}{!}{\includegraphics{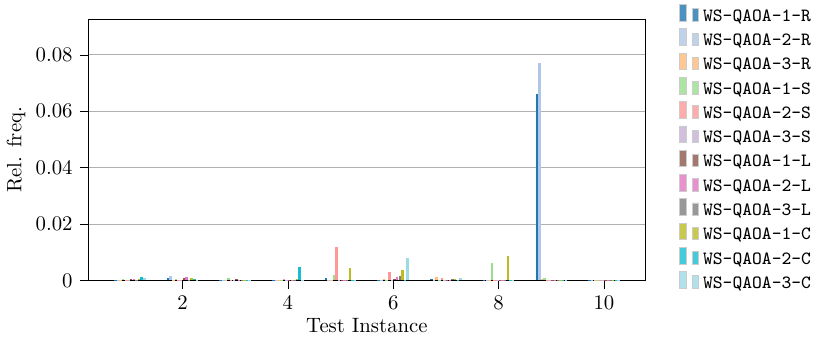}}
		\caption{Relative frequency plot (WS-QAOA)}
		\label{fig:fcflp1_3_prob_ws_qaoa}
	\end{subfigure}
	\caption{FCFLP (aggregated formulation and $n = 3$)}
	\label{fig:fcflp1_3}
\end{figure}

\begin{table}[h!]
	\centering
	\caption{Disaggregated formulation: frequencies of feasible solutions for $n =  3$}
	\label{tab:fcflp2_3}
	\begin{tabular}{@{}l r l r l r l r@{}}
		\toprule
		\texttt{METHOD} & Count & \texttt{METHOD} & Count & \texttt{METHOD} & Count & \texttt{METHOD} & Count \\
		\midrule
		\texttt{SA-20} & 9 &  \texttt{Tabu-0} & 0  & \texttt{QAOA-1} & 2  &  \texttt{Qrisp-1} & 1 \\
		\texttt{SA-100} & 10 &	\texttt{Tabu-50} & 1 & \texttt{QAOA-2} & 1 &  \texttt{Qrisp-2} & 3  \\
		\texttt{SA-500} & 10 &  \texttt{Tabu-250} & 3 &   \texttt{QAOA-3} & 3 &  \texttt{Qrisp-3} & 2   \\
		\texttt{WS-QAOA-1-R} & 0 & \texttt{WS-QAOA-1-S} & 3 & \texttt{WS-QAOA-1-L} & 5 & \texttt{WS-QAOA-1-C} & 5 
		\\
		\texttt{WS-QAOA-2-R} & 0 & \texttt{WS-QAOA-2-S} & 2 & \texttt{WS-QAOA-2-L} & 5 & \texttt{WS-QAOA-2-C} & 3 
		\\
		\texttt{WS-QAOA-3-R} & 2 & \texttt{WS-QAOA-3-S} & 2 & \texttt{WS-QAOA-3-L} & 4 & \texttt{WS-QAOA-3-C} & 5
		\\
		\bottomrule
	\end{tabular}
\end{table}
\begin{figure}[h!]
	\centering
	
	\begin{subfigure}{0.48\linewidth}
		\centering
		\resizebox{\linewidth}{!}{\includegraphics{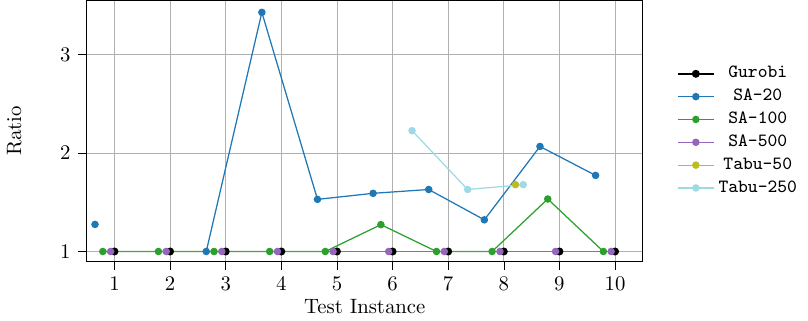}}
		\caption{Objective value ratios for classical heuristics}
		\label{fig:fcflp2_obj_val_h}
	\end{subfigure}
	\hfill
	\begin{subfigure}{0.48\linewidth}
		\centering
		\resizebox{\linewidth}{!}{\includegraphics{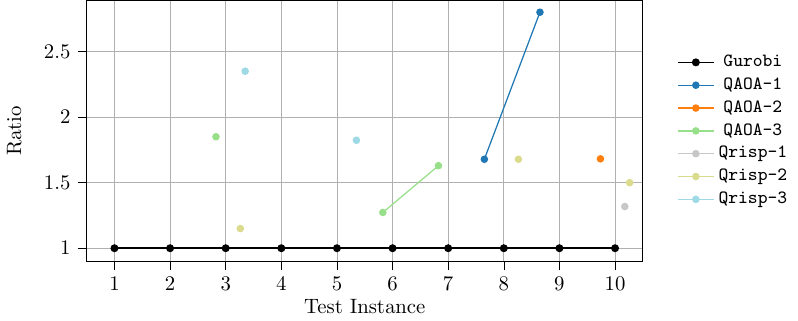}}
		\caption{Objective value ratios for quantum algorithms (QAOA)}
		\label{fig:fcflp2_3_obj_val_qaoa} 
	\end{subfigure}
	
	\vspace{0.25cm}
	
	\begin{subfigure}{0.48\linewidth}
		\centering
		\resizebox{\linewidth}{!}{\includegraphics{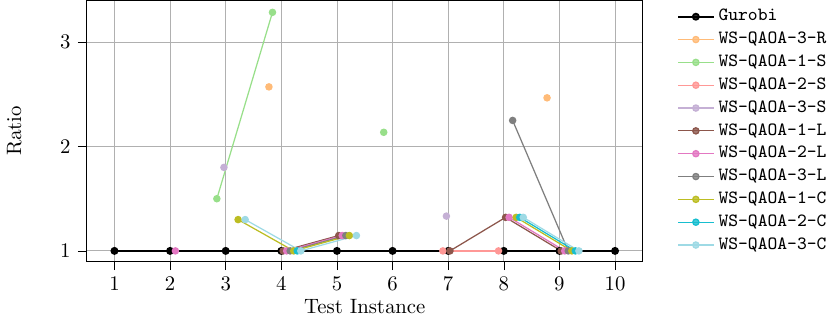}}
		\caption{Objective value ratios for quantum algorithms (WS-QAOA)}
		\label{fig:fcflp2_3_obj_val_ws_qaoa}
	\end{subfigure}
	\hfill
	\begin{subfigure}{0.48\linewidth}
		\centering
		\resizebox{\linewidth}{!}{\includegraphics{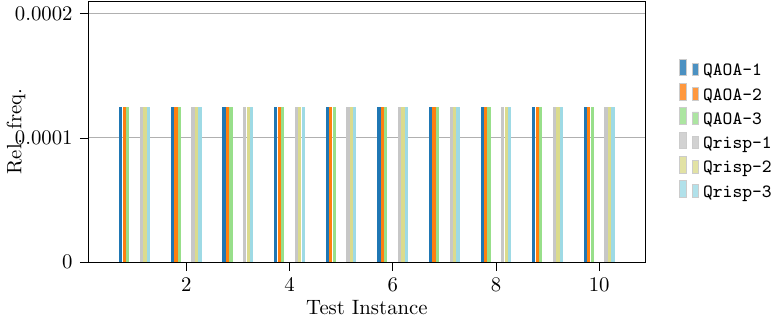}}
		\caption{Relative frequency plot (QAOA)}
		\label{fig:fcflp2_3_prob_qaoa}
	\end{subfigure}
	
	\vspace{0.25cm}
	
	\begin{subfigure}{0.48\linewidth}
		\centering
		\resizebox{\linewidth}{!}{\includegraphics{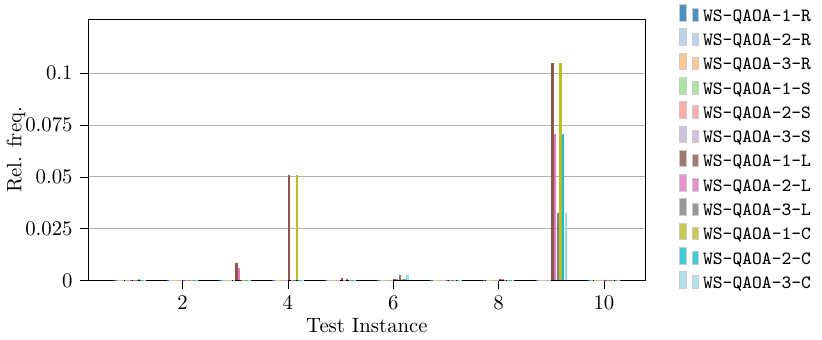}}
		\caption{Relative frequency plot (WS-QAOA)}
		\label{fig:fcflp2_3_prob_ws_qaoa}
	\end{subfigure}
	\caption{FCFLP (disaggregated formulation and $n = 3$)}
	\label{fig:fcflp2_3}
\end{figure}
The results differ markedly once again. Among the classical heuristics, \texttt{Tabu-0/50/250} exhibits comparatively weaker performance, whereas \texttt{SA-20/100/500} consistently identifies high-quality solutions, with objective values improving as the \texttt{num\_reads} parameter increases. Although Simulated Annealing is not a quantum algorithm, its strong performance demonstrates the general suitability of the QUBO formulations.
\newline
All quantum algorithms exhibit weaker performance overall. Among them, \texttt{WS-QAOA-1/2/3-R} performs worst, returning only a single feasible point for the aggregated formulation and only two feasible points for the disaggregated formulation. In contrast, the LP-based WS-QAOA variants \texttt{WS-QAOA-1/2/3-L} and \texttt{WS-QAOA-1/2/3-C} achieve the strongest results across both formulations, outperforming WS-QAOA initialized from both classical warm-start strategies, the continuous relaxation and semidefinite programming relaxation. We remark that \texttt{WS-QAOA-1/2/3-S} performs better than \texttt{WS-QAOA-1/2/3-R}. The LP relaxation of the disaggregated formulation is integral for test instances 4 and 9, explaining the peaks in relative frequencies observed in Figure \ref{fig:fcflp2_3_prob_ws_qaoa}.
\newline
Notably, both \texttt{QAOA-1/2/3} and \texttt{Qrisp-1/2/3} yield higher-quality solutions for the aggregated formulation compared to the disaggregated formulation.
\newline
Overall, the aggregated formulation produces slightly better results across most methods. For the aggregated formulation, \texttt{WS-QAOA-1/2/3-C} even outperforms \texttt{WS-QAOA-1/2/3-L}, indicating that the additional use of the local solver \texttt{L-BFGS-B} effectively improves the warm-start point derived from the LP relaxation and the binary encoding of the residual capacity. 
\newline
\newline
We conclude the numerical experiments by noting that the LP-based warm-start strategy introduced in Section \ref{sec:WarmStarting_L} (FCFLP) and the combined warm-start strategy introduced in Section \ref{sec:WarmStarting_C}, which integrates the LP-based approach with the continuous relaxation, constitute two novel and highly promising warm-start strategies for WS-QAOA. These approaches significantly outperform classical warm-start strategies (continuous relaxation with local solver and SDP relaxation) for the FCFLP and demonstrate significant potential for applying WS-QAOA to solve QUBOs derived from integer programs. We further emphasize that both warm-start strategies can be generalized and applied to general integer programs and their corresponding QUBO formulations.

\section{Conclusion}
We have introduced novel QUBO formulations for classical optimization problems arising in location science, network design, and logistics, which can serve as benchmark problems for the development of quantum optimization algorithms and the evaluation of quantum hardware. For general discrete optimization problems, we proposed a penalty parameter that ensures global optimizers of the QUBO correspond to global optimizers of the original integer program.
\newline
For the Fixed-Charge Facility Location Problem (FCFLP), we introduced two LP-based warm-start strategies for WS-QAOA that efficiently produce high-quality warm-start points. These warm-start approaches are broadly applicable to QUBOs derived from general integer programs. The numerical experiments have demonstrated that, on selected instances, gate-based quantum algorithms can outperform classical heuristics.
\newline
Overall, quantum optimization represents a promising complementary tool for discrete optimization, particularly as both quantum hardware and algorithms continue to advance.

\section*{Declarations}

\subsection*{Availability of data and materials}
All data generated or analysed during this study are included in this published article.
\subsection*{Competing interests}
The authors declare that they have no competing interests.
\subsection*{Funding}
The authors gratefully acknowledge funding from the BMFTR under the \textit{QuSol: Quantum Optimization Solver Kit} project.
\subsection*{Authors' contributions}
F.P.B.: Conceptualization, Data curation, Investigation, Methodology, Project administration, Software, Visualization, Writing - original draft, Writing - review \& editing.
S.N.: Conceptualization, Funding acquisition, Methodology, Project administration, Supervision, Writing - review \& editing. All authors read and approved the final manuscript.
\subsection*{Acknowledgements}
Not applicable.

\bibliography{sn-bibliography}


\begin{thebibliography}{62}
\ifx \bisbn   \undefined \def \bisbn  #1{ISBN #1}\fi
\ifx \binits  \undefined \def \binits#1{#1}\fi
\ifx \bauthor  \undefined \def \bauthor#1{#1}\fi
\ifx \batitle  \undefined \def \batitle#1{#1}\fi
\ifx \bjtitle  \undefined \def \bjtitle#1{#1}\fi
\ifx \bvolume  \undefined \def \bvolume#1{\textbf{#1}}\fi
\ifx \byear  \undefined \def \byear#1{#1}\fi
\ifx \bissue  \undefined \def \bissue#1{#1}\fi
\ifx \bfpage  \undefined \def \bfpage#1{#1}\fi
\ifx \blpage  \undefined \def \blpage #1{#1}\fi
\ifx \burl  \undefined \def \burl#1{\textsf{#1}}\fi
\ifx \doiurl  \undefined \def \doiurl#1{\url{https://doi.org/#1}}\fi
\ifx \betal  \undefined \def \betal{\textit{et al.}}\fi
\ifx \binstitute  \undefined \def \binstitute#1{#1}\fi
\ifx \binstitutionaled  \undefined \def \binstitutionaled#1{#1}\fi
\ifx \bctitle  \undefined \def \bctitle#1{#1}\fi
\ifx \beditor  \undefined \def \beditor#1{#1}\fi
\ifx \bpublisher  \undefined \def \bpublisher#1{#1}\fi
\ifx \bbtitle  \undefined \def \bbtitle#1{#1}\fi
\ifx \bedition  \undefined \def \bedition#1{#1}\fi
\ifx \bseriesno  \undefined \def \bseriesno#1{#1}\fi
\ifx \blocation  \undefined \def \blocation#1{#1}\fi
\ifx \bsertitle  \undefined \def \bsertitle#1{#1}\fi
\ifx \bsnm \undefined \def \bsnm#1{#1}\fi
\ifx \bsuffix \undefined \def \bsuffix#1{#1}\fi
\ifx \bparticle \undefined \def \bparticle#1{#1}\fi
\ifx \barticle \undefined \def \barticle#1{#1}\fi
\bibcommenthead
\ifx \bconfdate \undefined \def \bconfdate #1{#1}\fi
\ifx \botherref \undefined \def \botherref #1{#1}\fi
\ifx \url \undefined \def \url#1{\textsf{#1}}\fi
\ifx \bchapter \undefined \def \bchapter#1{#1}\fi
\ifx \bbook \undefined \def \bbook#1{#1}\fi
\ifx \bcomment \undefined \def \bcomment#1{#1}\fi
\ifx \oauthor \undefined \def \oauthor#1{#1}\fi
\ifx \citeauthoryear \undefined \def \citeauthoryear#1{#1}\fi
\ifx \endbibitem  \undefined \def \endbibitem {}\fi
\ifx \bconflocation  \undefined \def \bconflocation#1{#1}\fi
\ifx \arxivurl  \undefined \def \arxivurl#1{\textsf{#1}}\fi
\csname PreBibitemsHook\endcsname

\bibitem[\protect\citeauthoryear{Nemhauser and
  Wolsey}{1988}]{NemhauserWolsey:IntegerCombinatorialOptimization}
\begin{bbook}
\bauthor{\bsnm{Nemhauser}, \binits{G.}},
\bauthor{\bsnm{Wolsey}, \binits{L.}}:
\bbtitle{Integer and Combinatorial Optimization}.
John Wiley \& Sons, Ltd
(\byear{1988}).
\doiurl{10.1002/9781118627372}
\end{bbook}
\endbibitem

\bibitem[\protect\citeauthoryear{Wolsey}{2020}]{Wolsey:IntegerProgramming}
\begin{bbook}
\bauthor{\bsnm{Wolsey}, \binits{L.}}:
\bbtitle{Integer Programming}.
John Wiley \& Sons, Ltd
(\byear{2020}).
\doiurl{10.1002/9781119606475}
\end{bbook}
\endbibitem

\bibitem[\protect\citeauthoryear{Conforti et~al.}{2014}]{Conforti2014}
\begin{bbook}
\bauthor{\bsnm{Conforti}, \binits{M.}},
\bauthor{\bsnm{Cornu{\'e}jols}, \binits{G.}},
\bauthor{\bsnm{Zambelli}, \binits{G.}}:
\bbtitle{Integer Programming}.
\bpublisher{Springer},
\blocation{Cham}
(\byear{2014}).
\doiurl{10.1007/978-3-319-11008-0}
\end{bbook}
\endbibitem

\bibitem[\protect\citeauthoryear{Drezner and
  Hamacher}{2002}]{DreznerHamacher:FacilityLocation}
\begin{bbook}
\bauthor{\bsnm{Drezner}, \binits{Z.}},
\bauthor{\bsnm{Hamacher}, \binits{H.W.}}:
\bbtitle{Faciliy Location: Applications and Theory}.
\bpublisher{Springer},
\blocation{Berlin, Heidelberg}
(\byear{2002})
\end{bbook}
\endbibitem

\bibitem[\protect\citeauthoryear{Nickel and
  Puerto}{2005}]{NickelPuerto:LocationTheory}
\begin{bbook}
\bauthor{\bsnm{Nickel}, \binits{S.}},
\bauthor{\bsnm{Puerto}, \binits{J.}}:
\bbtitle{Location Theory: A Unified Approach}.
\bpublisher{Springer},
\blocation{Berlin, Heidelberg}
(\byear{2005}).
\doiurl{10.1007/3-540-27640-8}
\end{bbook}
\endbibitem

\bibitem[\protect\citeauthoryear{Laporte
  et~al.}{2019}]{LaporteNickelLocationScience}
\begin{bbook}
\bauthor{\bsnm{Laporte}, \binits{G.}},
\bauthor{\bsnm{Nickel}, \binits{S.}},
\bauthor{\bsnm{Saldanha-da-Gama}, \binits{F.}}:
\bbtitle{Location Science}.
\bpublisher{Springer},
\blocation{Cham}
(\byear{2019}).
\doiurl{10.1007/978-3-030-32177-2}
\end{bbook}
\endbibitem

\bibitem[\protect\citeauthoryear{Born and
  Fock}{1928}]{Born1928:QuantumAdiabaticTheorem}
\begin{barticle}
\bauthor{\bsnm{Born}, \binits{M.}},
\bauthor{\bsnm{Fock}, \binits{V.}}:
\batitle{Beweis des adiabatensatzes}.
\bjtitle{Zeitschrift f{\"u}r Physik}
\bvolume{51}(\bissue{3}),
\bfpage{165}--\blpage{180}
(\byear{1928})
\doiurl{10.1007/BF01343193}
\end{barticle}
\endbibitem

\bibitem[\protect\citeauthoryear{Kato}{1950}]{Kato:QuantumAdiabaticTheorem}
\begin{barticle}
\bauthor{\bsnm{Kato}, \binits{T.}}:
\batitle{On the adiabatic theorem of quantum mechanics}.
\bjtitle{Journal of the Physical Society of Japan}
\bvolume{5}(\bissue{6}),
\bfpage{435}--\blpage{439}
(\byear{1950})
\doiurl{10.1143/JPSJ.5.435}
\end{barticle}
\endbibitem

\bibitem[\protect\citeauthoryear{Farhi
  et~al.}{2000}]{farhi:AdiabaticQuantumOptimization}
\begin{botherref}
\oauthor{\bsnm{Farhi}, \binits{E.}},
\oauthor{\bsnm{Goldstone}, \binits{J.}},
\oauthor{\bsnm{Gutmann}, \binits{S.}},
\oauthor{\bsnm{Sipser}, \binits{M.}}:
Quantum Computation by Adiabatic Evolution
(2000).
\url{https://arxiv.org/abs/quant-ph/0001106}
\end{botherref}
\endbibitem

\bibitem[\protect\citeauthoryear{Farhi et~al.}{2014}]{farhi2014QAOA}
\begin{botherref}
\oauthor{\bsnm{Farhi}, \binits{E.}},
\oauthor{\bsnm{Goldstone}, \binits{J.}},
\oauthor{\bsnm{Gutmann}, \binits{S.}}:
A Quantum Approximate Optimization Algorithm
(2014).
\url{https://arxiv.org/abs/1411.4028}
\end{botherref}
\endbibitem

\bibitem[\protect\citeauthoryear{Lucas}{2014}]{LucasIsing}
\begin{barticle}
\bauthor{\bsnm{Lucas}, \binits{A.}}:
\batitle{Ising formulations of many {NP} problems}.
\bjtitle{Frontiers in Physics}
(\byear{2014})
\doiurl{10.3389/fphy.2014.00005}
\end{barticle}
\endbibitem

\bibitem[\protect\citeauthoryear{Karp}{1972}]{Karp1972}
\begin{bbook}
\bauthor{\bsnm{Karp}, \binits{R.M.}}:
In: \beditor{\bsnm{Miller}, \binits{R.E.}},
\beditor{\bsnm{Thatcher}, \binits{J.W.}},
\beditor{\bsnm{Bohlinger}, \binits{J.D.}} (eds.)
\bbtitle{Reducibility among Combinatorial Problems},
pp. \bfpage{85}--\blpage{103}.
\bpublisher{Springer},
\blocation{Boston, MA}
(\byear{1972}).
\doiurl{10.1007/978-1-4684-2001-2\_9}
\end{bbook}
\endbibitem

\bibitem[\protect\citeauthoryear{Abbas
  et~al.}{2024}]{Abbas:ChallengesandOpportunities}
\begin{barticle}
\bauthor{\bsnm{Abbas}, \binits{A.}},
\bauthor{\bsnm{Ambainis}, \binits{A.}},
\bauthor{\bsnm{Augustino}, \binits{B.}},
\bauthor{\bsnm{B{\"a}rtschi}, \binits{A.}},
\bauthor{\bsnm{Buhrman}, \binits{H.}},
\bauthor{\bsnm{Coffrin}, \binits{C.}},
\bauthor{\bsnm{Cortiana}, \binits{G.}},
\bauthor{\bsnm{Dunjko}, \binits{V.}},
\bauthor{\bsnm{Egger}, \binits{D.J.}},
\bauthor{\bsnm{Elmegreen}, \binits{B.G.}},
\bauthor{\bsnm{Franco}, \binits{N.}},
\bauthor{\bsnm{Fratini}, \binits{F.}},
\bauthor{\bsnm{Fuller}, \binits{B.}},
\bauthor{\bsnm{Gacon}, \binits{J.}},
\bauthor{\bsnm{Gonciulea}, \binits{C.}},
\bauthor{\bsnm{Gribling}, \binits{S.}},
\bauthor{\bsnm{Gupta}, \binits{S.}},
\bauthor{\bsnm{Hadfield}, \binits{S.}},
\bauthor{\bsnm{Heese}, \binits{R.}},
\bauthor{\bsnm{Kircher}, \binits{G.}},
\bauthor{\bsnm{Kleinert}, \binits{T.}},
\bauthor{\bsnm{Koch}, \binits{T.}},
\bauthor{\bsnm{Korpas}, \binits{G.}},
\bauthor{\bsnm{Lenk}, \binits{S.}},
\bauthor{\bsnm{Marecek}, \binits{J.}},
\bauthor{\bsnm{Markov}, \binits{V.}},
\bauthor{\bsnm{Mazzola}, \binits{G.}},
\bauthor{\bsnm{Mensa}, \binits{S.}},
\bauthor{\bsnm{Mohseni}, \binits{N.}},
\bauthor{\bsnm{Nannicini}, \binits{G.}},
\bauthor{\bsnm{O'Meara}, \binits{C.}},
\bauthor{\bsnm{Tapia}, \binits{E.P.}},
\bauthor{\bsnm{Pokutta}, \binits{S.}},
\bauthor{\bsnm{Proissl}, \binits{M.}},
\bauthor{\bsnm{Rebentrost}, \binits{P.}},
\bauthor{\bsnm{Sahin}, \binits{E.}},
\bauthor{\bsnm{Symons}, \binits{B.C.B.}},
\bauthor{\bsnm{Tornow}, \binits{S.}},
\bauthor{\bsnm{Valls}, \binits{V.}},
\bauthor{\bsnm{Woerner}, \binits{S.}},
\bauthor{\bsnm{Wolf-Bauwens}, \binits{M.L.}},
\bauthor{\bsnm{Yard}, \binits{J.}},
\bauthor{\bsnm{Yarkoni}, \binits{S.}},
\bauthor{\bsnm{Zechiel}, \binits{D.}},
\bauthor{\bsnm{Zhuk}, \binits{S.}},
\bauthor{\bsnm{Zoufal}, \binits{C.}}:
\batitle{Challenges and opportunities in quantum optimization}.
\bjtitle{Nature Reviews Physics}
\bvolume{6}(\bissue{12}),
\bfpage{718}--\blpage{735}
(\byear{2024})
\doiurl{10.1038/s42254-024-00770-9}
\end{barticle}
\endbibitem

\bibitem[\protect\citeauthoryear{Egger
  et~al.}{2021}]{Egger2021warmstartingquantum}
\begin{barticle}
\bauthor{\bsnm{Egger}, \binits{D.J.}},
\bauthor{\bsnm{Mare{\v{c}}ek}, \binits{J.}},
\bauthor{\bsnm{Woerner}, \binits{S.}}:
\batitle{Warm-starting quantum optimization}.
\bjtitle{{Quantum}}
\bvolume{5},
\bfpage{479}
(\byear{2021})
\doiurl{10.22331/q-2021-06-17-479}
\end{barticle}
\endbibitem

\bibitem[\protect\citeauthoryear{Preskill}{2018}]{PreskillNISQ}
\begin{barticle}
\bauthor{\bsnm{Preskill}, \binits{J.}}:
\batitle{Quantum {C}omputing in the {NISQ} era and beyond}.
\bjtitle{{Quantum}}
\bvolume{2},
\bfpage{79}
(\byear{2018})
\doiurl{10.22331/q-2018-08-06-79}
\end{barticle}
\endbibitem

\bibitem[\protect\citeauthoryear{Pellow-Jarman
  et~al.}{2024}]{Pellow_Jarman_2024_Noise}
\begin{botherref}
\oauthor{\bsnm{Pellow-Jarman}, \binits{A.}},
\oauthor{\bsnm{McFarthing}, \binits{S.}},
\oauthor{\bsnm{Sinayskiy}, \binits{I.}},
\oauthor{\bsnm{Park}, \binits{D.K.}},
\oauthor{\bsnm{Pillay}, \binits{A.}},
\oauthor{\bsnm{Petruccione}, \binits{F.}}:
The effect of classical optimizers and ansatz depth on qaoa performance in
  noisy devices.
Scientific Reports
\textbf{14}(1)
(2024)
\doiurl{10.1038/s41598-024-66625-6}
\end{botherref}
\endbibitem

\bibitem[\protect\citeauthoryear{Quek et~al.}{2024}]{Quek2024_Noise}
\begin{barticle}
\bauthor{\bsnm{Quek}, \binits{Y.}},
\bauthor{\bsnm{Stilck~Fran{\c{c}}a}, \binits{D.}},
\bauthor{\bsnm{Khatri}, \binits{S.}},
\bauthor{\bsnm{Meyer}, \binits{J.J.}},
\bauthor{\bsnm{Eisert}, \binits{J.}}:
\batitle{Exponentially tighter bounds on limitations of quantum error
  mitigation}.
\bjtitle{Nature Physics}
\bvolume{20}(\bissue{10}),
\bfpage{1648}--\blpage{1658}
(\byear{2024})
\doiurl{10.1038/s41567-024-02536-7}
\end{barticle}
\endbibitem

\bibitem[\protect\citeauthoryear{Javadi-Abhari et~al.}{2024}]{qiskit2024}
\begin{botherref}
\oauthor{\bsnm{Javadi-Abhari}, \binits{A.}},
\oauthor{\bsnm{Treinish}, \binits{M.}},
\oauthor{\bsnm{Krsulich}, \binits{K.}},
\oauthor{\bsnm{Wood}, \binits{C.J.}},
\oauthor{\bsnm{Lishman}, \binits{J.}},
\oauthor{\bsnm{Gacon}, \binits{J.}},
\oauthor{\bsnm{Martiel}, \binits{S.}},
\oauthor{\bsnm{Nation}, \binits{P.D.}},
\oauthor{\bsnm{Bishop}, \binits{L.S.}},
\oauthor{\bsnm{Cross}, \binits{A.W.}},
\oauthor{\bsnm{Johnson}, \binits{B.R.}},
\oauthor{\bsnm{Gambetta}, \binits{J.M.}}:
Quantum computing with {Q}iskit
(2024).
\doiurl{10.48550/arXiv.2405.08810}
\end{botherref}
\endbibitem

\bibitem[\protect\citeauthoryear{{IBM Quantum}}{2025}]{IBMQuantum}
\begin{botherref}
\oauthor{\bsnm{{IBM Quantum}}}:
(2025).
\url{https://www.ibm.com/quantum}
\end{botherref}
\endbibitem

\bibitem[\protect\citeauthoryear{Glover
  et~al.}{2022}]{QUBO:Kochenberger_Tutorial}
\begin{barticle}
\bauthor{\bsnm{Glover}, \binits{F.}},
\bauthor{\bsnm{Kochenberger}, \binits{G.}},
\bauthor{\bsnm{Hennig}, \binits{R.}},
\bauthor{\bsnm{Du}, \binits{Y.}}:
\batitle{Quantum bridge analytics i: a tutorial on formulating and using qubo
  models}.
\bjtitle{Annals of Operations Research}
\bvolume{314}(\bissue{1}),
\bfpage{141}--\blpage{183}
(\byear{2022})
\doiurl{10.1007/s10479-022-04634-2}
\end{barticle}
\endbibitem

\bibitem[\protect\citeauthoryear{Bayerstadler
  et~al.}{2021}]{QuantumApplications}
\begin{barticle}
\bauthor{\bsnm{Bayerstadler}, \binits{A.}},
\bauthor{\bsnm{Becquin}, \binits{G.}},
\bauthor{\bsnm{Binder}, \binits{J.}},
\bauthor{\bsnm{Botter}, \binits{T.}},
\bauthor{\bsnm{Ehm}, \binits{H.}},
\bauthor{\bsnm{Ehmer}, \binits{T.}},
\bauthor{\bsnm{Erdmann}, \binits{M.}},
\bauthor{\bsnm{Gaus}, \binits{N.}},
\bauthor{\bsnm{Harbach}, \binits{P.}},
\bauthor{\bsnm{Hess}, \binits{M.}},
\bauthor{\bsnm{Klepsch}, \binits{J.}},
\bauthor{\bsnm{Leib}, \binits{M.}},
\bauthor{\bsnm{Luber}, \binits{S.}},
\bauthor{\bsnm{Luckow}, \binits{A.}},
\bauthor{\bsnm{Mansky}, \binits{M.}},
\bauthor{\bsnm{Mauerer}, \binits{W.}},
\bauthor{\bsnm{Neukart}, \binits{F.}},
\bauthor{\bsnm{Niedermeier}, \binits{C.}},
\bauthor{\bsnm{Palackal}, \binits{L.}},
\bauthor{\bsnm{Pfeiffer}, \binits{R.}},
\bauthor{\bsnm{Polenz}, \binits{C.}},
\bauthor{\bsnm{Sepulveda}, \binits{J.}},
\bauthor{\bsnm{Sievers}, \binits{T.}},
\bauthor{\bsnm{Standen}, \binits{B.}},
\bauthor{\bsnm{Streif}, \binits{M.}},
\bauthor{\bsnm{Strohm}, \binits{T.}},
\bauthor{\bsnm{Utschig-Utschig}, \binits{C.}},
\bauthor{\bsnm{Volz}, \binits{D.}},
\bauthor{\bsnm{Weiss}, \binits{H.}},
\bauthor{\bsnm{Winter}, \binits{F.}},
\bauthor{\bsnm{Technology}, \binits{Q.}},
\bauthor{\bsnm{--~QUTAC}, \binits{A.C.}}:
\batitle{Industry quantum computing applications}.
\bjtitle{EPJ Quantum Technology}
\bvolume{8}(\bissue{1}),
\bfpage{25}
(\byear{2021})
\doiurl{10.1140/epjqt/s40507-021-00114-x}
\end{barticle}
\endbibitem

\bibitem[\protect\citeauthoryear{Egger et~al.}{2020}]{QUBO:Egger_Finance}
\begin{barticle}
\bauthor{\bsnm{Egger}, \binits{D.J.}},
\bauthor{\bsnm{Gambella}, \binits{C.}},
\bauthor{\bsnm{Marecek}, \binits{J.}},
\bauthor{\bsnm{McFaddin}, \binits{S.}},
\bauthor{\bsnm{Mevissen}, \binits{M.}},
\bauthor{\bsnm{Raymond}, \binits{R.}},
\bauthor{\bsnm{Simonetto}, \binits{A.}},
\bauthor{\bsnm{Woerner}, \binits{S.}},
\bauthor{\bsnm{Yndurain}, \binits{E.}}:
\batitle{Quantum computing for finance: State-of-the-art and future prospects}.
\bjtitle{IEEE Transactions on Quantum Engineering}
\bvolume{1},
\bfpage{1}--\blpage{24}
(\byear{2020})
\doiurl{10.1109/TQE.2020.3030314}
\end{barticle}
\endbibitem

\bibitem[\protect\citeauthoryear{Brandhofer
  et~al.}{2022}]{QUBO:Brandhofer_Portfolio}
\begin{barticle}
\bauthor{\bsnm{Brandhofer}, \binits{S.}},
\bauthor{\bsnm{Braun}, \binits{D.}},
\bauthor{\bsnm{Dehn}, \binits{V.}},
\bauthor{\bsnm{Hellstern}, \binits{G.}},
\bauthor{\bsnm{H{\"u}ls}, \binits{M.}},
\bauthor{\bsnm{Ji}, \binits{Y.}},
\bauthor{\bsnm{Polian}, \binits{I.}},
\bauthor{\bsnm{Bhatia}, \binits{A.S.}},
\bauthor{\bsnm{Wellens}, \binits{T.}}:
\batitle{Benchmarking the performance of portfolio optimization with qaoa}.
\bjtitle{Quantum Information Processing}
\bvolume{22}(\bissue{1}),
\bfpage{25}
(\byear{2022})
\doiurl{10.1007/s11128-022-03766-5}
\end{barticle}
\endbibitem

\bibitem[\protect\citeauthoryear{Hodson et~al.}{2019}]{QUBO:Hodson_portfolio}
\begin{botherref}
\oauthor{\bsnm{Hodson}, \binits{M.}},
\oauthor{\bsnm{Ruck}, \binits{B.}},
\oauthor{\bsnm{Ong}, \binits{H.}},
\oauthor{\bsnm{Garvin}, \binits{D.}},
\oauthor{\bsnm{Dulman}, \binits{S.}}:
Portfolio rebalancing experiments using the Quantum Alternating Operator Ansatz
(2019).
\url{https://arxiv.org/abs/1911.05296}
\end{botherref}
\endbibitem

\bibitem[\protect\citeauthoryear{Bochkarev et~al.}{2025}]{QUBO:Bochkarev}
\begin{barticle}
\bauthor{\bsnm{Bochkarev}, \binits{A.}},
\bauthor{\bsnm{Heese}, \binits{R.}},
\bauthor{\bsnm{Jäger}, \binits{S.}},
\bauthor{\bsnm{Schiewe}, \binits{P.}},
\bauthor{\bsnm{Schöbel}, \binits{A.}}:
\batitle{Quantum computing for discrete optimization: A highlight of three
  technologies}.
\bjtitle{European Journal of Operational Research}
(\byear{2025})
\doiurl{10.1016/j.ejor.2025.07.063}
\end{barticle}
\endbibitem

\bibitem[\protect\citeauthoryear{Amaro et~al.}{2022}]{QUBO:Scheduling}
\begin{barticle}
\bauthor{\bsnm{Amaro}, \binits{D.}},
\bauthor{\bsnm{Rosenkranz}, \binits{M.}},
\bauthor{\bsnm{Fitzpatrick}, \binits{N.}},
\bauthor{\bsnm{Hirano}, \binits{K.}},
\bauthor{\bsnm{Fiorentini}, \binits{M.}}:
\batitle{A case study of variational quantum algorithms for a job shop
  scheduling problem}.
\bjtitle{EPJ Quantum Technology}
\bvolume{9}(\bissue{1}),
\bfpage{5}
(\byear{2022})
\doiurl{10.1140/epjqt/s40507-022-00123-4}
\end{barticle}
\endbibitem

\bibitem[\protect\citeauthoryear{Bo{\.{z}}ejko
  et~al.}{2024}]{QUBO:BinaryKnapsack_Dwave}
\begin{barticle}
\bauthor{\bsnm{Bo{\.{z}}ejko}, \binits{W.}},
\bauthor{\bsnm{Burduk}, \binits{A.}},
\bauthor{\bsnm{Pempera}, \binits{J.}},
\bauthor{\bsnm{Uchro{\'{n}}ski}, \binits{M.}},
\bauthor{\bsnm{Wodecki}, \binits{M.}}:
\batitle{Optimal solving of a binary knapsack problem on a d-wave quantum
  machine and its implementation in production systems}.
\bjtitle{Annals of Operations Research}
(\byear{2024})
\doiurl{10.1007/s10479-024-06025-1}
\end{barticle}
\endbibitem

\bibitem[\protect\citeauthoryear{Cattelan and
  Yarkoni}{2024}]{QUBO:Cattelan_Routing}
\begin{barticle}
\bauthor{\bsnm{Cattelan}, \binits{M.}},
\bauthor{\bsnm{Yarkoni}, \binits{S.}}:
\batitle{Modeling routing problems in qubo with application to ride-hailing}.
\bjtitle{Scientific Reports}
\bvolume{14}(\bissue{1}),
\bfpage{19768}
(\byear{2024})
\doiurl{10.1038/s41598-024-70649-3}
\end{barticle}
\endbibitem

\bibitem[\protect\citeauthoryear{Mohanty et~al.}{2023}]{QUBO:Routing_2}
\begin{barticle}
\bauthor{\bsnm{Mohanty}, \binits{N.}},
\bauthor{\bsnm{Behera}, \binits{B.K.}},
\bauthor{\bsnm{Ferrie}, \binits{C.}}:
\batitle{Analysis of the vehicle routing problem solved via hybrid quantum
  algorithms in the presence of noisy channels}.
\bjtitle{IEEE Transactions on Quantum Engineering}
\bvolume{4},
\bfpage{1}--\blpage{14}
(\byear{2023})
\doiurl{10.1109/TQE.2023.3303989}
\end{barticle}
\endbibitem

\bibitem[\protect\citeauthoryear{Harwood et~al.}{2021}]{QUBO:Routing2021}
\begin{barticle}
\bauthor{\bsnm{Harwood}, \binits{S.}},
\bauthor{\bsnm{Gambella}, \binits{C.}},
\bauthor{\bsnm{Trenev}, \binits{D.}},
\bauthor{\bsnm{Simonetto}, \binits{A.}},
\bauthor{\bsnm{Bernal~Neira}, \binits{D.}},
\bauthor{\bsnm{Greenberg}, \binits{D.}}:
\batitle{Formulating and solving routing problems on quantum computers}.
\bjtitle{IEEE Transactions on Quantum Engineering}
\bvolume{2},
\bfpage{1}--\blpage{17}
(\byear{2021})
\doiurl{10.1109/TQE.2021.3049230}
\end{barticle}
\endbibitem

\bibitem[\protect\citeauthoryear{von Th\"unen}{1842}]{vonThunen1842}
\begin{bbook}
\bauthor{\bsnm{Th\"unen}, \binits{J.H.}}:
\bbtitle{The Isolated State}.
\bpublisher{Pergamon Press},
\blocation{Oxford}
(\byear{1842}).
\bcomment{Originally published 1842, translated 1966}
\end{bbook}
\endbibitem

\bibitem[\protect\citeauthoryear{Weber}{1909}]{Weber1909}
\begin{bbook}
\bauthor{\bsnm{Weber}, \binits{A.}}:
\bbtitle{Theory of the Location of Industries}.
\bpublisher{University of Chicago Press},
\blocation{Chicago}
(\byear{1909}).
\bcomment{Originally published 1909, translated 1929}
\end{bbook}
\endbibitem

\bibitem[\protect\citeauthoryear{Eiselt and Marianov}{2011}]{Eiselt2011}
\begin{bbook}
\bauthor{\bsnm{Eiselt}, \binits{H.A.}},
\bauthor{\bsnm{Marianov}, \binits{V.}}:
\bbtitle{Foundations of Location Analysis}.
\bpublisher{Springer},
\blocation{New York, NY}
(\byear{2011}).
\doiurl{10.1007/978-1-4419-7572-0}
\end{bbook}
\endbibitem

\bibitem[\protect\citeauthoryear{Laporte
  et~al.}{2019}]{LocationScience:Introduction}
\begin{bbook}
\bauthor{\bsnm{Laporte}, \binits{G.}},
\bauthor{\bsnm{Nickel}, \binits{S.}},
\bauthor{\bsnm{Saldanha-da-Gama}, \binits{F.}}:
In: \beditor{\bsnm{Laporte}, \binits{G.}},
\beditor{\bsnm{Nickel}, \binits{S.}},
\beditor{\bsnm{Gama}, \binits{F.}} (eds.)
\bbtitle{Introduction to Location Science},
pp. \bfpage{1}--\blpage{21}.
\bpublisher{Springer},
\blocation{Cham}
(\byear{2019}).
\doiurl{10.1007/978-3-030-32177-2\_1}
\end{bbook}
\endbibitem

\bibitem[\protect\citeauthoryear{Bazaraa et~al.}{2006}]{Bazaraa}
\begin{bbook}
\bauthor{\bsnm{Bazaraa}, \binits{M.S.}},
\bauthor{\bsnm{Sherali}, \binits{H.D.}},
\bauthor{\bsnm{Shetty}, \binits{C.M.}}:
\bbtitle{Nonlinear Programming}.
John Wiley \& Sons, Ltd
(\byear{2006}).
\doiurl{10.1002/0471787779}
\end{bbook}
\endbibitem

\bibitem[\protect\citeauthoryear{Hakimi}{1964}]{HakimiPMedianProblem1}
\begin{barticle}
\bauthor{\bsnm{Hakimi}, \binits{S.L.}}:
\batitle{Optimum locations of switching centers and the absolute centers and
  medians of a graph}.
\bjtitle{Operations Research}
\bvolume{12}(\bissue{3}),
\bfpage{450}--\blpage{459}
(\byear{1964}).
Accessed 2025-09-04
\end{barticle}
\endbibitem

\bibitem[\protect\citeauthoryear{Mar{\'i}n and
  Pelegr{\'i}n}{2019}]{LocationScience:Chapter2pMedian}
\begin{bbook}
\bauthor{\bsnm{Mar{\'i}n}, \binits{A.}},
\bauthor{\bsnm{Pelegr{\'i}n}, \binits{M.}}:
In: \beditor{\bsnm{Laporte}, \binits{G.}},
\beditor{\bsnm{Nickel}, \binits{S.}},
\beditor{\bsnm{Gama}, \binits{F.}} (eds.)
\bbtitle{p-Median Problems},
pp. \bfpage{25}--\blpage{50}.
\bpublisher{Springer},
\blocation{Cham}
(\byear{2019}).
\doiurl{10.1007/978-3-030-32177-2\_2}
\end{bbook}
\endbibitem

\bibitem[\protect\citeauthoryear{Hakimi}{1965}]{HakimiPMedianProblem2}
\begin{barticle}
\bauthor{\bsnm{Hakimi}, \binits{S.L.}}:
\batitle{Optimum distribution of switching centers in a communication network
  and some related graph theoretic problems}.
\bjtitle{Operations Research}
\bvolume{13}(\bissue{3}),
\bfpage{462}--\blpage{475}
(\byear{1965}).
Accessed 2025-09-04
\end{barticle}
\endbibitem

\bibitem[\protect\citeauthoryear{{\c{C}}al{\i}k
  et~al.}{2019}]{LocationScience:Chapter3pCenter}
\begin{bbook}
\bauthor{\bsnm{{\c{C}}al{\i}k}, \binits{H.}},
\bauthor{\bsnm{Labb{\'e}}, \binits{M.}},
\bauthor{\bsnm{Yaman}, \binits{H.}}:
In: \beditor{\bsnm{Laporte}, \binits{G.}},
\beditor{\bsnm{Nickel}, \binits{S.}},
\beditor{\bsnm{Gama}, \binits{F.}} (eds.)
\bbtitle{p-Center Problems},
pp. \bfpage{51}--\blpage{65}.
\bpublisher{Springer},
\blocation{Cham}
(\byear{2019}).
\doiurl{10.1007/978-3-030-32177-2\_3}
\end{bbook}
\endbibitem

\bibitem[\protect\citeauthoryear{Fern{\'a}ndez and
  Landete}{2019}]{LocationScience:ChapterFCFLP}
\begin{bbook}
\bauthor{\bsnm{Fern{\'a}ndez}, \binits{E.}},
\bauthor{\bsnm{Landete}, \binits{M.}}:
In: \beditor{\bsnm{Laporte}, \binits{G.}},
\beditor{\bsnm{Nickel}, \binits{S.}},
\beditor{\bsnm{Gama}, \binits{F.}} (eds.)
\bbtitle{Fixed-Charge Facility Location Problems},
pp. \bfpage{67}--\blpage{98}.
\bpublisher{Springer},
\blocation{Cham}
(\byear{2019}).
\doiurl{10.1007/978-3-030-32177-2\_4}
\end{bbook}
\endbibitem

\bibitem[\protect\citeauthoryear{Fisher et~al.}{1986}]{Fisher:GAP}
\begin{barticle}
\bauthor{\bsnm{Fisher}, \binits{M.L.}},
\bauthor{\bsnm{Jaikumar}, \binits{R.}},
\bauthor{\bsnm{Wassenhove}, \binits{L.N.V.}}:
\batitle{A multiplier adjustment method for the generalized assignment
  problem}.
\bjtitle{Management Science}
\bvolume{32}(\bissue{9}),
\bfpage{1095}--\blpage{1103}
(\byear{1986}).
Accessed 2025-11-10
\end{barticle}
\endbibitem

\bibitem[\protect\citeauthoryear{Nickel}{2001}]{Nickel:DOMP}
\begin{bchapter}
\bauthor{\bsnm{Nickel}, \binits{S.}}:
\bctitle{Discrete ordered weber problems}.
In: \beditor{\bsnm{Fleischmann}, \binits{B.}},
\beditor{\bsnm{Lasch}, \binits{R.}},
\beditor{\bsnm{Derigs}, \binits{U.}},
\beditor{\bsnm{Domschke}, \binits{W.}},
\beditor{\bsnm{Rieder}, \binits{U.}} (eds.)
\bbtitle{Operations Research Proceedings},
pp. \bfpage{71}--\blpage{76}.
\bpublisher{Springer},
\blocation{Berlin, Heidelberg}
(\byear{2001}).
\doiurl{10.1007/978-3-642-56656-1\_12}
\end{bchapter}
\endbibitem

\bibitem[\protect\citeauthoryear{Puerto and
  Rodr{\'i}guez-Ch{\'i}a}{2019}]{LocationScience:ChapterOMP}
\begin{bbook}
\bauthor{\bsnm{Puerto}, \binits{J.}},
\bauthor{\bsnm{Rodr{\'i}guez-Ch{\'i}a}, \binits{A.M.}}:
In: \beditor{\bsnm{Laporte}, \binits{G.}},
\beditor{\bsnm{Nickel}, \binits{S.}},
\beditor{\bsnm{Gama}, \binits{F.}} (eds.)
\bbtitle{Ordered Median Location Problems},
pp. \bfpage{261}--\blpage{302}.
\bpublisher{Springer},
\blocation{Cham}
(\byear{2019}).
\doiurl{10.1007/978-3-030-32177-2\_10}
\end{bbook}
\endbibitem

\bibitem[\protect\citeauthoryear{Boland et~al.}{2006}]{Nickel:DOMPExact}
\begin{barticle}
\bauthor{\bsnm{Boland}, \binits{N.}},
\bauthor{\bsnm{Domínguez-Marín}, \binits{P.}},
\bauthor{\bsnm{Nickel}, \binits{S.}},
\bauthor{\bsnm{Puerto}, \binits{J.}}:
\batitle{Exact procedures for solving the discrete ordered median problem}.
\bjtitle{Computers \& Operations Research}
\bvolume{33}(\bissue{11}),
\bfpage{3270}--\blpage{3300}
(\byear{2006})
\doiurl{10.1016/j.cor.2005.03.025} .
\bcomment{Part Special Issue: Operations Research and Data Mining}
\end{barticle}
\endbibitem

\bibitem[\protect\citeauthoryear{Dom{\'i}nguez-Mar{\'i}n
  et~al.}{2005}]{Nickel:DOMPHeuristic}
\begin{barticle}
\bauthor{\bsnm{Dom{\'i}nguez-Mar{\'i}n}, \binits{P.}},
\bauthor{\bsnm{Nickel}, \binits{S.}},
\bauthor{\bsnm{Hansen}, \binits{P.}},
\bauthor{\bsnm{Mladenovi{\'{c}}}, \binits{N.}}:
\batitle{Heuristic procedures for solving the discrete ordered median problem}.
\bjtitle{Annals of Operations Research}
\bvolume{136}(\bissue{1}),
\bfpage{145}--\blpage{173}
(\byear{2005})
\doiurl{10.1007/s10479-005-2043-3}
\end{barticle}
\endbibitem

\bibitem[\protect\citeauthoryear{McCormick}{1976}]{McCormick1976}
\begin{barticle}
\bauthor{\bsnm{McCormick}, \binits{G.P.}}:
\batitle{Computability of global solutions to factorable nonconvex programs:
  Part i --- convex underestimating problems}.
\bjtitle{Mathematical Programming}
\bvolume{10}(\bissue{1}),
\bfpage{147}--\blpage{175}
(\byear{1976})
\doiurl{10.1007/BF01580665}
\end{barticle}
\endbibitem

\bibitem[\protect\citeauthoryear{Nickel and
  da~Gama}{2015}]{LocationScience:ChapterMultiPeriod}
\begin{bbook}
\bauthor{\bsnm{Nickel}, \binits{S.}},
\bauthor{\bsnm{Gama}, \binits{F.S.}}:
In: \beditor{\bsnm{Laporte}, \binits{G.}},
\beditor{\bsnm{Nickel}, \binits{S.}},
\beditor{\bsnm{Gama}, \binits{F.}} (eds.)
\bbtitle{Multi-Period Facility Location},
pp. \bfpage{289}--\blpage{310}.
\bpublisher{Springer},
\blocation{Cham}
(\byear{2015}).
\doiurl{10.1007/978-3-319-13111-5\_11}
\end{bbook}
\endbibitem

\bibitem[\protect\citeauthoryear{Wesolowsky and
  Truscott}{1975}]{WesolowskyTruscottMultiPeriod}
\begin{barticle}
\bauthor{\bsnm{Wesolowsky}, \binits{G.O.}},
\bauthor{\bsnm{Truscott}, \binits{W.G.}}:
\batitle{The multiperiod location-allocation problem with relocation of
  facilities}.
\bjtitle{Management Science}
\bvolume{22}(\bissue{1}),
\bfpage{57}--\blpage{65}
(\byear{1975}).
Accessed 2025-09-04
\end{barticle}
\endbibitem

\bibitem[\protect\citeauthoryear{Kirkpatrick
  et~al.}{1983}]{KirkppatrickSimulatedAnnealing}
\begin{barticle}
\bauthor{\bsnm{Kirkpatrick}, \binits{S.}},
\bauthor{\bsnm{Gelatt}, \binits{C.D.}},
\bauthor{\bsnm{Vecchi}, \binits{M.P.}}:
\batitle{Optimization by simulated annealing}.
\bjtitle{Science}
\bvolume{220}(\bissue{4598}),
\bfpage{671}--\blpage{680}
(\byear{1983})
\doiurl{10.1126/science.220.4598.671}
\end{barticle}
\endbibitem

\bibitem[\protect\citeauthoryear{{D-Wave Systems}}{2025}]{DWaveSampler}
\begin{botherref}
\oauthor{\bsnm{{D-Wave Systems}}}:
DWaveSamplers.
\url{https://docs.dwavequantum.com/en/latest/ocean/api\_ref\_samplers/index.html}.
Accessed: 7 November 2025
(2025)
\end{botherref}
\endbibitem

\bibitem[\protect\citeauthoryear{Glover}{1986}]{Glover:TabuSearch}
\begin{barticle}
\bauthor{\bsnm{Glover}, \binits{F.}}:
\batitle{Future paths for integer programming and links to artificial
  intelligence}.
\bjtitle{Computers \& Operations Research}
\bvolume{13}(\bissue{5}),
\bfpage{533}--\blpage{549}
(\byear{1986})
\doiurl{10.1016/0305-0548(86)90048-1} .
\bcomment{Applications of Integer Programming}
\end{barticle}
\endbibitem

\bibitem[\protect\citeauthoryear{{Gurobi Optimization, LLC}}{2025}]{Gurobi}
\begin{botherref}
\oauthor{\bsnm{{Gurobi Optimization, LLC}}}:
{Gurobi Optimizer Reference Manual}
(2025).
\url{https://www.gurobi.com}
\end{botherref}
\endbibitem

\bibitem[\protect\citeauthoryear{Seidel et~al.}{2024}]{qrisp}
\begin{botherref}
\oauthor{\bsnm{Seidel}, \binits{R.}},
\oauthor{\bsnm{Bock}, \binits{S.}},
\oauthor{\bsnm{Zander}, \binits{R.}},
\oauthor{\bsnm{Petrič}, \binits{M.}},
\oauthor{\bsnm{Steinmann}, \binits{N.}},
\oauthor{\bsnm{Tcholtchev}, \binits{N.}},
\oauthor{\bsnm{Hauswirth}, \binits{M.}}:
Qrisp: A Framework for Compilable High-Level Programming of Gate-Based Quantum
  Computers
(2024).
\url{https://arxiv.org/abs/2406.14792}
\end{botherref}
\endbibitem

\bibitem[\protect\citeauthoryear{Powell}{1994}]{Powell:COBYLA}
\begin{bbook}
\bauthor{\bsnm{Powell}, \binits{M.J.D.}}:
In: \beditor{\bsnm{Gomez}, \binits{S.}},
\beditor{\bsnm{Hennart}, \binits{J.-P.}} (eds.)
\bbtitle{A Direct Search Optimization Method That Models the Objective and
  Constraint Functions by Linear Interpolation},
pp. \bfpage{51}--\blpage{67}.
\bpublisher{Springer},
\blocation{Dordrecht}
(\byear{1994}).
\doiurl{10.1007/978-94-015-8330-5\_4}
\end{bbook}
\endbibitem

\bibitem[\protect\citeauthoryear{Virtanen et~al.}{2020}]{SciPy}
\begin{barticle}
\bauthor{\bsnm{Virtanen}, \binits{P.}},
\bauthor{\bsnm{Gommers}, \binits{R.}},
\bauthor{\bsnm{Oliphant}, \binits{T.E.}},
\bauthor{\bsnm{Haberland}, \binits{M.}},
\bauthor{\bsnm{Reddy}, \binits{T.}},
\bauthor{\bsnm{Cournapeau}, \binits{D.}},
\bauthor{\bsnm{Burovski}, \binits{E.}},
\bauthor{\bsnm{Peterson}, \binits{P.}},
\bauthor{\bsnm{Weckesser}, \binits{W.}},
\bauthor{\bsnm{Bright}, \binits{J.}},
\bauthor{\bsnm{Walt}, \binits{S.J.}},
\bauthor{\bsnm{Brett}, \binits{M.}},
\bauthor{\bsnm{Wilson}, \binits{J.}},
\bauthor{\bsnm{Millman}, \binits{K.J.}},
\bauthor{\bsnm{Mayorov}, \binits{N.}},
\bauthor{\bsnm{Nelson}, \binits{A.R.J.}},
\bauthor{\bsnm{Jones}, \binits{E.}},
\bauthor{\bsnm{Kern}, \binits{R.}},
\bauthor{\bsnm{Larson}, \binits{E.}},
\bauthor{\bsnm{Carey}, \binits{C.J.}},
\bauthor{\bsnm{Polat}, \binits{{\. {I}}.}},
\bauthor{\bsnm{Feng}, \binits{Y.}},
\bauthor{\bsnm{Moore}, \binits{E.W.}},
\bauthor{\bsnm{VanderPlas}, \binits{J.}},
\bauthor{\bsnm{Laxalde}, \binits{D.}},
\bauthor{\bsnm{Perktold}, \binits{J.}},
\bauthor{\bsnm{Cimrman}, \binits{R.}},
\bauthor{\bsnm{Henriksen}, \binits{I.}},
\bauthor{\bsnm{Quintero}, \binits{E.A.}},
\bauthor{\bsnm{Harris}, \binits{C.R.}},
\bauthor{\bsnm{Archibald}, \binits{A.M.}},
\bauthor{\bsnm{Ribeiro}, \binits{A.H.}},
\bauthor{\bsnm{Pedregosa}, \binits{F.}},
\bauthor{\bsnm{Mulbregt}, \binits{P.}},
\bauthor{\bsnm{Vijaykumar}, \binits{A.}},
\bauthor{\bsnm{Bardelli}, \binits{A.P.}},
\bauthor{\bsnm{Rothberg}, \binits{A.}},
\bauthor{\bsnm{Hilboll}, \binits{A.}},
\bauthor{\bsnm{Kloeckner}, \binits{A.}},
\bauthor{\bsnm{Scopatz}, \binits{A.}},
\bauthor{\bsnm{Lee}, \binits{A.}},
\bauthor{\bsnm{Rokem}, \binits{A.}},
\bauthor{\bsnm{Woods}, \binits{C.N.}},
\bauthor{\bsnm{Fulton}, \binits{C.}},
\bauthor{\bsnm{Masson}, \binits{C.}},
\bauthor{\bsnm{H{\"a}ggstr{\"o}m}, \binits{C.}},
\bauthor{\bsnm{Fitzgerald}, \binits{C.}},
\bauthor{\bsnm{Nicholson}, \binits{D.A.}},
\bauthor{\bsnm{Hagen}, \binits{D.R.}},
\bauthor{\bsnm{Pasechnik}, \binits{D.V.}},
\bauthor{\bsnm{Olivetti}, \binits{E.}},
\bauthor{\bsnm{Martin}, \binits{E.}},
\bauthor{\bsnm{Wieser}, \binits{E.}},
\bauthor{\bsnm{Silva}, \binits{F.}},
\bauthor{\bsnm{Lenders}, \binits{F.}},
\bauthor{\bsnm{Wilhelm}, \binits{F.}},
\bauthor{\bsnm{Young}, \binits{G.}},
\bauthor{\bsnm{Price}, \binits{G.A.}},
\bauthor{\bsnm{Ingold}, \binits{G.-L.}},
\bauthor{\bsnm{Allen}, \binits{G.E.}},
\bauthor{\bsnm{Lee}, \binits{G.R.}},
\bauthor{\bsnm{Audren}, \binits{H.}},
\bauthor{\bsnm{Probst}, \binits{I.}},
\bauthor{\bsnm{Dietrich}, \binits{J.P.}},
\bauthor{\bsnm{Silterra}, \binits{J.}},
\bauthor{\bsnm{Webber}, \binits{J.T.}},
\bauthor{\bsnm{Slavi{\v{c}}}, \binits{J.}},
\bauthor{\bsnm{Nothman}, \binits{J.}},
\bauthor{\bsnm{Buchner}, \binits{J.}},
\bauthor{\bsnm{Kulick}, \binits{J.}},
\bauthor{\bsnm{Sch{\"o}nberger}, \binits{J.L.}},
\bauthor{\bsnm{Miranda~Cardoso}, \binits{J.V.}},
\bauthor{\bsnm{Reimer}, \binits{J.}},
\bauthor{\bsnm{Harrington}, \binits{J.}},
\bauthor{\bsnm{Rodr{\'i}guez}, \binits{J.L.C.}},
\bauthor{\bsnm{Nunez-Iglesias}, \binits{J.}},
\bauthor{\bsnm{Kuczynski}, \binits{J.}},
\bauthor{\bsnm{Tritz}, \binits{K.}},
\bauthor{\bsnm{Thoma}, \binits{M.}},
\bauthor{\bsnm{Newville}, \binits{M.}},
\bauthor{\bsnm{K{\"u}mmerer}, \binits{M.}},
\bauthor{\bsnm{Bolingbroke}, \binits{M.}},
\bauthor{\bsnm{Tartre}, \binits{M.}},
\bauthor{\bsnm{Pak}, \binits{M.}},
\bauthor{\bsnm{Smith}, \binits{N.J.}},
\bauthor{\bsnm{Nowaczyk}, \binits{N.}},
\bauthor{\bsnm{Shebanov}, \binits{N.}},
\bauthor{\bsnm{Pavlyk}, \binits{O.}},
\bauthor{\bsnm{Brodtkorb}, \binits{P.A.}},
\bauthor{\bsnm{Lee}, \binits{P.}},
\bauthor{\bsnm{McGibbon}, \binits{R.T.}},
\bauthor{\bsnm{Feldbauer}, \binits{R.}},
\bauthor{\bsnm{Lewis}, \binits{S.}},
\bauthor{\bsnm{Tygier}, \binits{S.}},
\bauthor{\bsnm{Sievert}, \binits{S.}},
\bauthor{\bsnm{Vigna}, \binits{S.}},
\bauthor{\bsnm{Peterson}, \binits{S.}},
\bauthor{\bsnm{More}, \binits{S.}},
\bauthor{\bsnm{Pudlik}, \binits{T.}},
\bauthor{\bsnm{Oshima}, \binits{T.}},
\bauthor{\bsnm{Pingel}, \binits{T.J.}},
\bauthor{\bsnm{Robitaille}, \binits{T.P.}},
\bauthor{\bsnm{Spura}, \binits{T.}},
\bauthor{\bsnm{Jones}, \binits{T.R.}},
\bauthor{\bsnm{Cera}, \binits{T.}},
\bauthor{\bsnm{Leslie}, \binits{T.}},
\bauthor{\bsnm{Zito}, \binits{T.}},
\bauthor{\bsnm{Krauss}, \binits{T.}},
\bauthor{\bsnm{Upadhyay}, \binits{U.}},
\bauthor{\bsnm{Halchenko}, \binits{Y.O.}},
\bauthor{\bsnm{V{\'a}zquez-Baeza}, \binits{Y.}},
\bauthor{\bsnm{1.0~Contributors}, \binits{S.}}:
\batitle{Scipy 1.0: fundamental algorithms for scientific computing in python}.
\bjtitle{Nature Methods}
\bvolume{17}(\bissue{3}),
\bfpage{261}--\blpage{272}
(\byear{2020})
\doiurl{10.1038/s41592-019-0686-2}
\end{barticle}
\endbibitem

\bibitem[\protect\citeauthoryear{Zhang}{2023}]{Zhang_2023}
\begin{botherref}
\oauthor{\bsnm{Zhang}, \binits{Z.}}:
{PRIMA: Reference Implementation for Powell's Methods with Modernization and
  Amelioration}
(2023).
\doiurl{10.5281/zenodo.8052654} .
\url{http://www.libprima.net}
\end{botherref}
\endbibitem

\bibitem[\protect\citeauthoryear{Zhu et~al.}{1997}]{LBFGSB1}
\begin{barticle}
\bauthor{\bsnm{Zhu}, \binits{C.}},
\bauthor{\bsnm{Byrd}, \binits{R.H.}},
\bauthor{\bsnm{Lu}, \binits{P.}},
\bauthor{\bsnm{Nocedal}, \binits{J.}}:
\batitle{Algorithm 778: L-bfgs-b: Fortran subroutines for large-scale
  bound-constrained optimization}.
\bjtitle{ACM Trans. Math. Softw.}
\bvolume{23}(\bissue{4}),
\bfpage{550}--\blpage{560}
(\byear{1997})
\end{barticle}
\endbibitem

\bibitem[\protect\citeauthoryear{Byrd et~al.}{1995}]{LBFGSB2}
\begin{barticle}
\bauthor{\bsnm{Byrd}, \binits{R.H.}},
\bauthor{\bsnm{Lu}, \binits{P.}},
\bauthor{\bsnm{Nocedal}, \binits{J.}},
\bauthor{\bsnm{Zhu}, \binits{C.}}:
\batitle{A limited memory algorithm for bound constrained optimization}.
\bjtitle{SIAM Journal on Scientific Computing}
\bvolume{16}(\bissue{5}),
\bfpage{1190}--\blpage{1208}
(\byear{1995})
\doiurl{10.1137/0916069}
\end{barticle}
\endbibitem

\bibitem[\protect\citeauthoryear{Sagnol and Stahlberg}{2022}]{PICOS}
\begin{barticle}
\bauthor{\bsnm{Sagnol}, \binits{G.}},
\bauthor{\bsnm{Stahlberg}, \binits{M.}}:
\batitle{{PICOS}: A {Python} interface to conic optimization solvers}.
\bjtitle{Journal of Open Source Software}
\bvolume{7}(\bissue{70}),
\bfpage{3915}
(\byear{2022})
\doiurl{10.21105/joss.03915}
\end{barticle}
\endbibitem

\bibitem[\protect\citeauthoryear{{MOSEK ApS}}{2025}]{mosek}
\begin{botherref}
\oauthor{\bsnm{{MOSEK ApS}}}:
MOSEK.
\url{https://www.mosek.com/}
\end{botherref}
\endbibitem

\bibitem[\protect\citeauthoryear{Wang and
  K{\i}l{\i}n{\c{c}}-Karzan}{2022}]{QCQP1}
\begin{barticle}
\bauthor{\bsnm{Wang}, \binits{A.L.}},
\bauthor{\bsnm{K{\i}l{\i}n{\c{c}}-Karzan}, \binits{F.}}:
\batitle{On the tightness of sdp relaxations of qcqps}.
\bjtitle{Mathematical Programming}
\bvolume{193}(\bissue{1}),
\bfpage{33}--\blpage{73}
(\byear{2022})
\doiurl{10.1007/s10107-020-01589-9}
\end{barticle}
\endbibitem

\bibitem[\protect\citeauthoryear{Shor}{1990}]{QCQPShor}
\begin{barticle}
\bauthor{\bsnm{Shor}, \binits{N.Z.}}:
\batitle{Dual quadratic estimates in polynomial and boolean programming}.
\bjtitle{Annals of Operations Research}
\bvolume{25}(\bissue{1}),
\bfpage{163}--\blpage{168}
(\byear{1990})
\doiurl{10.1007/BF02283692}
\end{barticle}
\endbibitem

\end{thebibliography}

\end{document}